\documentclass{article}
\usepackage[margin=1.3in]{geometry}

\usepackage{cite}
\usepackage{amsthm}
\usepackage{dsfont}
\usepackage{amsmath}
\usepackage{mathtools}
\usepackage{amssymb}
\usepackage{amsmath,amssymb,amsfonts}
\usepackage{algorithmic}
\usepackage{graphicx}
\usepackage{textcomp}
\usepackage{xcolor}
\usepackage{enumerate}
\usepackage{moresize}
\usepackage{bm}
\newcommand{\indicator}[1]{\boldsymbol{1}_{  \{  #1 \}  }    }

\newcommand{\indc}{\mbox{\large {${\mathds 1}$}}}

\newtheorem{thm}{{\bf Theorem}}
\newtheorem{cor}{Corollary}
\newtheorem{lemma}{{\bf Lemma}}

\newcommand{\kcmnts}[1] {\textcolor{red}{kav: #1}}
\newcommand{\kcmnt}[1]{} 

\newcommand{\kt}[1]{{#1}}  

\newcommand{\eop}{\hfill{$\blacksquare$}}

\newcommand{\mue}{{\mu_\epsilon}}
\newcommand{\lame}{{\lambda\epsilon}}

\newcommand{\eqdist}{\stackrel{d}{=}}

\newcommand{\ra}{\rightarrow} 
\newcommand{\ua}{\uparrow}

\newcommand{\prob}[1]{P\left(#1\right)}

\newcommand{\extrabits}[1]{} 
\newcommand{\ignore}[1]{}
\renewcommand{\vspace}[1]{} 
\newcommand{\g}{ {\bf g}}

\newcommand{\rmin }{{\underline \rho}}
\newcommand{\rmax }{{\bar \rho}}

\newcommand{\dmin }{{\underline d}}
\newcommand{\dmax }{{\bar d}}

\renewcommand{\i} {{ \epsilon }}
\renewcommand{\t}{{\tau}}
\usepackage[normalem]{ulem}

\usepackage{amsmath,amssymb,amsthm}
\usepackage{xspace}

\newcommand{\convas}{\stackrel{a.s.}{\to}}

\newcommand{\beq}{\begin{eqnarray*}}
\newcommand{\eeq}{\end{eqnarray*}}

\newcommand{\Exp}[1]{E\left[#1\right]} 

\renewcommand{\i} {{ \epsilon }}
\renewcommand{\t}{{\tau}}
\newcommand{\mut}{{\mu_\tau}}
\newcommand{\B}{\mathcal{B}}
\renewcommand{\S}{\mathcal{S}}

\newcommand{\N}{\mathbb{N}}

\newboolean{wiopt}
\setboolean{wiopt}{false}
\newcommand{\conftext}[1]{\ifthenelse{\boolean{wiopt}}{#1}{}}
\newcommand{\TRtext}[1]{\ifthenelse{\boolean{wiopt}}{}{{\color{red}#1}}}
\newcommand{\revision}[1]{#1}

\newboolean{showcomments}
\setboolean{showcomments}{true}
\newcommand{\jk}[1]{\ifthenelse{\boolean{showcomments}} {\textcolor{green}{(JK says: #1)}} {}}
\newcommand{\kc}[1]{\ifthenelse{\boolean{showcomments}} {\textcolor{blue}{(KC says: #1)}} {}}
\newcommand{\vk}[1]{\ifthenelse{\boolean{showcomments}} {\textcolor{blue}{(VK says: #1)}} {}}

\newcommand{\h}{h}

\def\BibTeX{{\rm B\kern-.05em{\sc i\kern-.025em b}\kern-.08em
    T\kern-.1667em\lower.7ex\hbox{E}\kern-.125emX}}
\begin{document}

\title{Dynamic scheduling in a partially fluid, partially lossy
  queueing system\thanks{A preliminary version of this work was
    presented at the WiOpt conference, 2019 \cite{Chaudhary2019}.}}

\author{Kiran Chaudhary\thanks{Kiran Chaudhary and Veeraruna Kavitha
    are with the Industrial Engineering and Operations Research (IEOR)
    Department at IIT Bombay.}  \and Veeraruna Kavitha\footnotemark[2]
  \and Jayakrishnan~Nair\thanks{Jayakrishnan Nair is with the
    Electrical Engineering Department at IIT Bombay; he acknowledges
    support from DST and CEFIPRA.}}%

\maketitle

\begin{abstract}

  We consider a single server queueing system with two classes of
  jobs: {\it eager} jobs with small sizes that require service to
  begin almost immediately upon arrival, and {\it tolerant} jobs with
  larger sizes that can wait for service. While blocking probability
  is the relevant performance metric for the eager class, the tolerant
  class seeks to minimize its mean sojourn time. In this paper, we
  analyse the performance of each class under dynamic scheduling
  policies, where the scheduling of both classes depends on the
  instantaneous state of the system. This analysis is carried out
  under a certain fluid limit, where the arrival rate and service rate
  of the eager class are scaled to infinity, holding the offered load
  constant. Our performance characterizations reveal a {\it (dynamic)
    pseudo-conservation law} that ties the performance of both the
  classes to the standalone blocking probabilities associated with the
  scheduling policies for the eager class.
  Further, the performance is robust to other specifics of the
  scheduling policies. We also characterize the Pareto frontier of the
  achievable region of performance vectors under the same fluid limit,
  and identify a (two-parameter) class of {\it Pareto-complete}
  scheduling policies.
\end{abstract}

\section{Introduction}
\label{sec:intro}

\ignore{Modern cellular systems handle two distinct classes of traffic
  --- voice and data. Voice calls must be either admitted or dropped
  upon arrival. Thus, the relevant performance metric for voice calls
  is the probability of a call drop. On the other hand, data traffic
  is elastic and can be queued. As a result, the performance of data
  traffic is best captured by the delay statistics. Given that voice
  and data calls share the same underlying service capacity (wireless
  spectrum), the admission control and scheduling in a cellular system
  has to be designed to optimally balance the performance experienced
  by these two (heterogeneous) classes.}

In this paper, we analyse a single server queueing system with two
heterogeneous customer classes. One class of customers is {\it
  eager}---they require service to commence (almost) immediately upon
arrival.
The performance of the eager class is captured by the blocking
probability, i.e., the long run fraction of eager customers that are
blocked. The second class of customers is {\it tolerant}---these
customers can tolerate delays and may be queued. The performance of
this class is captured via the mean response time of the tolerant
customers.

Service systems of this kind are motivated by modern cellular
networks, which handle voice calls (which must be either admitted or
dropped upon arrival) as well as data traffic (which can be
queued).  \revision{ Another motivation comes from  super-markets, where it is
 common practice to provide  prioritized service to customers with fewer items via dedicated express
counters; these customers have limited patience, and may balk/renege if service is not provided  almost  immediately.}  
However, such part-loss, part-queueing multi-class service
systems are analytically intractable even under the simplest
scheduling disciplines (see \cite{Mahabhashyam05} and the references
therein). In this paper, we derive tractable approximations of the
performance experienced by each class using a certain fluid limit,
referred to as the short-frequent-jobs (SFJ) limit.

The SFJ limit corresponds to scaling the arrival rate as well as the
service rate of the eager class to infinity, such that the offered
load is held constant. This gives rise to a time-scale separation
between the two classes, the eager class operating at a faster
time-scale. Under the SFJ limit, we obtain a closed form
characterization of the performance of both classes under a broad
class of dynamic policies that allow the admission control of the
eager class and the scheduling of both classes to be dependent on the
current state of the system.\footnote{From here on, we follow the
  convention that admission control (if used) is included in the eager
  scheduling policy.\label{footnote:eps-scheduling}} Interestingly, a
dynamic pseudo-conservation law follows from this characterization (in
the SJF limit)---the performance of both classes depends only on the
standalone blocking probabilities (resulting when a single eager
scheduling policy is used oblivious to the tolerant state) associated
with the eager scheduling schemes employed for each occupancy level of
the tolerant queue. In particular, the performance does not depend on
the specific scheduling policies that generate those blocking
probabilities, as well as on the details of the tolerant scheduler
(subject to work conservation and serial, non-anticipative
processing). Conservation laws typically allow one to compute the
performance of a complex system in terms of the performance of simpler
ones. In our case, once the relevant standalone blocking probabilities
are known (these can usually be computed easily as they result from
the analysis of a single-class loss system), one can compute the
performance of both the classes.

%
%
\ignore{Under the dynamic policies considered in our paper, the
  performance of both the classes are inter-dependent.  Consequently,
  along with seeing the time varying service process, the tolerant
  class experience the service process which depends on the number of
  customers in the tolerant system.  As the tolerant customer receive
  the time and state varying service process the analysis of the
  tolerant class becomes more difficult than that of in the case with
  static policy. However, we showed that under the partial fluid
  regime, the analysis of both the classes is tractable. Considering
  the dynamic policies clearly expands the achievable region of the
  system (i.e., the set of achievable performance vectors across all
  policies under consideration) relative to the restriction to static
  policies.

  This clearly expands the achievable region of the system (i.e., the
  set of achievable performance vectors across all policies under
  consideration) relative to the restriction to static policies.
}
 
We further analyse the Pareto frontier of performance vectors
achievable under the class of dynamic schedulers, which defines the
set of efficient operating points for the system. Remarkably, we are
able to identify a {\it Pareto-complete} family of scheduling policies
(we call a family of schedulers Pareto-complete if it spans the entire
Pareto frontier over its parameter space). This family, parametrized
by $(L,d),$ where $L \in \mathbb{N}$ and $d \in (0,1),$ blocks eager
customers with the minimum blocking probability when the tolerant
occupancy is less than~$L,$ with probability $d$ when the occupancy
equals $L,$ and with the maximum blocking probability when it
exceeds~$L.$

Finally, via numerical experiments, we show that our performance
characterizations under the SFJ limit are extremely accurate in the
pre-limit (i.e., for moderate values of arrival and service rates of
the eager class). This shows that our approximations, which are
provably accurate under the SFJ fluid limit, are also applicable in
practice.

The remainder of this paper is organised as follows. After a review of
the related literature below, we describe our system model and state
some preliminary results in Section~\ref{sec:model}. We also give some
examples of dynamic schedulers in Section~\ref{sec:model}. Under the
SFJ limit, we characterize the performance of the tolerant class and
the eager class in Section~\ref{sec:tolerant}. We formally define the
dynamic achievable region in Section~\ref{sec:Achievable Region}, and
demonstrate the Pareto-complete family of dynamic schedulers in
Section~\ref{sec:pareto}. We conclude the paper in
Section~\ref{sec:conclusion}. A summary of our notation can be found
in   Table~\ref{Table_notations} in Appendix~\ref{Appendix_sample}.

\subsection*{Related Literature}

\ignore{A study has been done for a similar model in \cite{ANOR} but
  with static policies, in which the admission control of eager class
  does not depend upon the number of customers in the tolerant
  class. . In \cite{ANOR}, a pseudo-conservation law has been
  established which relates the blocking the blocking probability of
  the eager customers and the performance experienced by the tolerant
  class. These conservation laws are derived under a partial fluid
  scaling, in which the arrival and service rate of eager customers
  are scaled to infinity, such that the offered load of this class
  remains constant. This scaling regime is also characterised as the
  SFJ scaling.  The conservation laws derived in \cite{ANOR} implies
  that, under the partial scaling regime mentioned above, the
  performance experienced by the tolerant class depends only on the
  blocking probability of the eager class and not on the admission
  control and scheduling policy that produce that blocking
  probability.} 
The present paper is a follow-up of our prior work \cite{Value,ANOR},
which analyses the same heterogeneous queueing system under the SFJ
limit for a class of (partially) \emph{static} scheduling
policies. Under this class of policies, the scheduling of the eager
class is oblivious to the state of the tolerant queue, with the
tolerant queue simply utilizing the service capacity left unused by
the eager class. Clearly, this class of schedulers is restrictive. In
the present paper, we consider general \emph{dynamic} policies, where
eager scheduling depends on the occupancy of the tolerant queue. This
generalization, which requires a non-trivial analysis, results is a
substantial expansion of the achievable region of feasible performance
vectors (as is shown in Sections~\ref{sec:Achievable Region} \&
\ref{sec:pareto}). Moreover, the generalization to dynamic policies
necessitates the identification of a Pareto-complete family of
schedulers (which is the goal of Section~\ref{sec:pareto}); within the
restricted class of static schedulers analysed in \cite{Value,ANOR},
it turns out that all policies are efficient.

Aside from \cite{ANOR,Value}, the only prior work we are aware of that
analyses a part-queueing, part-loss service system is
\cite{CD_close}. In this paper, the authors obtain the performance
metrics for all classes in closed form, assuming exponential
inter-arrival and service times for all classes, under a certain {\it
  static} priority scheduling discipline. However, we note that
\cite{CD_close} does not attempt to address the \emph{tradeoff}
between the performance of the two classes, which is central to the
present work.
 
From an application standpoint, this paper is also related to the
considerable literature on sharing the capacity of a cellular system
between voice and data traffic; for example,
see~\cite{Li04,Tang,Zhang}. In this line of work, both voice and data
classes are treated as lossy, the focus being on characterizing the
blocking probability of each class under different (static and
dynamic) admission rules. However, to the best of our knowledge, these
papers do not analyse the achievable region of performance vectors, or
characterize its Pareto frontier.

We also note that there is a well-developed literature on multiclass
queueing systems with multiple tolerant classes on a single server
(e.g., conservation laws, pioneered by \cite{Kleinrock1965}). The
achievable region is well understood in such a `homogeneous'
multi-class setting
\cite{coffman,shanthikumar1992multiclass}. Interestingly, in this
case, it is known that the static and dynamic achievable regions
coincide (see \cite{Value}), in contrast with the `heterogeneous'
multi-class setting considered here, where we see that the static
achievable region is a strict subset of the dynamic achievable region.
Moreover, the achievable region in the homogeneous setting is its own
Pareto frontier (i.e., all points of the achievable region are
efficient) under work conserving policies, also in contrast with the
heterogeneous setting considered here.

This paper is an extended version of \cite{Chaudhary2019}, which
analyses a restricted class of dynamic scheduling policies, wherein
the number of distinct eager sub-policies is assumed to be finite. In
this paper, we extend the analysis to general dynamic policies, and
further provide complete proofs of all results.


\section{System Description}
\label{sec:model}

We consider a single server queueing system with two job (a.k.a.
customer) classes: eager customers (also denoted as
$\epsilon$-customers) demand service immediately upon arrival, whereas
tolerant customers (also denoted as $\tau$-customers) can wait in a
queue (of infinite capacity) to be served.\footnote{\revision{We
    discuss the generalization where eager customers have limited
    patience and abandon/renege after a short wait time in
    Section~\ref{sec:conclusion}; see also~\cite{ANOR}.}} The
$\tau$-customers can be interrupted either partially (i.e., their
service rate may be reduced) or completely by $\i$-customers, but not
by other $\tau$-customers.  Without loss of generality, we assume {\it
  a unit server speed.}
We assume that $\epsilon$-customers (respectively, $\t$-customers)
arrive according to a Poisson process with rate $\lambda_{\i}$
(respectively, $\lambda_\t$). The sequence of job sizes
(a.k.a. service requirements) for both the classes is i.i.d., with
$B_\i$ denoting a generic $\i$ job size, and $B_\t$ denoting a generic
$\t$ job size. Throughout, {\it we assume that $B_\t$ is exponentially
  distributed with mean $1/\mu_\t,$ and that $\Exp{B_\i^2} < \infty.$}
Let $\mu_\i := 1/\Exp{B_\i}.$

\subsection{Dynamic schedulers}

We consider dynamic scheduling, wherein the scheduling policy for each
class depends upon the state (occupancy) of both the classes.  An
eager scheduling (sub)policy (a set of rules that assign fractions of
server capacity to various eager customers, which may also include the
admission control rules) is chosen depending upon the tolerant
state. The overall dynamic scheduler is characterized by a sequence of
such eager sub-policies, one corresponding to each value of tolerant
occupancy---the tolerant queue in turn utilizes the service capacity
left unused by the eager class in a work-conserving manner. Our
dynamic schedulers can thus be viewed to be of nested type: a
top-level policy chooses the sub-policy used for scheduling the
$\i$-class based on the occupancy (state) of the
$\t$-class. \revision{The sub-policy in turn determines, based on the
  number of eager customers in the system, the admission rule for
  subsequent arrivals, and the service rate to allocate to each
  existing eager job. It is important to note that these nested
  schedulers are not restrictive; rather this is simply a convenient
  representation of dynamic schedulers; details and examples follow.}
As a result, the service processes of the two classes are
interdependent (unlike in the case of static scheduling as considered
in \cite{ANOR,Value}).

\revision{Specifically, let $(X_\t (t), X_\i (t))$ represent the
  number of tolerant and eager customers in the system at time~$t$. We
  consider stationary Markov scheduling policies that determine the
  service plan for both classes (including, possibly, admission
  control of eager customers), depending on the tuple
  $(X_\t (t), X_\i (t)).$
  Such Markov policies are known to be sufficient for a wide range of
  sequential decision problems (see \cite{Puter}); in the context of
  our model, they are sufficient if one additionally assumes that
  eager customers have exponentially distributed service requirements.
  As an example, consider the following policy parameterized by
  $(L, N) \in \mathbb{N}^2$: a) when $(X_\t(t), X_\i (t)) = (j, c)$
  with $c \le N$, and $j \le L$ then each $\i$-customer is served with
  rate (speed) $1/N$, while longest waiting $\t$-customer is served at
  rate $(N-c)/N$; and b) if $j > L$, the longest waiting $\t$-customer
  is served with the full service capacity (i.e., at rate 1), and no
  further $\i$-customers are admitted. It is convenient to view such
  policies as being \emph{nested}; for the above example, when
  $X_\t (t) \le L$ the eager class is served according to one
  sub-policy, while it is served according to a second sub-policy when
  tolerant occupancy is more than $L$. Under the former sub-policy,
  the eager class is served as an $M/G/N/N$ queue (with each `server'
  taken to have speed~$1/N$), while under the latter sub-policy, no
  eager customers are admitted. There is of course the issue of how to
  deal with the existing eager customers in the system when there is
  an arrival/departure in the $\t$ queue, leading to change in the
  eager sub-policy. In the above example, this corresponds to the
  handling of any eager customers that remain when the $\t$ occupancy
  increases from $L$ to $L+1.$ This issue is addressed later in this
  section (see Assumption {\bf A}.1 and the discussion in
  Section~\ref{sec:example_models}).}

Note that the eager sub-policy switches whenever there is an
arrival/departure in the tolerant queue, though the further details of
each sub-policy depend only on the eager occupancy.
We typically refer to the $\i$-sub-policy implemented when there are
$j$ tolerant jobs in the system as sub-policy~$j.$

\noindent {\bf $\i$-schedulers:} Note that while the occupancy of the
tolerant queue dictates the selection of the eager sub-policy, the
sub-policies themselves are oblivious to the state of the tolerant
queue.
%
We make the following additional assumptions.  Some examples of
schedulers that satisfy these are presented in
Section~\ref{sec:example_models}:
\begin{enumerate}[{\bf A}.1]
\item To simplify the transition from one $\i$-sub-policy to the next,
  we assume that all $\i$-customers are dropped when there is an
  arrival/departure in the $\t$-queue.
\item The scheduling of each sub-policy depends only on the number of
  $\i$-jobs present in the system.
\item On arrival, an eager job is either admitted into service or gets
  dropped/blocked. If admitted, an eager job receives service at a
  minimum rate/speed of $c_{\min} > 0$ until its departure \emph{under
    any sub-policy.}
\end{enumerate}

Assumption {\bf A}.1 is a technical condition required in our proofs.
In Section~\ref{sec:pareto}, we show that this assumption has
negligible influence on the performance of either class under our
partial fluid limit.\footnote{ Under the SFJ limit, this `flushing' of
  the $\i$-system is only performed at a bounded rate (since
  arrivals/departures in the $\t$-queue occur at a bounded rate),
  while the arrival rate of the $\i$ system scales to infinity. Thus,
  we expect that this assumption will not impact the blocking
  probability of the eager class (this is also evident from the Monte
  Carlo simulation based study presented in
  Section~\ref{sec:pareto}).} \revision{We are also able to relax
  Assumption {\bf A}.1 for the special case of exponentially
  distributed eager jobs (see Theorem~\ref{Thm_withou_A1}). As
  discussed before, Assumption {\bf A}.2 is not restrictive; it allows
  all (non-size based) Markov policies.
  We provide examples of such schedulers in
  Section~\ref{sec:example_models}.}
Assumption~{\bf A}.3 implies that the eager class operates as a
\emph{loss system}. Indeed, the requirement of a minimum service
rate/speed $c_{\min} > 0$ captures the `eagerness' of eager
customers. This assumption also implies that there exists a uniform
upper bound $\mathcal{K}$ on the number of $\i$-jobs in the system at
any time \emph{across sub-policies} (since the server is assumed to
have a unit speed).

The above assumptions imply the following condition (proved as
Lemma~\ref{lemma:bp_upper_bound} in Appendix~\ref{sec:busycycles});
the same is stated as a (redundant) assumption to emphasize that the
uniform convergence required in our analysis is predominantly due to
the following uniform bound on the eager busy cycles:
\begin{enumerate}[{\bf A}.4]
\item Across sub-policies, the second moment of the $\i$-busy cycle,
  defined as the interval between the start of two successive $\i$
  busy periods, is uniformly bounded.
\end{enumerate}
%

%


\noindent {\bf $\t$-schedulers:} Next, we state our assumptions on the
scheduling policy of the tolerant class.
\begin{enumerate}[{\bf B}.1]
\item The $\tau$-scheduler is work conserving, i.e., it utilizes all
  the service capacity left unused by $\epsilon$-jobs, so long as the
  $\tau$-queue is non-empty.
\item The $\tau$-jobs are served in a serial fashion, i.e.,
  $\tau$-jobs cannot pre-empt one another.
\item The $\tau$-scheduler is blind to the size of $\tau$-jobs.
\end{enumerate}

Assumption~{\bf B}.1 implies that the tolerant class experiences a
time varying service process, which depends on both the $\tau$-state
as well as the $\epsilon$-state. Assumptions~\textbf{B}.2-3 imply that
we consider $\tau$-schedulers which are non-pre-emptive and
non-anticipative, for instance, {\it first come first served} (FCFS),
{\it last come first served} (LCFS), and {\it random order of service}
(see
\cite[Chapter~29]{Harchol-Balter}). 

We require another assumption ({\bf B}.4) regarding the stability of
the $\t$-queue under the SFJ limit, when $\i$-customers employ a
single sub-policy (irrespective of the $\t$-state).
We provide the required background on these schedulers, define
formally the SFJ scaling, and state the resulting pseudo-conservation
law (see \cite{ANOR} for more details), after which we state
Assumption~{\bf B}.4.

\subsection{$\t$-static schedulers and background }
\label{sec:background}

We now consider the special case of $\t$-static scheduling, where a
single $\i$-sub-policy is used at all times (irrespective of
$\t$-state); this case was analysed in~\cite{ANOR}. Let ${P_B}_j$
represent the blocking probability \kt{(long run fraction of losses)}
of the $\i$-class, if sub-policy-$ j$ is used in a $\t$-static manner
(${P_B}_j$ was referred to as the standalone blocking probability of
sub-policy~$j$ in Section~\ref{sec:intro}); we also refer to these as
$\t$-static blocking probabilities.

\noindent {\bf Short-Frequent Jobs (SFJ) Scaling:} Under the SFJ
scaling (as in \cite{ANOR}), we let $\lambda_\i \ra \infty$ and
$\mu_\i \ra \infty,$ such that $\rho_\i := {\lambda_\i}/{\mu_\i}$
remains constant. This corresponds to scaling the arrival as well as
the service rate of the eager class to infinity proportionately, so
that the offered load (the long term rate at which work arrives into
the system) is held constant. We use $\mu_\i$ as the scale parameter
for this partial scaling.
Specifically, we scale the job size distribution of the eager class
as, $B^{\mu_\i}_\i \eqdist B^1_\i/\mu_{\i}$, where $B^{\mu_\i}_\i$
denotes a generic eager job size at scale $\mu_\i$ and $\eqdist$ is
equality in distribution. This scaling (plus Poisson arrivals) under
{\bf A}.2 ensures that at scale $\mue,$ the occupancy process of the
$\i$-class gets time-scaled (fast-forwarded) by~$\mue;$ the details
can be found in Appendix~A. Note that the tolerant workload remains
unscaled. Thus, the SFJ scaling may be viewed as a time-scale
separation, with the eager class operating at a faster time-scale.

\noindent {\bf Static Pseudo Conservation:} Let
$\Omega^{\mu_\epsilon}_{j}(t)$ represent the total amount of server
capacity left unused by the $\epsilon$-customers in time interval
$[0,t]$, under sub-policy $j$ operating in a $\t$-static manner. Note
that $\Omega^{\mu_\epsilon}_{j}(t)$ is the (cumulative) service
process seen by the $\t$-system. Then by \cite[Lemma 1]{ANOR}, for all
$\mue,$ the asymptotic (in time) growth rate of
$\Omega^{\mu_\epsilon}_{j}(t)$ satisfies:

\vspace{-5mm}
{\small \begin{align}
    \label{Eqn_nujs}
          \lim_{t \to \infty}  \frac{\Omega^{\mu_\epsilon}_{j}(t) }{t}
          \ \to \   \nu_j  := 1 - \rho_\epsilon( 1 - {P_B}_j)
          \text{ almost surely.}
        \end{align}}
      In other words, {\it the long run time average service rate seen by the
      $\t$-queue equals $\nu_j,$ which depends only on the blocking
      probability ${P_B}_j$ of the eager class, and not on the specific
      $\i$-sub-policy that produced that blocking probability.} Further under
      the SFJ limit, the service process seen by the $\t$-class becomes
      uniform. Specifically it follows from \cite[Theorem~1]{ANOR} that, as
      $\mue \to \infty$ in the SFJ limit,
\begin{align}
\nonumber
& \sup_{t \le W}  \left | \Omega^{\mu_\epsilon}_{j}(t)   - \nu_j t \right |  \to   0 \ a.s. \mbox{ for any finite $W,$ and } \\
& \Upsilon_j^{\mu_\epsilon}  \  \convas  \   \frac { B_\tau }{ \nu_j} \mbox{, both  for any  initial $\epsilon$-state, } \label{Eqn_Upsilon_conv}
\end{align}
  where $\Upsilon_j^{\mu_\epsilon}$ denotes the time required to
  finish $B_\tau$ amount of work using the service process
  $\Omega^{\mu_\epsilon}_{j}(\cdot)$.
  This uniformity of the service process under SFJ limit enables a
  closed form characterization of the performance of the $\t$-class
  (see \cite[Theorems~2,3]{ANOR}). A key feature of the above results
  is a \emph{pseudo-conservation law} that expresses the performance
  of the tolerant class purely in terms of the blocking probability of
  the eager class, independent of the underlying $\i$-policy that
  produced the blocking probability. {\it We show an analogous
  pseudo-conservation for dynamic scheduling policies in this~paper.
}

Finally, we state the following assumption, which ensures that the
$\t$-queue remains stable under each
 eager sub-policy, when
they are applied in a $\t$-static manner.

\begin{enumerate}[\textbf{B}.4] 
\item   {There exists $\delta_\tau  > 0$
  such that $\rho_{j} := \frac{ \lambda_\tau } { \mu_\tau \nu_j } < 1-\delta_\tau
  \mbox{ for all } j.$}
\end{enumerate}
Assumption~\textbf{B}.4 guarantees that the $\t$-system is stable in
the dynamic setting as well, as is shown by
Lemma~\ref{lemma:TauStability}.

\subsection{Some example models}
\label{sec:example_models}

We begin with the description of one example system that satisfies our
assumptions, in which the system capacity is not completely
transferred to one class at any time, but rather a fraction of it is
used by each $\i$-customer, whilst the left is utilized by one
$\t$-customer.
 
\noindent {\bf Capacity Division (CD-$(p,K)$) policy:} Each
$\i$-customer uses {\small $(1/K)$} part of the service capacity.  If
there are {\small $0 \le \ell \le K$} number of $\i$-customers
receiving service, then $\i$-customers are served at a net service
rate of {\small $({\ell}/{K})$}, while the $\t$-customer in service
(if any) is served at rate {\small $((K-\ell)/K)$}. This continues up
to {\small $K$} $\i$-customers, and any further $\i$-arrival departs
without service. Note that whenever an existing $\i$-customer departs,
{\it the service rate of the $\t$-customer gets increased by {\small
    $1/K$}}. Of course, $1/K \geq c_{\min}$ to satisfy Assumption
\textbf{A}.3. Further there is a prior admission control on
$\epsilon$-arrivals, they are admitted with probability $p$
independent of all other events.  This is the description of the
$\i$-sub-policy (as in \cite{Value,ANOR}). Now the top-level policy
varies the probability of admission $p$ (and possibly the maximum $\i$
occupancy $K$) based on the occupancy of the $\t$-queue.
 
We refer to the above $\i$-sub-policy as Capacity Division or briefly
as the CD-$(p,K)$ policy.  Note that the
CD-$(p,K)$ policy captures a multi-server setting for the eager class.
While the $\i$-scheduler need not be work conserving, the $\t$-class
uses all the left over capacity. 
Another example of an $\i$-sub-policy is the following.
    
\noindent {\bf Limited Processor Sharing (LPS-$(p,K)$) policy:} This
sub-policy, denoted by LPS-$(p,K)$ (as in \cite{Value,ANOR}), admits
an incoming eager job into the system with probability $p,$ so long as
the number of eager jobs already in service is less than or equal to
$K.$ The entire service capacity of the server is shared equally
between the eager jobs in service (i.e., each eager job gets served at
rate $1/\ell$ when there are $\ell$ jobs in service). As before,
$1/K \geq c_{\min}.$ Note that under the LPS-$(p,K)$ policy, the
tolerant class receives service only when there are no eager jobs in
the system.

\ignore{
  \noindent {\bf Balking and Reneging:} \sout{As a part of a
    particular $\i$-sub-policy, the system may allocate a certain
    number of servers (recall that the system may be viewed as
    multi-server from the standpoint of the eager class) and might
    allocate a certain amount of waiting space for $\i$-class.  The
    $\i$-customers might respond to the resources allocated based on
    their patience levels: a) an $\i$-customer may not enter the
    system depending upon the $\i$-number already in system according
    to some probabilistic rule, as in {\it balking} models; or b) may
    leave the system after waiting for an exponentially distributed
    patience time of rate $\alpha^{\mue},$ as in {\it reneging}
    models.  In the case of reneging, we will require that the
    parameter $\alpha^{\mue}$ scales linearly with $\mue$, i.e.,
    $\alpha^\mue = \alpha \mue$ for some $\alpha \in (0, \infty)$} (as
  in \cite{ANOR}).
}

  \noindent {\bf Top-level policies:} The top-level policy can chose
  any one of these sub-policies for any $\t$-state. For example, when
  the $\t$-occupancy is greater than a certain threshold $L,$ one may
  allocate fewer individual servers to $\i$-customers (using, for
  example, the CD policy), while one may allocate the entire capacity
  to the $\i$-class and may serve them in LPS mode when the
  $\t$-occupancy is smaller.

  \revision{ \noindent {\bf Switching between sub-policies:} Note that
    under Assumption {\bf A}.1, the transition between sub-policies,
    triggered by an arrival/departure in the $\t$ queue, is
    simplified---any left over eager jobs at the time of the
    transition are simply `flushed'. This simplifies our analysis, and
    as argued before, would have a negligible impact on the
    performance of both classes in the SFJ scaling regime. In
    practice, a more natural implementation strategy would be to
    simply let the new sub-policy dictate the scheduling of the
    existing eager jobs in the system after the change of the $\t$
    occupancy. This implies a possible re-assignment of the service
    rates of eager customers at the time of the $\t$-transition.
    We prove formally that doing this does not affect our main results
    under the SFJ limit, for the special case of exponentially
    distributed eager job sizes (see Theorem~\ref{Thm_withou_A1}).}





\section{Performance Characterization under the SFJ limit}
\label{sec:tolerant}

In this section, we characterize the performance of the $\t$-class and
the $\i$-class under the SFJ limit.

Recall that in our system, both classes are scheduled in a
\emph{dynamic} manner, with the $\i$-sub-policy being selected based
on current $\t$-occupancy, and the $\t$-queue in turn being served (in
a work conserving fashion) using the unused service process of the
$\i$-sub-policy. Thus, the two classes experience random, time
varying, state-dependent, and also \emph{inter-dependent} service
processes. This makes a precise performance evaluation intractable;
indeed, no closed form characterization of the performance of the
tolerant class is possible even in the simplified setting where the
eager class is scheduled by a policy that is oblivious to the tolerant
state; see \cite{Mahabhashyam05,ANOR}. In this section, we show that
under the (partially fluid) SFJ scaling, tractable performance
characterizations are possible.

To provide intuition for the form of our results, recall that the
timescale separation resulting from the SFJ scaling results in the
tolerant queue obtaining, in the limit, service at a `steady' rate of
$\nu_{j} := 1 - \rho_\epsilon \left(1 - P_{B_j}\right)$ when there are
$j$ tolerant jobs in the system. Here, $P_{B_j}$ is the standalone, or
$\t$-static blocking probability associated with sub-policy~$j.$ One
might then anticipate that the limiting performance of the tolerant
class is described in terms of a system with a state-dependent service
rate, i.e., where the service rate is a (deterministic) function of the
queue occupancy. We prove that this is indeed the case. Before we
state our main results, we first describe this limit system, which we
refer to as a state-dependent service rate M/M/1 (SDSR-M/M/1) queue.

\subsection{State-dependent service rate M/M/1 queue}

An SDSR-M/M/1 queue sees the same arrival process as an M/M/1 queue:
job arrivals are according to a Poisson process (of rate
$\lambda$). Further the job sizes are independent and exponentially
distributed with mean $1/\mu$. However, unlike the standard M/M/1
queue, the SDSR-M/M/1 queue has a state dependent service rate
(a.k.a. server speed). Specifically, the server operates with service
rate $\nu_j$ if the number of jobs in the queue (including the job in
service) equals $j.$ Thus, the SDSR-M/M/1 queue is parametrized by
$(\lambda, \mu, \bm{\nu})$, where $\bm{\nu} = \{\nu_j\}_{j \geq 0}$ is
the vector of service rates.

The number of jobs in the SDSR-M/M/1 queue evolves as a continuous
time Markov process with birth-death structure (see Figure
\ref{Fig_SDSR-system}). 
\begin{figure}
	\begin{center}
		\includegraphics[width= .9 \textwidth]{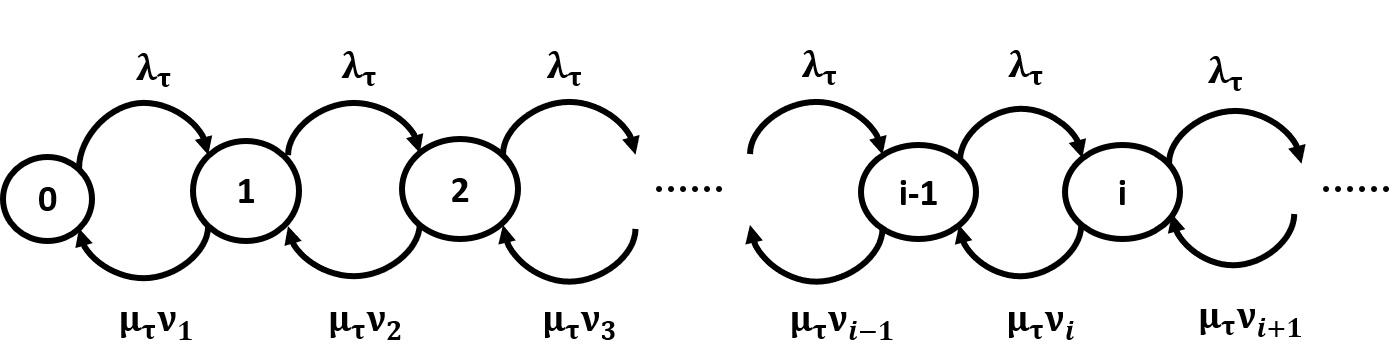}
		\caption{SDSR-M/M/1 queue with parameter $(\lambda_\t, \mu_\t, \bm{\nu} = \{\nu_1, \nu_2, \nu_3, \cdots\}).$}
		\label{Fig_SDSR-system}
	\end{center}
\end{figure}
Moreover, Assumption~{\bf B.}4 ensures that this Markov process is
positive recurrent. Thus, its steady state distribution can be
obtained by elementary techniques. In particular, the stationary
distribution ${\bm \pi} := \{\pi_{i} \}_{i = 0}^{\infty}$, is given by
(see \cite{Gallager2012}):
\begin{eqnarray}
\label{Eqn_SDSR_SD}
\pi_{i}   &=&   \frac{  \indicator{i=0} + \indicator{i \ge 1}  \prod_{\ell=1}^i \rho_{\ell} }{1 + \sum_{k \ge 1}    \prod_{\ell=1}^k \rho_{\ell} } \  
\mbox{,}  \  \rho_{\ell} :=  \frac{\lambda_\t}{\mu_\t \nu_{\ell}}.
\end{eqnarray}

\subsection{Main results}

Recall that the scheduler is specified by a sequence of eager
sub-policies one for each tolerant occupancy; $P_{B_j}$ being the
$\t$-static probability of the eager sub-policy used under tolerant
occupancy $j$. The performance of the tolerant class under the SJF
limit, and under any scheduler, is characterized as follows. The
steady state occupancy of the tolerant queue converges, in
distribution, and also in expectation, to the corresponding quantities
of an SDSR-M/M/1 system parameterized by
$(\lambda_\t, \mu_\t, \bm{\nu}).$

\begin{thm} {\bf{[Stationary distribution of $\t$-occupancy}]}
\label{Thm_Lim_SD_Conv}
Assume $\bm{A}$.1-4 and $\bm{B}$.1-4. Under the SFJ limit, the steady
state number $\t$-occupancy converges in distribution to
the steady state occupancy in an SDSR-M/M/1
$(\lambda_\tau, \mu_\tau, (\nu_1, \nu_2, \cdots ))$ queue, with
  $$
  \hspace{10mm}
  \nu_j :=    ( 1-  \rho_\epsilon (1-{P_B}_j) ) ,  \mbox{ for }  j \geq 0.
  $$
  That is,
  $$
  {\pi}_i^\mue \stackrel{\mue \to \infty}{\longrightarrow} {\pi}_i \ \mbox{for all} \  i \geq 0, 
  $$ where ${\pi}_i$ is given by \eqref{Eqn_SDSR_SD}.
\end{thm}

\begin{thm} {\bf [Stationary expected $\t$-occupancy]}
\label{Thm_tau}
Assume $\bm{A}$.1-4 and $\bm{B}$.1-4. Under the SFJ limit, the
stationary expected number of $\t$-customers converges to that of the
limit system described in Theorem~\ref{Thm_Lim_SD_Conv}. That is,
\begin{eqnarray*}
\label{Eqn_SDSR_Expected_Number}
   E^\mue[N] \stackrel{\mue \to \infty}{\longrightarrow}  \sum_{i=1}^\infty i \pi_{i},
\end{eqnarray*}
where ${\pi}_i$ is given by \eqref{Eqn_SDSR_SD}.
\end{thm}

Given that ${\bm \pi} := \{\pi_{i} \}_{i = 0}^{\infty}$ is the
stationary distribution of a simple birth-death CTMC,
Theorems~\ref{Thm_Lim_SD_Conv} and~\ref{Thm_tau} provide closed form
characterizations of the limiting performance of the tolerant class
under the SFJ scaling. It is important to note that while the
statements of Theorems~\ref{Thm_Lim_SD_Conv} and~\ref{Thm_tau} might
seem intuitive, their proofs are rather involved. In particular, they
rely crucially on a \emph{uniformity} in the convergence of the
service process seen by the tolerant queue under the SFJ limit,
\emph{across $\i$-sub-policies}; see
Subsection~\ref{sec:tol_perf_proofs}.

Next, we describe the performance of the eager class under the SFJ
limit, which is captured by its blocking probability, i.e., the long
fraction of eager customers blocked. Formally, the blocking
probability is defined as follows. Let $N_B^\mue (t)$ denote the total
number of eager customers that are blocked (i.e., returned without
service), before time $t$ and $N_A^\mue (t)$ be the total number of
eager customers arrived in the same time interval, for the system at
scale $\mu_\i.$
Then the blocking probability of the eager class is defined as the
long run fraction of customers blocked, i.e.,
\begin{equation}
  \label{eq:PB-def}
P_B^\mue := \lim\limits_{t \to \infty } \frac{N_B^\mue (t)}{N_A^\mue (t)}.
\end{equation}
Note that eager jobs are blocked by different sub-policies, which
operate dynamically based on the occupancy of the tolerant
queue. Thus, it is not a priori even clear that the limit in
\eqref{eq:PB-def} exists almost surely. That it does, and that the
eager blocking probability converges under the SFJ limit to a convex
combination of the $\t$-static blocking probabilities $\{ P_{B_j}\},$
weighted by the (limiting) long run fractions of time the $\t$-system
spends in each state, is established by the following theorem.

\begin{thm} \textbf{[Blocking probability of eager class]} 
\label{Thm_PB}
Assume {\bf A}.1-4, {\bf B}.1-4. Then the steady state blocking
probability of $\epsilon$-jobs in SFJ limit is given by:
{\small
  \begin{align*}
    \hspace{25mm} 
    \hspace{5mm}
    P_B^{\mu_\epsilon} \ \ \stackrel{\mu_\epsilon \to
  \infty}{\longrightarrow} \ \
\sum_{j=1}^{\infty } {P_B}_j {\pi}_j =: P_B^{\infty},
    \hspace{25mm} 
  \end{align*}}
where ${\pi}_i$ is given by \eqref{Eqn_SDSR_SD}.
\end{thm}

A key takeaway from Theorems~\ref{Thm_Lim_SD_Conv}--\ref{Thm_PB} is
that the performance of both classes under the SFJ limit can be
characterized in terms of the $\t$-static blocking probabilities
$\{ P_{B_j} \}_{j \geq 0}$. The probabilities
$\{ P_{B_j} \}_{j \geq 0}$ themselves are typically easy to compute,
since they involve the analysis of a single-class (stationary) loss
system. Thus, by virtue of
Theorems~\ref{Thm_Lim_SD_Conv}--\ref{Thm_PB}, we have the following
conservation law.

\noindent {\bf Dynamic Pseudo Conservation}: Under the SFJ limit, the
performance of \emph{both} classes depends only on the $\t$-static
blocking probabilities $\{ P_{B_j} \}_{j \geq 0}$ and not the
specifics of the $\i$-sub-policies that produced these blocking
probabilities.

\revision{Finally, since Assumption {\bf A}.1 may seem unreasonable,
  we show below that the conclusions of
  Theorems~\ref{Thm_Lim_SD_Conv}-\ref{Thm_PB} hold even without this
  assumption, if eager job sizes are taken to be exponentially
  distributed.
  \begin{thm}
    \label{Thm_withou_A1}
    Assume {\bf A}.2-4 and {\bf B}.1-4. Also assume that the
    $\i$-service times are exponentially distributed. 
The
    conclusions of Theorems~\ref{Thm_Lim_SD_Conv}-\ref{Thm_PB} hold, once   the
    transitions between eager sub-policies are handled  as described in
    Section~\ref{sec:example_models} (i.e., following an
    arrival/departure in the tolerant queue, the new eager sub-policy
    dictates the scheduling of the eager queue from then on).
  \end{thm}
  Theorem~\ref{Thm_withou_A1} shows that Assumption {\bf A}.1 is
  benign, i.e., it does not affect the performance of either class
  under the SFJ limit. As expected, the proof of
  Theorem~\ref{Thm_withou_A1} is considerably more involved as
  compared to the proofs of
  Theorems~\ref{Thm_Lim_SD_Conv}-\ref{Thm_PB}. Specifically,
  Assumption {\bf A}.1 allows the evolution of the $\t$ queue (across
  arrival/departure epochs) to be described by a one-dimensional
  birth-death Markov chain (see
  Section~\ref{sec:tol_perf_proofs}). This is no longer possible in
  the absence of Assumption {\bf A}.1; the system evolution has to be
  captured via a certain two-dimensional Markov chain, whose long-run
  time averages must be shown to match those corresponding to the
  former birth-death chain as $\mue \ra \infty$ (see
  Appendix~\ref{app:withoutA1}).
}

We note that while our performance characterization is derived under
the SFJ limit, we show in Section~\ref{sec:pareto} that they provide
accurate approximations of the performance experienced in the
pre-limit. The remainder of this section is devoted to highlighting
the main steps in the proofs of Theorems~\ref{Thm_Lim_SD_Conv}
and~\ref{Thm_tau}. Most details are relegated to
Appendix~\ref{Appendix_tolerant}. The proof of Theorem~\ref{Thm_PB} is
presented in Appendix~\ref{Appendix_eager}. \revision{The proof of
  Theorem~\ref{Thm_withou_A1} can be found in
  Appendix~\ref{app:withoutA1}.}

\subsection{Analysing tolerant performance under the SFJ limit}
\label{sec:tol_perf_proofs}

\begin{figure}
\begin{center}
\includegraphics[width= .9 \textwidth]{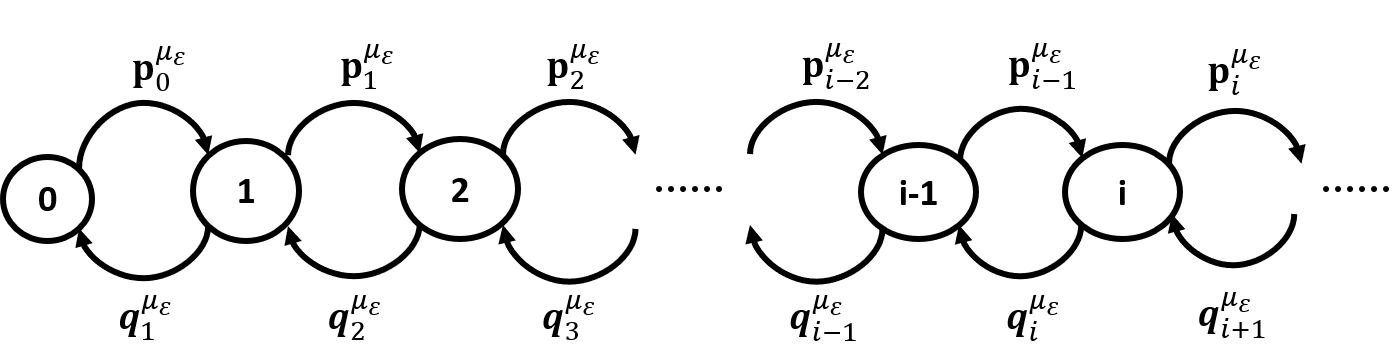}
\caption{State Transition $\tau$-queue corresponding to eager service
  rate $\mue$.}
\label{Fig_StateTrans_PreLimit}
\end{center}
\end{figure}

\revision{We now sketch the main steps in the proofs of
  Theorems~\ref{Thm_Lim_SD_Conv} and~\ref{Thm_tau} (i.e., under
  Assumption {\bf A}.1).} We characterize the $\t$-performance by
first analysing the $\t$-queue at its arrival/departure epochs. Let
$X_n$ denote the occupancy of the $\t$-queue immediately following the
$n$th arrival/departure. Under Assumption~{\bf A}.1 (dropping of
existing $\i$-customers at $\t$-transitions) and because of
exponentially distributed tolerant job sizes, $\{X_n\}$ is a
discrete-time Markov chain with birth-death (BD) structure; see
Figure~\ref{Fig_StateTrans_PreLimit}. Let $p_j^\mue$ and $q_j^\mue$,
respectively, denote the forward and backward transition probabilities
of this (pre-limit) tolerant BD chain, when there are $j$ tolerant
customers in the system.  Let $A_\tau$ and $B_\tau$ represent, a
generic inter-arrival time associated with the tolerant queue, and a
generic tolerant job size, respectively. Also,
recall~$\Upsilon_j^{\mu_\epsilon}(B_\t)$ denotes the time required to
finish $B_\tau$ amount of work using the service
process~$\Omega^{\mu_\epsilon}_{j}(\cdot)$ (see
Section~\ref{sec:background}), then we have:
\begin{equation*}
  \label{eq:q_j-def}
q^{\mue}_j = 1 - p^{\mue}_j:= \prob{A_{\t} > \Upsilon^{\mu_\epsilon}_j(B_\t)}.
\end{equation*}
Observe by {\bf A}.1 that the job completion times, distributed as
$\Upsilon_j^{\mu_\epsilon}(B_{\t }) $, are i.i.d.  across all those
customers that were served when the tolerant occupancy
equals~$j$. Moreover, if such a tolerant service is interrupted by a
$\t$-arrival, the remaining completion time is distributed as
$\Upsilon_{j+1}^{\mu_\epsilon}(B_{\t} ),$ thanks to the memoryless
property of tolerant job sizes and~{\bf A}.1.

Our first step is to prove that the above transition probabilities
converge to the transition probabilities associated with the embedded
BD chain corresponding to the
$(\lambda_\tau, \mu_\tau, (\nu_1, \nu_2, \cdots ))$ SDSR-M/M/1
system. Moreover, we show that \emph{the above convergence takes place
  uniformly over sub-policies~$j.$} Before stating this result, we
recall the definition of uniform convergence (which we denote by
$u.f.$):

\textbf{Definition:} $\left[\textbf{Uniform
    $(u.f.)$ convergence }\right]$ A parameterized family of sequences
$\{x_{j}^{\mue}\}_{j \geq 0}$ is said to be uniformly convergent to
the limit $\{x_j\}_{j \geq 0}$ as $\mue \ra \infty$ if, for every
$\epsilon >0$ there exists $\bar{\mu}$, such that
$$
\left| x_{j}^{\mue} - x_j \right| < \epsilon \ \ \mbox{for all} \ \
\mue > \bar{\mu} \mbox{ and for all } j \geq 0.
$$ 
\begin{lemma}
	\label{Lem_TP_Convergence}
	$\left[\textbf{Uniform convergence of transition
            probabilities}\right]$ i) If $p_j^\mue$ denotes the
        probability of a $\t$-arrival before a $\t$-departure in the
        (pre-limit) system when the number of tolerant customers in
        the system equals $j$, then
	\begin{eqnarray}
	\label{eq:q_conv}
	p_j^\mue \stackrel{\mue \to \infty}{\rightarrow} p_j^{\infty} := \frac{\lambda_\t}{\lambda_\tau + \nu_j \mu_\tau} \ \mbox{and } \ q_j^\mue := 1 - p_j^\mue   \stackrel{\mue \to \infty}{\rightarrow} q_j^{\infty}:= \frac{\nu_j \mu_\t}{\lambda_\tau + \nu_j \mu_\tau} . 
	\end{eqnarray}
	\label{lemma:uniform}
	ii) The convergence in \eqref{eq:q_conv} occurs u.f. over the
        sub-policies $j$.  Precisely, for every $\epsilon > 0$ there
        exists $\bar{\mu} > 0$, such that for all
        $\mu_\epsilon \geq \bar{\mu}$ and for all $j \geq 0,$ we have:
	\begin{eqnarray}
	\bigg| q_j^{\mu_\epsilon} - \frac{\nu_j \mu_\tau}{\lambda_\tau + \nu_j \mu_\tau}  \bigg| < \epsilon.
	\end{eqnarray} 
\end{lemma}
The proof of Lemma~\ref{Lem_TP_Convergence} is provided in
Appendix~\ref{sec:lemma:uniform}.

The next step is to show that the embedded BD chain characterizing the
evolution of the occupancy of tolerant queue is positive recurrent for
large enough $\mue.$
\begin{lemma}
  \label{lemma:TauStability}
  There exists $\bar{\mu} > 0$ such that for $\mue > \bar{\mu},$ the
  (embedded) Markov chain $\{X_n\}$ is positive recurrent.
\end{lemma}
The proof of Lemma~\ref{lemma:TauStability} is provided in
Appendix~\ref{sec:lemma:TauStability}. Lemma~\ref{lemma:TauStability}
also implies that for large enough $\mue,$ the $\tau$-queue is stable
and has a well defined stationary behaviour.\footnote{The occupancy of
  the tolerant queue evolves as a semi-Markov process in continuous
  time. The regularity of this process (i.e., each state is visited
  only finitely often in any finite interval of time with
  probability~1) follows by noting that the number of state
  transitions is lower bounded by those in a standard M/M/1 queue with
  arrival rate $\lambda_\t$ and service rate $\mu_\t$. The positive
  recurrence of this process follows from Theorem~5.9 in
  \cite{Gallager2012}, noting that mean residence time in any state is
  uniformly upper bounded by $1/\lambda_\t.$} The stationary
distribution for $\mue > \bar{\mu}$ can be calculated using standard
techniques and is given by (see, for example,~\cite{Hoel}),
\begin{eqnarray}
\label{Eqn_SD_pre-limit}
\tilde{\pi}_i^\mue = \frac{p_0^\mue p_1^\mue \cdots p_{i-1}^\mue}{q_1^\mue q_2^\mue \cdots q_i^\mue} \tilde{\pi}_0^\mue, \ \ \ \ \mbox{for} \ i \geq 1,
\end{eqnarray}with $\tilde{\pi}_0^\mue$  defined as:
\begin{equation}
\label{Eqn_pi_0_pre-limit}
\tilde{\pi}_0^\mue = \frac{1}{1 + \sum_{k = 1}^\infty \frac{p_0^\mue p_1^\mue \cdots p_{k-1}^\mue}{q_1^\mue q_2^\mue \cdots q_k^\mue}}.
\end{equation}

Now, given that the embedded BD chain capturing the evolution of the
tolerant queue is positive recurrent for large enough $\mue$
(Lemma~\ref{lemma:TauStability}), and has transition probabilities
that converge (under the SFJ limit) uniformly to those of the embedded
chain of the $(\lambda_\tau, \mu_\tau, (\nu_1, \nu_2, \cdots ))$
SDSR-M/M/1 queue (Lemma~\ref{Lem_TP_Convergence}), we can now
establish convergence of the stationary distribution of this embedded
process to that of the SDSR-M/M/1 system as well.

\begin{thm}
  $\left[\textbf{Stationary number in the embedded tolerant
      chain}\right]$
  \label{Thm_SD_Convergence}
  Under SFJ limit,
  the stationary distribution of the embedded (BD) chain corresponding
  to the tolerant Markov process given by equation
  (\ref{Eqn_SD_pre-limit}) and (\ref{Eqn_pi_0_pre-limit}) converges to
  that of the SDSR-M/M/1 queue with variable service rates. That is,
  $$
  \tilde{\pi}_i^\mue \stackrel{\mue \to \infty}{\longrightarrow} \tilde{\pi}_i \ \ \mbox{for all} \  i \geq 0
  $$ where, $\tilde{\pi}_i$ is given by \eqref{Eqn_SDSR_SD}.  
\end{thm}

Note that Theorem~\ref{Thm_SD_Convergence} deals with the limiting
behavior of the \emph{embedded} chain corresponding to the (continuous
time) occupancy process of the tolerant queue. Our main results
Theorems~\ref{Thm_Lim_SD_Conv} and~\ref{Thm_tau} can be proved from
Theorem~\ref{Thm_SD_Convergence} by invoking a relationship between
the stationary distribution of an embedded Markov chain and that of
the corresponding continuous time process (see
Lemma~\ref{lemma:EMC_dist} in Appendix~\ref{sec:Thm_Lim_SD_Conv}); the
details can be found in Appendices~\ref{sec:Thm_Lim_SD_Conv}
and~\ref{sec:Thm_tau}.

\ignore{ 
Recall that under the SFJ limit,
$\Upsilon_j^\mue(B_\t) \stackrel{\mue \to \infty}{\longrightarrow}
\frac{B_\t}{\nu_j}$ almost surely. Thus, with $A_\t$ being
exponentially distributed with mean $1/\lambda_\t,$ and $B_\t$ being
exponentially distributed with mean $1/\mu_\t,$ one would expect that
$$q^{\mue}_j = \prob{A_{\t} > \Upsilon^{\mu_\epsilon}_j(B_\t)}
\stackrel{\mue \to \infty}{\longrightarrow} \frac{\nu_j
  \mu_\t}{\lambda_\t + \nu_j \mu_\t}.$$ Lemma~\ref{Lem_TP_Convergence}
below shows that the above statement is true. But more importantly, it
establishes that the above convergence takes place \emph{uniformly}
over sub-policies~$j.$ Indeed, establishing uniform convergence of the
transition probabilites of the embedded BD chain (capturing the
evolution of the $\t$-occupancy) under the SFJ limit is essential for
our performance characterizations.

Recall that under $\t$-static policies the left over service capacity
to the tolerant class converges uniformly, under SFJ limit, to the
rate $\nu := 1 - \rho_\epsilon \left(1 - P_B\right)$, where $P_B$
denotes the blocking probability of the eager class. In a similar
fashion, for the dynamic setting of our paper, {\it we anticipate that
  the left over (by eager class) service process converges to the
  constant rate processes with rate
  $\nu_{j} := 1 - \rho_\epsilon \left(1 - P_{B_j}\right)$ (as in
  (\ref{Eqn_Upsilon_conv})), when $\tau$-state is $j$, once again by
  the resulting time scale separation}; recall the standalone blocking
probability corresponding to the sub-policy $j$ equals $P_{B_j}$.

Thus, under SFJ limit, the tolerant customers may see uniform service
process (see Figure(?)) but with different service rates depending
upon the $\t$-occupancy.  We further claim that the service times
experienced by the tolerant customers, at SFJ limit and with
$j$-occupancy, follows Exponential distribution with rate
$ \nu_j \mut$; that is, we claim that
$\Upsilon_j^\mue(B_\t) \stackrel{\mue \to \infty}{\longrightarrow}
\frac{B_\t}{\nu_j}$, which is distributed according to
$ Exp(\nu_j \mut)$.

We will indeed prove that the above claims are true.  In fact, we
would directly prove the convergence of transition probabilities in
the Lemma given below; recall that $\{q^{\mue}_j \}$ precisely depend
upon the service process available to the tolerant customers.  To be
precise, we would show the following:
\begin{equation}
q^{\mue}_j 
\stackrel{\mue \ua \infty}{\longrightarrow} \frac{\nu_j
	\mu_\t}{\lambda_\t + \nu_j
	\mu_\t}, \ \ \mbox{uniformly (u.f.) over sub-policies} \ \  j,
\end{equation}
when the second moment of the $\i$ busy-cycles (for $\mue = 1$),
represented by $\{{\B}_j \}_j$, are uniformly bounded across all the
sub-policies (proof in Appendix~\ref{Appendix_tolerant}).
}

\ignore{
\noindent
{\bf Convergence of stationary distribution:} The stationary distribution of the \textcolor{red}{(embedded)} tolerant BD chain,  given by Equation (\ref{Eqn_SD_pre-limit}),  converges to that of the SDSR-M/M/1 queue with variable service rate defined above (Proof in Appendix~\ref{Appendix_tolerant}). The same thing is true for steady state distribution of $\t$-process, (under ergodicity) which is also the time average of the number in the system as defined below:
$$
\pi^\mue_i = \lim_{T \to \infty}  \frac{1}{T}  \int_0^T N^\mue_\tau (t) dt \mbox{ almost surely  and for any }  i \ge 0.  
$$
	\noindent
	{\bf Remark:} Note that the tolerant class experience time
        varying service process corresponding to different
        $\i$-sub-policy, which further depends on the
        $\t$-state. Also, under the fluid regime we considered, a
        constant service rate is experienced by the tolerant customers
        for the each fixed $\t$ occupancy. We have demonstrated, when
        the eager rate tends to infinity, tolerant B-D chain converges
        to that of SDSR M/M/1 chain and hence their
        performance(stationary distribution) will also converge. In
        view of Lemma \ref{lemma:EMC_dist}, once the stationary
        distribution of enbedded chain is known, the stationary
        distribution of the corresponding process can be calculated
        easily (should we write equation relating the two SD
        here?). Thus, convergence of SD of enbedded process implies
        the convergence of SD of the process, which we will prove in
        the next theorem (with proof in
        Appendix~\ref{Appendix_tolerant}).

To provide intuition for Theorem~\ref{Thm_tau}, note from
(\ref{Eqn_Upsilon_conv}) that $\Upsilon_j^{\mue}$, the time required
to complete $B_\t$ amount of job (when uninterrupted) converges to
$B_\t / \nu_j$, which is exponentially distributed.  Thus one can
anticipate the following $\t$-system at SFJ limit (further because the
residual service times are exponentially distributed): a) Poisson
arrivals; b) exponential service times, whose rate depends upon the
$\t$-number in the system. And this is precisely the SDSR-M/M/1 queue.
%

Theorem \ref{Thm_tau} tells that the performance of $\t$-class can be characterised via SDSR-M/M/1 queuing system with appropriately defined service rates. This concludes a remarkable result saying  the $\t$-perfomance depends  
only on the $\t$-static blocking probabilities of the eager class and not on any other detail of the system.
} 

\kcmnt{
When $\g = \infty$:   For each $j$, as $\mue \to \infty$
$$
\rho_j^\mue \to \rho_j :=  \frac{\lambda}{\mu \nu_j}, \mbox { so does this, }
 \prod_{i=1}^j \rho_i^\mue  \to  \prod_{i=1}^j \rho_i 
$$
Then by DCT, as applied to series with counting measure:
$$
\sum_{k \ge 1}    \prod_{j=1}^k \rho_j^\mue   \to \sum_{k \ge 1}    \prod_{j=1}^k \rho_j ,
$$
if we can say that  there exists an ${\bar \mu} < \infty $  such that  
$$
\boxed{
\prod_{j=1}^k \rho_j^\mue   \le  (\alpha+\delta)^k \mbox{ for all }  k \mbox{ and for all }   \mue \ge {\bar \mu} ,}
$$
where $\alpha$ is the upper bound on limit as defined below and $\delta >0 $ is a constant such that $\alpha + \delta < 1.$
And  the limit is a finite sum, if we have an upper bound 
$$
\mbox{i.e., if }  \boxed{ \sup_j  \frac{\lambda_\tau}{\mu_\t \nu_j} := \alpha < 1}, \mbox{ then, }
\sum_{k \ge 1}    \prod_{j=1}^k \rho_j  \le \sum_{k \ge 1}     \alpha^k < \frac{1}{1-\alpha}  < \infty.
$$
Then 
$$
\pi_{ i}^\mue \to \pi_{ i}  \mbox{ for each } i.
$$
Then by MCT again 
$$
E^\mue [N  ]  \to E^\infty[N ],
$$
and the limit is finite b/c the SDSR-M/M/1  is stable. 
}

\ignore{
\section{Performance of Eager Class}
\label{sec:eager}

We now focus on the performance of eager class.  As already discussed
the $\epsilon$ sub-policy changes dynamically 
depending only upon the $\tau$-number in the system\conftext{~(we assume $\g < \infty$ in this section; the generalization to
  arbitrary partitions  is in \cite{arxiv})}. To be more
precise, if the number of $\t$-customers at $\t$-transition
(arrival/departure) is $j$
sub-policy $j$ is used till the next $\t$-transition.

The $\i$-class is a lossy system, and the blocking probability would
be ${P_B}_j$ if $j$-th sub-policy is used in $\t$-static manner.  We
now derive the `dynamic' blocking probability, when these sub-policies
are selected based on $\t$-dynamics.  We show that, in SFJ limit, the
overall blocking probability of the eager class is a convex
combination of the $\t$-static blocking probabilities $\{ P_{B_j}\},$
weighted by the long run fractions of time the $\t$-system spends in
the groups~$\{\mathcal{G}_j\}.$ We begin with the definition, followed by the required result. 

{\color{red}\noindent {\bf Blocking probability:} Let $N_B^\mue (t)$ denote the total number of eager customers that are blocked (i.e., returned without service),  before  time $t$ and  $N_A^\mue (t)$ be  the total number of eager customers arrived in the same time interval, for system with  the  given  $\mu_\i$ and  under any given dynamic policy.   
Then the blocking probability of the eager class is defined as the long run fraction of customers blocked, i.e.,  
	$$
P_B^\mue :=	\lim\limits_{t \to \infty } \frac{N_B^\mue (t)}{N_A^\mue (t)}.
$$The  $\t$-static blocking probability $P_{B_j}^\mue $ is defined in a similar way, when sub-policy $j$ is used for all $\t$-occupancies (i.e., in $\t$-static manner) and  $P^\infty_{B_j}$ represents the limit of this  $\t$-static blocking probability under SFJ limit.  One first needs to ensure that the above limits exist and then consider the limit of the same under SFJ limit, both these tasks are considered in the following theorem.
}


\begin{thm} \textbf{[Blocking probability of eager class]} 
\label{Thm_PB}
Assume {\bf A}.1-4, {\bf B}.1-4.\conftext{~Also assume $\g < \infty$.}
Let ${{\bm \pi}} := \{{\pi}_{i} \}_{i \geq
  0}$ denote the stationary distribution of SDSR-M/M/1 limit tolerant
system, given by Theorem~\ref{Thm_tau}.
Then the steady state blocking probability of $\epsilon$-jobs in SFJ
limit is given by:

\vspace{-5mm}
\conftext{
{\small\begin{align*}
\hspace{5mm} P_B^{\mu_\epsilon} \ \ \stackrel{\mu_\epsilon \to
  \infty}{\to} \ \
\sum_{j=1}^{\g } {P_B}_j \left(\sum_{i \in {\mathcal G}_j} \pi_{\t}^\infty (i) \right) =: P_B^{\infty}.
\hspace{5mm}   \mbox{ \eop }
\end{align*}}}
\TRtext{
{\small\begin{align*}
\hspace{25mm} 
\hspace{5mm} P_B^{\mu_\epsilon} \ \ \stackrel{\mu_\epsilon \to
  \infty}{\longrightarrow} \ \
\sum_{j=1}^{\infty } {P^\infty_B}_j {\pi}_j =: P_B^{\infty} \mbox{ almost surely. }
\hspace{25mm}   \mbox{ \eop }
\end{align*}}}
\end{thm}
\conftext{A sketch of the proof of Theorem~\ref{Thm_PB} can be found
  in the appendix; the complete proof, which also generalizes to  countably many partitions, can be found in \cite{arxiv}.}\TRtext{The proof of
  Theorem~\ref{Thm_PB} can be found in the Appendix~\ref{Appendix_eager}.}

\subsubsection*{Dynamic Pseudo Conservation and its relevance}  
  
  The key challenge in the performance evaluation of our multi-class
  system is the interdependence between the service processes of the
  two classes. However, Theorems~\ref{Thm_tau}-\ref{Thm_PB} show that
  one may approximate the performance of both classes (the
  approximations being accurate under the SFJ limit) using only the
  $\t$-static blocking probabilities \conftext{$\{ P_{B_j} \}_{j \leq
      \g}$}\TRtext{$\{ P_{B_j} \}_{j \geq 0}$}. The probabilities
  \conftext{$\{ P_{B_j} \}_{j \leq \g}$}\TRtext{$\{ P_{B_j} \}_{j \geq
      0}$} themselves are typically easy to compute, since they
  involve the analysis of a single-class (stationary) loss
  system. Finally, we note that by virtue of
  Theorems~\ref{Thm_tau}-\ref{Thm_PB}, we have a {\it`Dynamic Pseudo
    Conservation'}: Under the SFJ limit, the performance of
  \emph{both} classes depends only on the $\t$-static blocking
  probabilities \conftext{$\{ P_{B_j} \}_{j \leq \g}$}\TRtext{$\{
    P_{B_j} \}_{j \geq 0}$}\conftext{~and the partitions
    $\{\mathcal{G}_j\}_{j \leq \g},$} and not on other specifics of
  the\conftext{~$\g$} sub-policies.

} 

\section{Dynamic Achievable Region}
\label{sec:Achievable Region}

A queuing system can be analysed using several performance metrics;
for example, number of customers in the system, sojourn time (the
total time spent by the customer), waiting time (of the customer
before the service starts), fraction of the customers blocked (in a
loss system), etc. The achievable region of a multi-class system is
defined as the region of all possible vectors (one component for one
class) of the relevant performance metrics.\footnote{In this paper we consider the achievable region corresponding to schedulers which satisfy Assumptions {\bf A.}1-4 and {\bf B.}1-4.}  In our model,
corresponding to the eager class we have a lossy system, thus we
consider blocking probability as the performance metric. For the
tolerant class, one can consider the steady state expected number of
customers in the system as the performance metric.

By Lemma \ref{lemma:TauStability}, the system is stable for all
$\mu_\epsilon \geq \bar{\mu}$, for some $\bar{\mu} < \infty$.  Thus
for all such $\mu_\epsilon$, by Little's Law, $E^{\mu_\epsilon} [S],$
the stationary expected sojourn time of a typical $\tau$-customer, and
$E^{\mu_\epsilon} [N],$ the stationary expected number of
$\tau$-customers in the system are related as $ E^{\mu_\epsilon} [ N ]
= \lambda_\tau E^{\mu_\epsilon} [S ].$ Thus it is sufficient to
consider any one of these metrics.\\

{\bf Stationary Markov top-level policies:} In any general sequential
decision problem, a Stationary Markov (SM) policy is a sequence of
decisions, in which one decision is chosen for each value of the state
and the same decision is applicable in any time slot.  In our case we
consider {\it the top-level policies among the Stationary Markov (SM)
  family.}  This means, a top level policy $\phi$ is a sequence of
$\i$-sub-policies, and that if the $\t$-state equals $j$ at any time
slot, then the $j$-th sub-policy of~$\phi$ is used for scheduling the
$\i$-class.

\conftext{While the statements of Theorems \ref{Thm_tau}-\ref{Thm_PB}
  in the present paper assume finitely many $\i$-sub-policies (one for
  each of the subsets in $\{{\mathcal G}_j \}$), the study of SM
  strategies requires us to move on to the general case of countably infinite
  $\i$-sub-policies, one for each value of $\t$-occupancy. However, as
  stated before, the statements of Theorems \ref{Thm_tau}-\ref{Thm_PB}
  do extend to general SM top-level policies (see
  \cite{arxiv}). {\it Accordingly, in the remainder of this paper, we
  proceed with our analysis of the achievable region and its Pareto
  frontier disregarding the restriction to finitely many
  $\i$-sub-policies.}}

As understood from Theorems~\ref{Thm_tau}-\ref{Thm_PB}, the {\it only
  characteristic of the $\i$-sub-policies that influences the system
  (dynamic) performance are the $\t$-static blocking probabilities $\{
  {P_B}_j\}_j$}, obtained when respective sub-policies are used in
$\t$-static manner. Thus to define an efficient dynamic system, one
effectively needs to choose (based on the $\t$-state), one among these
blocking probabilities (and no further details of the sub-policy are
important).  This {\it is a consequence of the `dynamic
  pseudo-conservation' mentioned in the previous section.}

Any Stationary Markov (SM) top-level policy is generally given by a
sequence of $\epsilon$-sub-policies, one for each
$\tau$-state. However, in view of the above observation, a stationary
Markov policy can be thought of as a sequence of $\epsilon$-blocking
probabilities (derived when the corresponding sub-policies are applied
in $\t$-static manner.  In other words, a SM top-level policy is
defined by $\phi = (d_0, d_1, \cdots )$, where decision $d_j$
specifies a `$\t$-static blocking probability' to be chosen when
number of $\t$-customers equals $j$.  {\it Towards this we implicitly
  require the existence of at least one sub-policy, that achieves the
  given value of `$\t$-static blocking probability', which is any
  value between the system specified limits $\dmin:={\underline
    {P_B}}$ (minimum possible blocking probability) and $\dmax :=
  {\overline {P_B}}$} (the maximum possible blocking probability).
This for example, is achieved by CD-$(p, K)$/LPS-$(p,K)$ policies
mentioned in Section \ref{sec:model}, when one considers all possible
values of $\{(p, K)\}$ (see \cite{Value, ANOR} for more details). In
the rest of the paper, we refer to the top-level policies simply as
policies for brevity.\\

\ignore{
{\bf Achievable region:} For any $\mue > {\bar \mu}$ (given by Lemma
\ref{lemma:TauStability}), the system is stable and one can define the
achievable region under stationary policies.  The achievable region is
the set of all possible pairs of blocking probability of the eager
class and the stationary expected number of tolerant customers in the
system, i.e., \kcmnts{we may not have common ${\bar \mu}$}
\begin{align*}
\mathcal{A}^{\mu_\epsilon} =\left  \{  \left (   P_{B, \phi}^{\mu_\epsilon} , \   \  E^{\mu_\epsilon}_\phi [N ]  \right ): \phi\ \ \mbox{ is an SM    policy }  \right \}.
\end{align*}
}

{\bf Limit Achievable region} Our focus from here on will be the
dynamic achievable region $\mathcal{A}^{\infty}$ of performance
vectors under the SFJ limit. Recall that for tolerant class, the limit
is an SDSR-M/M/1 queue. The $\epsilon$-limit can be seen as a mixture
model made up of many lossy systems, each described by their
$\t$-static blocking probabilities, and mixed independently according
the stationary distribution of the limit SDSR-M/M/1 queue.  Thus, we
define the limit achievable region as follows:

{\small \begin{align*} \hspace{15mm}
\mathcal{A}^{\infty} =\left  \{  \left  ( \ P_{B, \phi}^\infty, E_\phi^{\infty}[N ]  \right  )  : \phi \ \ \mbox{is an SM  policy} \right  \}.
\end{align*}}
Note that $\mathcal{A}^{\infty}$ is the set of limiting performance
vectors under SM policies. In this sense, one may view
$\mathcal{A}^{\infty}$ as the limit of achievable region of our
multi-class system as $\mue \ra \infty.$

For simplicity of notations {\it we avoid the super-script $\infty$}
when the discussion is clearly about the limit system.  {\it At times
  we also drop $\phi$, the SM policy, when there is no ambiguity.}
\\ 

  {\bf A Numerical Example:} To visualize the limit achievable region,
  we consider a system with a top-level policy parametrized by
  ($(p_1, p_2, L, K)$). In this system, the CD-$(p_1,K)$ policy is
  employed when the $\t$-occupancy is less than $L,$ and the
  CD-$(p_2,K)$ policy is employed when the $\t$-occupancy is greater
  than or equal to $L.$
    The tolerant customers are served serially with total capacity of
  all the leftover servers.
 Using the Erlang-B formula, the two $\t$-static blocking probabilities
  of $\i$-customers equal   
  \beq \hspace{15mm} {P_B}_i &=& (1-p_i) + p_i
    \frac{\frac{ (K\rho_\epsilon p_i)^K} {K!}  }{\sum_{k=0}^K \frac{
        (K \rho_\epsilon p_i)^k} {k!  }} , \mbox{ for } i = 1, 2,  
 \eeq  where ${P_B}_1$ is the $\t$-static blocking probability when $\t$-occupancy is less than $L$, while  ${P_B}_2$ is the $\t$-static blocking probability  for the rest of the tolerant states. The performance of such a system at
  limit can be obtained using the results of Theorems
  \ref{Thm_tau}-\ref{Thm_PB}.  We set $ K = 5$, $\rho_\i = 0.4$,
  $\lambda_\tau = 4 $ and $\mu_\tau = 8$, generate the three
  parameters $(p_1, p_2, L)$ randomly. The scatter plot of the
  corresponding values of $E^\infty [N]$ and $P^\infty_B$ is shown in
  Figure \ref{Fig_SM_AR_PF}. The resulting figure is a part of the
  limit achievable region. {\it As seen from the figure, the achievable
  region is a non-zero measure set; thus we need the `efficient' Pareto frontier.} Further, the plot indicates that
  the achievable region is bounded. We will now address the Pareto
  frontier associated with this system. 

  \begin{figure}[h]
 
\vspace{-34mm}
\begin{center}

 \includegraphics[width=7cm, height=5.4cm]{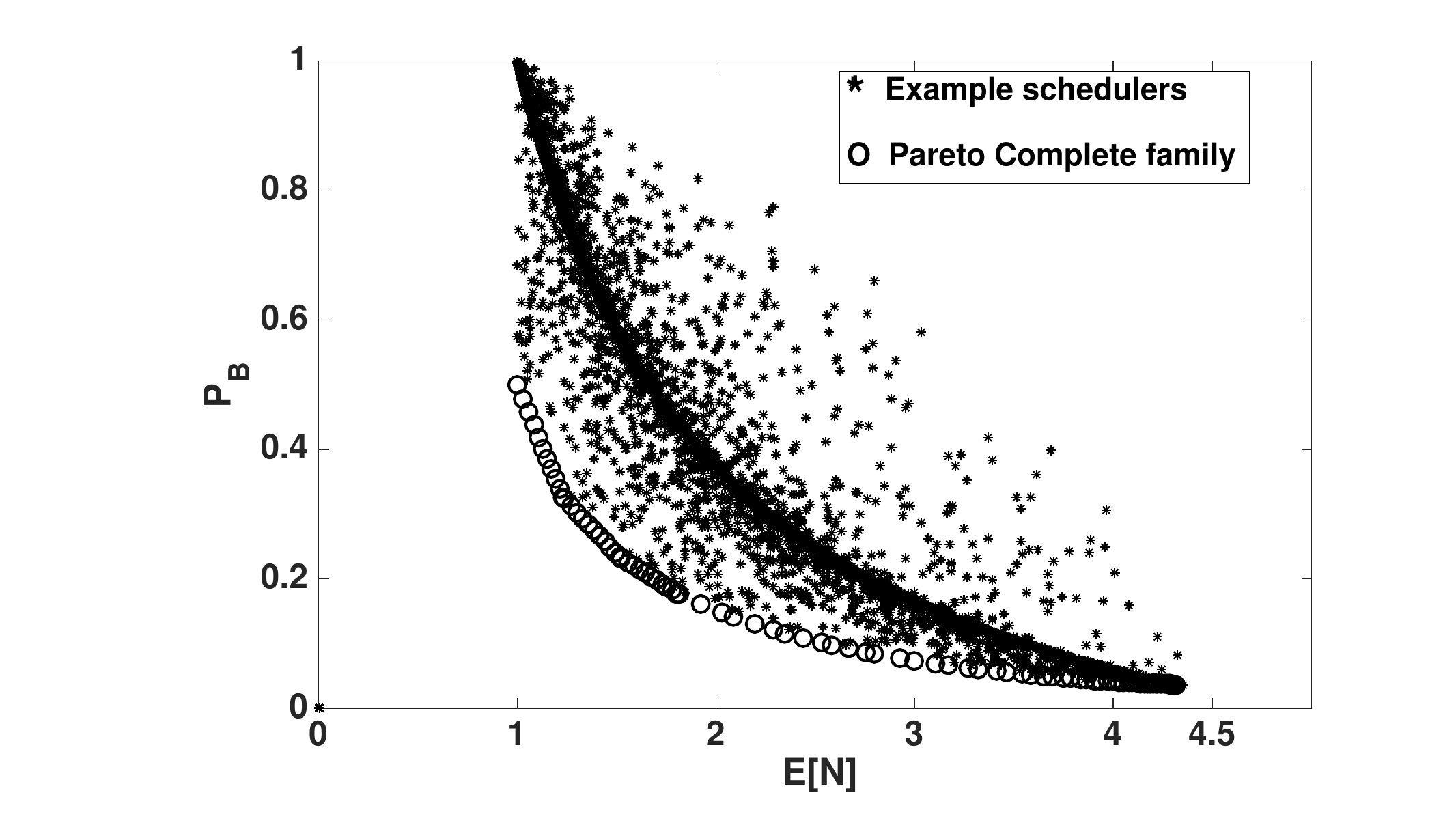} 
 \vspace{-29mm}
  \caption{Achievable region,   Pareto Frontier   $\rmax =  0.8134$, $\rmin = .5$ \label{Fig_SM_AR_PF} }
 
\end{center}

\end{figure}
\vspace{-2mm}
\section{Limit Pareto frontier}
\label{sec:pareto}

 The Pareto frontier is the efficient sub-region of an achievable
 region which consists of dominating performance vectors.
{\it  A pair $(P_{B, \phi^*}, E_{\phi^*}[N])$ (produced by a  
 policy~$\phi^*$) is on Pareto frontier of the limit system, if there
 exists no other SM policy $\phi$ that achieves a better performance
 pair $(P_{B, \phi}, E_\phi [N])$ (in the limit system),} i.e., $P_{B,
   \phi} \le P_{B, \phi^*} \mbox{ and } E_{\phi}[N] \le
 E_{\phi^*}[N],$ one of the inequalities being strict.

The Pareto frontier of the limit system is obtained by solving an
appropriate set of parametrized optimization problems.  Prior to that,
we discuss the limit performance, under any
given SM policy. Invoking Theorems~\ref{Thm_tau}
and~\ref{Thm_PB}\conftext{~(specifically, their generalizations in
  \cite{arxiv})}, under any $\phi$, 
\begin{eqnarray}
\label{Eqn_EN}
E^{\mu_\epsilon}_\phi[N]    \stackrel{\mu_\epsilon \to \infty }{ \longrightarrow}  E^{\infty}_\phi[N]  \hspace{-1mm}&  \hspace{-1mm} = \hspace{-1mm}&  \hspace{-1mm}\sum_{i = 1}^\infty   i \frac{ \h^\phi_i }{1 + \sum_{l\ge 1} \h^\phi_l},  \mbox{ and } \\
\label{Eqn_PB_under_phi}
P_{B, \phi}^{\mu_\epsilon}  \stackrel{\mu_\epsilon \to \infty }{ \longrightarrow}    P^\infty_{B, \phi}  \hspace{-1mm}&  \hspace{-1mm} = \hspace{-1mm}&  \hspace{-1mm} \frac{d_0}{1 + \sum_{i \ge 1}  \h^\phi_i } +  \sum_{i=1}^\infty d_i  
\frac{ \h^\phi_i}{1  + \sum_{l \ge 1} \h_l^\phi}, \mbox{ where }
\\ 
\label{Eqn_pi_under_phi}
   \h^\phi_j  \hspace{-1mm}&  \hspace{-1mm} = \hspace{-1mm}&  \hspace{-1mm}  \indicator{j = 0} + \indicator{j >0} \prod_{i=1}^j \rho^\phi_i
   , \ \mbox{with } \rho^\phi_i := \frac { \lambda_\tau }{\mu_\tau (1
     - \rho_\epsilon (1-d_i) )}.
\end{eqnarray}

\ignore{We would like to iterate again here that Theorems
  \ref{Thm_tau}-\ref{Thm_PB} are proved only for finite SM policies;
  the above {\it is only a conjecture for infinite policies.}  However
  we will notice that the {\it Pareto-optimal policies are of finite
    nature and the theorems are true for them.}}

 \vspace{-4mm}
\subsection{Pareto-complete family}
We now derive a family of {\it Pareto-complete} policies, i.e., a
parametrized family of policies that span the entire Pareto-frontier
of our system.
One can obtain all the points on the Pareto frontier by considering
the following parametrized (by $C$) constrained optimization problems
(with $\h_i^\phi $ defined in (\ref{Eqn_pi_under_phi})). 
l\begin{align}
\label{Eqn_Pareto_optimal} 
\hspace{10mm}
\min_{\phi}    P_{B, \phi} \mbox{ such that }   E_\phi [N]  \le C \mbox{, i.e., equivalently } \hspace{1mm} \\
\min_{\phi  }   \sum_{i=0}^\infty d_i  \frac{ \h_i^\phi }{1 + \sum_{l \ge 1} \h_l^\phi  }   \mbox{ such that }    \sum_{i=0}^\infty i \frac{ \h_i^\phi } {1 + \sum_{l \ge 1} \h_l^\phi  } \le C.
\nonumber
\end{align}
 \extrabits{
 \begin{figure}
 \vspace{-5mm}
 \hspace{-16mm}
 \includegraphics[width=8cm, height=5.2cm]{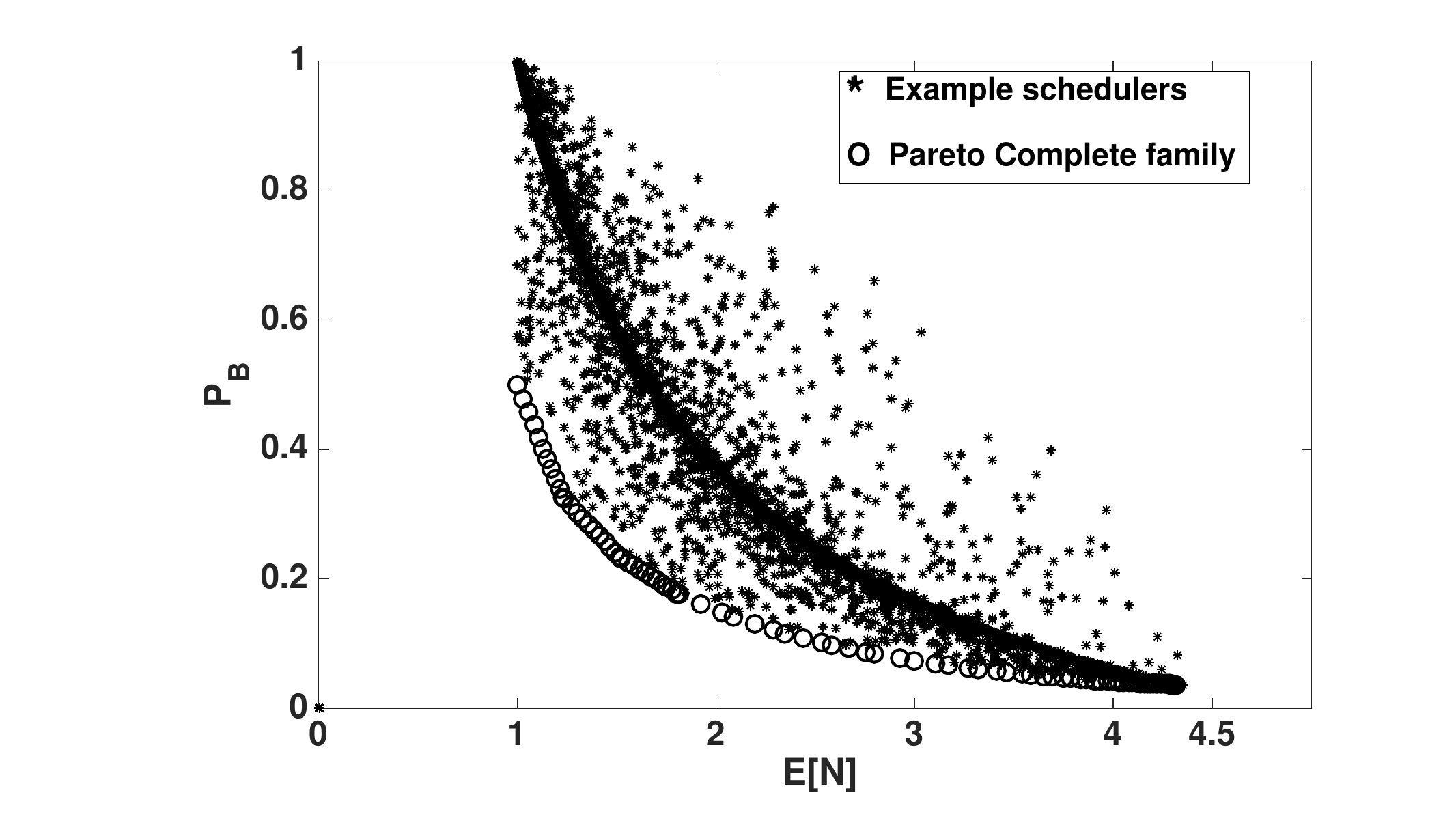} 
 \vspace{-3mm}
  \caption{Achievable region,   Pareto Frontier   $\rmax =  0.83$, $\rmin = .5$ }
   \includegraphics[width=9cm, height=5.2cm]{ParetoForTwoConfig-eps-converted-to.pdf}
  \caption{  Pareto Frontier  via Pareto complete family }
    
 \end{figure}}
Recall $\dmin$, $\dmax$ respectively represent the best and worst
sub-policy (with respect to $\i$-customers), in that these represent
the minimum and maximum possible blocking probabilities.  Define:
 
 \vspace{-6mm}
 {\small 
 \begin{align}
 \label{Eqn_rhomax_min}\hspace{5mm}
 \rmax := \frac {\lambda_\tau }{ \mu_\tau  (1-\rho_\epsilon  (1- \dmin)) }, 
  \mbox{  and }  \rmin = \frac { \lambda_\tau } { \mu_\tau  (1-\rho_\epsilon (1- \dmax))  }.
 \end{align}}
 The terms $(\rmax, \rmin)$ represent the (worst and best) load factor of the $\t$-customers in the limit system,  when $\i$-customers are scheduled respectively with the 
best  (blocking probability $\dmin$) and  worst  (blocking probability $\dmax$)
   sub-policies, in $\t$-static manner.

Suppose that the constraint $C$ on the expected $\t$-number satisfies
$C \ge \rmax / (1-\rmax)$ (observe $\rmax / (1-\rmax)$ is the expected number
in M/M/1 queue with maximum load factor $\rmax$);
then the problem (\ref{Eqn_Pareto_optimal}) becomes an unconstrained
problem and the optimal policy clearly equals $\phi^*= (\dmin, \dmin,
\cdots)$.  We show that the optimal policy for any given $C < \rmax /
(1-\rmax)$, is monotone (but not strictly monotone) in $\tau$-state and further derive its closed
form expression (proof in Appendix~\ref{Appendix_pareto}):
\begin{thm}
\label{Thm_Opt_Policy} 
The  policy $\phi^* = \{d_0^*, d_1^* ,  \cdots\}$ that optimizes the   problem defined in (\ref{Eqn_Pareto_optimal})  is monotone and is given by:

\kt{
\begin{align}
\label{Eqn_pareto_optimal_policy}
d_i^* 
=  \indicator{i < L^*} \dmin + \indicator{i = L^*} d^* + \indicator{i > L^*} \dmax
\end{align}}which is parametrized by two parameters $(L^*, d^*)$. The  expressions for $(L^*, d^*)$ are given by 
~(see (\ref{Eqn_rhomax_min})):

\vspace{-6mm}
{\small \begin{align*}
L^*  
=&  \sup \left  \{ i :  \frac {  \sum_{j \le i}  \rmax^j j +    \rmax^i    \sum_{j > i} j \rmin^{j-i} }{
\sum_{j \le i} \rmax^j  +    \rmax^i    \sum_{j > i}  \rmin^{j-i}
} < C \right  \},   \\
 d^* =& \frac{1}{\mu_\tau \rho_\epsilon}\bigg[ \frac{- \lambda_\tau   \rmax^{L^*}    \left  (L^*+1 - C +   \sum_{j > L^*+1} (j -C )\rmin^{j-L^*-1}  \right )}{ \sum_{j \le L^*}  \rmax^j ( j - C) }  - \mu_\tau(1 - \rho_\epsilon) \bigg]  \indicator{L^* < \infty } + \ \     \underline{d} \indicator{L^* = \infty }.
\end{align*}}In the above $L^* \ge 1$,  is set to  $\infty$ when the inequality defining the $\sup$ is satisfied for all $i$ (i.e., when $C \ge \rmax /
(1-\rmax)$). 

\end{thm}
An immediate consequence of the above theorem is the following:
\begin{cor}[{\bf Pareto-Complete family}]
	The family of schedulers given by (\ref{Eqn_pareto_optimal_policy}), parametrized by $(L^*, d^*)$ with $1 \le L^* \le \infty $ and
$\dmin \le d^* \le \dmax$, 
is  Pareto-complete.  
\end{cor}

\extrabits{
\textcolor{red}{
\beq
\rho^* =   
   \frac {  \sum_{j \le L^*}  \rmax^j ( j - C) \ \   \indc_{L^* < \infty }  }{
   -  \rmax^{L^*}    \left  ( L^* +1 - C +   \sum_{j > L^*+1} (j -C )\rmin^{j-L^*-1}  \right )  }  + \rmax \indc_{L^* = \infty }
\eeq 
From  (\eqref{Eqn_pi_under_phi}) we get:
\beq
\rho^* &=&   
   \frac {  \sum_{j \le L^*}  \rmax^j ( j - C) \ \   \indc_{L^* < \infty }  }{
   -  \rmax^{L^*}    \left  (L^* +1 - C +   \sum_{j > L^* +1} (j -C )\rmin^{j-L^* -1}  \right )  } = \frac{\lambda_\tau}{\mu_\tau(1 - \rho_\epsilon(1 - d^*))}\\
 & \Rightarrow & \mu_\tau(1 - \rho_\epsilon(1 - d^*)) = \frac{- \lambda_\tau   \rmax^{L^*}    \left  (L^*+1 - C +   \sum_{j > L^* +1} (j -C )\rmin^{j-L^*-1}  \right )}{ \sum_{j \le L^*}  \rmax^j ( j - C) \ \   \indc_{L^* < \infty }} \\
& \Rightarrow & \mu_\tau(1 - \rho_\epsilon) + \mu_\tau \rho_\epsilon d^* = \frac{- \lambda_\tau   \rmax^{L^*}    \left  (L^* +1 - C +   \sum_{j > L^* +1} (j -C )\rmin^{j-L^*-1}  \right )}{ \sum_{j \le L^*}  \rmax^j ( j - C) \ \   \indc_{L^* < \infty }} \\
& \Rightarrow &  \mu_\tau \rho_\epsilon d^* = \bigg[ \frac{- \lambda_\tau   \rmax^{L^*}   \left  (L^*+1 - C +   \sum_{j > L^*+1} (j -C )\rmin^{j-L^*-1}  \right )}{ \sum_{j \le L^*}  \rmax^j ( j - C) \ \   \indc_{L^* < \infty }} - \mu_\tau(1 - \rho_\epsilon) \bigg] \\
& \Rightarrow &   d^* = \frac{1}{\mu_\tau \rho_\epsilon}\bigg[ \frac{- \lambda_\tau   \rmax^{L^*}    \left  (L^*+1 - C +   \sum_{j > L^*+1} (j -C )\rmin^{j-L^*-1}  \right )}{ \sum_{j \le L^*}  \rmax^j ( j - C) \ \   \indc_{L^* < \infty }} - \mu_\tau(1 - \rho_\epsilon) \bigg].
\eeq
}
}
\extrabits{
\subsection*{Important sub family (One threshold policies)}
Consider a sub-family of schedulers parametrized by $L$, that choose  $\epsilon$-scheduling sub-policy with maximum blocking (i.e., $d=1$) when the $\tau$-number is greater than  $L$ and choose the sub-policy with minimum blocking (i.e., $d = \dmin$) for the rest. Define:
 $$
 \rmax := \frac {\lambda_\tau }{ \mu_\tau  (1-\rho_\epsilon  (1- \dmin)) } \mbox{  and }  \rmin = \frac { \lambda_\tau } { \mu_\tau  (1-\rho_\epsilon)  }.$$
 By Theorem \ref{Thm_Opt_Policy}, this sub-family is a Pareto-optimal family and one can easily compute the performance under these policies:  

\vspace{-4mm}
{\small\begin{align}
EN &=& \frac { 1 }
{
\psi
} \left ( \rmax  \frac{1-\rmax^L}{(1-\rmax)^2}   - \frac{ L \rmax^{L} }{ 1- \rmax } + \frac{ \rmin^2 \rmax^{L-1} }{ (1- \rmin)^2 } + \frac{ L \rmax^{L-1} \rmin }{ 1- \rmin }
\right )
  \nonumber \\
&=& \frac { 1   }
{
\psi 
} \left ( \rmax  \frac{1-\rmax^L}{(1-\rmax)^2}   + \frac{ L \rmax^{L-1} (\rmax-\rmin) }{ (1- \rmax) (1-\rmin) } + \frac{ \rmin^2 \rmax^{L-1} }{ (1- \rmin)^2 }  \right )
\nonumber
\\
PB &=& \frac{1}
{
\psi
} \left( \frac {
 \rmax^{L-1} \rmin}{1 -\rmin} \dmax  + \frac{\rmax^{L}-1  }{\rmax- 1 }  \dmin  \right )  \mbox{ with }\\
 \psi &:=&\frac{\rmax^L - 1}{\rmax- 1}  + \frac{\rmax^{L-1} \rmin}{1 -\rmin} .
\end{align}}
}

The policies in this family choose the
`worst' $\epsilon$-sub-policy (i.e., with $d=\dmax$) when the
$\tau$-number is greater than or equal to $L+1$, choose a sub-policy
with intermediate blocking $d$ when $\tau$-number equals $L$ and
choose the `best' sub-policy (i.e., with $d = \dmin$) for the rest
(see (\ref{Eqn_pareto_optimal_policy})).  One can easily compute the
performance under these policies, as below (see \ref{Eqn_rhomax_min}):
{\begin{eqnarray}
\label{Eqn_pareto_perf}
E_{(L,d)}[N] &=& \psi \bigg ( \rmax \frac{1-\rmax^L}{(1-\rmax)^2} -
\frac{ L \rmax^{L - 1} }{ 1- \rmax } + \frac{\rho \rmax^{L-1} (\rmin +
  L - L \rmin) }{ (1- \rmin)^2 } \bigg ) \nonumber \\
  \nonumber \\
P_{B,(L,d)} &=&  \psi \left( d \rho \rmax^{L-1} + \rho \dmax  \frac {
 \rmax^{L-1} \rmin}{1 -\rmin}  + \frac{\rmax^{L}-1  }{\rmax- 1 }  \dmin    \right )   \mbox{ with  }  \nonumber  \\
&& \hspace{-19mm}
 \psi \ = \ \frac{1}
{
\frac{\rmax^L - 1}{\rmax- 1}+ \rho \rmax^{L-1}  + \rho  \frac{ \rmax^{L-1} \rmin  }{1 -\rmin}    
},  \  \   \rho =  \frac { \lambda_\tau } { \mu_\tau  (1-\rho_\epsilon (1- d))  } \hspace{2mm}.
\end{eqnarray}}
Thus {\it we derived Pareto complete family as well as the performance under this family, which can readily be used for any relevant optimization problem. }
\\

{\bf Numerical example:} We continue with the numerical example of
Figure \ref{Fig_SM_AR_PF}.  For this example, one can easily compute
that $\rmax  =  0.8134$ (no admission control on eager class, i.e., with $p_i = 1$ for all $i$) and $\rmin = 0.5$ (eager class is completely blocked with $\dmax = 1$ 
and hence $\rmin = \lambda_\tau/ \mu_\tau$), when the system can at maximum serve 5 eager customers in parallel.  By substituting these values
into (\ref{Eqn_pareto_perf}), one can obtain the Pareto frontier. The
circles in the figure represent this Pareto frontier, and are obtained
by varying $(L, d)$ appropriately.  It is clear from the figure that
the derived set of points are indeed dominating and are on the Pareto
frontier.

\extrabits{ 
\subsection*{Dynamic  Pseudo Conservation law}
 One can replace any pareto optimal  policy  (\ref{Eqn_pareto_optimal_policy})  with a mixed policy as below:
\begin{align}
\label{Eqn_pareto_alt_policy}
d_i^* 
= \left \{ \begin{array}{llll}
\dmin   &\mbox{ if }  i < L \\
(\dmin,  \dmax )  \mbox{ with probabilities } q \mbox{ and }      (1-q)         & \mbox{ if }   i = L \\
\dmax  & \mbox{ if }  i > L,
\end{array} \right .
\end{align}
where exact value of $q$  can be computed given $(L, d)$ or constraint $C.$ One can also derive the same Pareto frontier by considering a constraint on $P_B$ and then minimizing the  $E[N]$
$$
\min_{\phi}   E_\phi[N]  \mbox{ such   that  }  P_{B, \phi}  \le \eta \mbox{ for some }  0 <  \eta < \infty.
$$
 and  the same family of schedulers (given by (\ref{Eqn_pareto_alt_policy})) again optimize.  
 
  A close look at (\ref{Eqn_pareto_alt_policy})  indicates that any Pareto optimal policy schedules using maximum $\t$-static blocking probability for all $\t$-states above a threshold, while during the rest  minimum $\t$-static blocking probability  is used.  
  Thus we have the following\\
   \underline{Dynamic (Pareto) Pseduo conservation},  at SFJ limit:
{\it"  For any given system, specified by $(\rmax, \rmin, \rho_\i)$ the  maximum and minimum $\t$-load factors (derived  respectively  with the minimum  and the  maximum  $\t$-static blocking probabilities and the load factor of $\i$-class), 
   the time asymptotic fraction of  $\epsilon$-customers to be  blocked $P_B$,  completely determines the best performance of the tolerant class."  
The   best performance of the tolerant class equals that of the state dependent service rate M/M/1 system operating at either of the load factors,  $(\rmax, \rmin)$, and further such that best rate service is provided when the $\t$-strength is higher than a threshold.  The threshold is uniquely 
determined by $P_B$. }

 Thus $P_B$ alone determines the performance of the $\t$-class and no other details of the two scheduling policies have any influence on this relation. This is exactly as in  static Pseudo conservation law, but the difference between the two conservations is the placement of the fraction of $\i$-customers to be blocked: a) In static case, they are blocked independent of $\t$-state; and b) in dynamic case, they are blocked  only when the $\t$-queue is above a threshold.  Note here that the inherent blocking of $\i$-customers due to presence of fellow $\i$-customers, captured by the best blocking probability $\dmin$,  cannot be avoided. But the extra blocking to provide a controlled fair service to $\t$-customers can be placed towards the higher state values of $\t$-customers. 

 Given $P_{B, \phi}$, $E_\phi[N]$ is unique fixed point of the following equation:
 $$
 \pi_{\t, \ge  L} = P_B \mbox { and  }  \pi_\t = s.d. ( SDSR (\lambda_\t, \mu_\t,  \{\dmin, \dmax\},  \{0, \cdots, L\}, \ \{L, \cdots, \} ) ).
 $$
}

\subsection{Monte Carlo  based  case study in pre-limit}

We consider an example case-study with CD-$(p, K)$ sub-policies of
Subsection~\ref{sec:example_models}. Specifically, for fixed $K,$ we
consider top-level policies that perform CD-$(1,K)$ when the
$\t$-occupancy is less than certain $L $ (with $L\ge 1$) and perform
CD-$(0,K)$ (thus blocking all $\i$-jobs) when the $\t$-occupancy is
greater than or equal to $L.$ In view of Theorem~\ref{Thm_Opt_Policy},
by stepping over $L$ as above\footnote{ Here again, $\dmin$
  (respectively $\dmax$) equals the blocking probability without eager
  admission control (respectively if eager class is admitted only when
  $\t$-queue is empty).}, we sample performance vectors from the limit
Pareto-frontier of the system.


It is very complicated to obtain an exact analysis of this
heterogeneous system. However by Theorems \ref{Thm_tau}-\ref{Thm_PB},
one can obtain an approximate analysis for this system. \revision{In
  this section, we validate these approximations against Monte Carlo
  (MC) simulations of the actual system.} Importantly, {\it the Monte
  Carlo simulations do not even drop $\i$-customers at
  $\t$-transitions as required by {\bf A}.1}.

\begin{figure}[h]

\begin{minipage}{7cm}
\vspace{-34mm}
\begin{center}
\includegraphics[width=5.7cm, height =6cm]{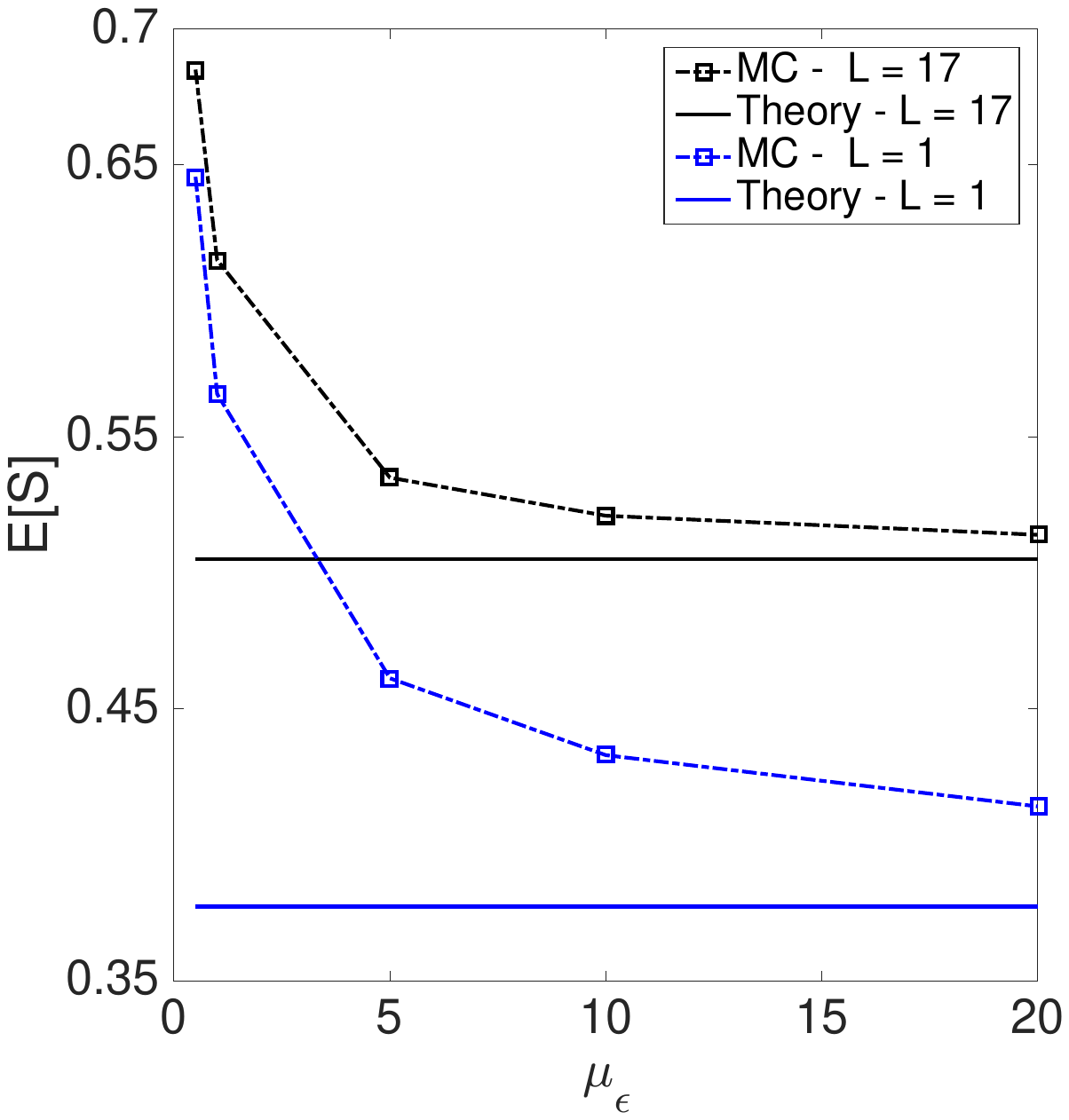}
\vspace{-76mm}
\end{center}
\end{minipage} \hspace{5mm}
\begin{minipage}{7cm}
\begin{center}

\includegraphics[width=5.7cm, height = 6cm]{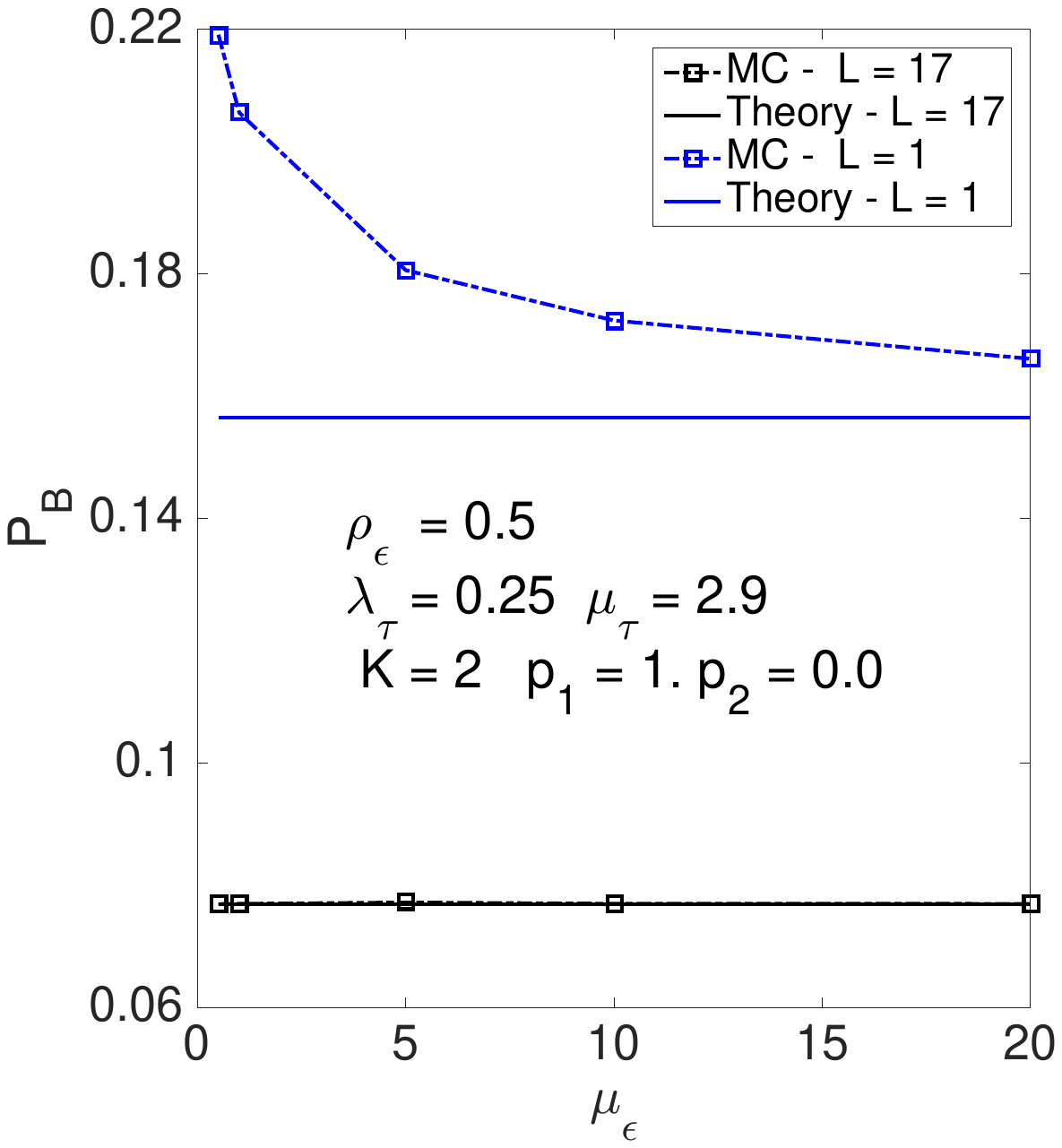}
\vspace{-76mm}

\vspace{-10mm}

\end{center}
\end{minipage}
\vspace{-10mm}
\caption{  $E[S]$, $P_B$ versus   $\mu_\epsilon$ for light traffic.  \label{Fig_PB_mue_light}}

\begin{minipage}{7cm}
\vspace{-34mm}
\begin{center}
\includegraphics[width=5.7cm, height =6cm]{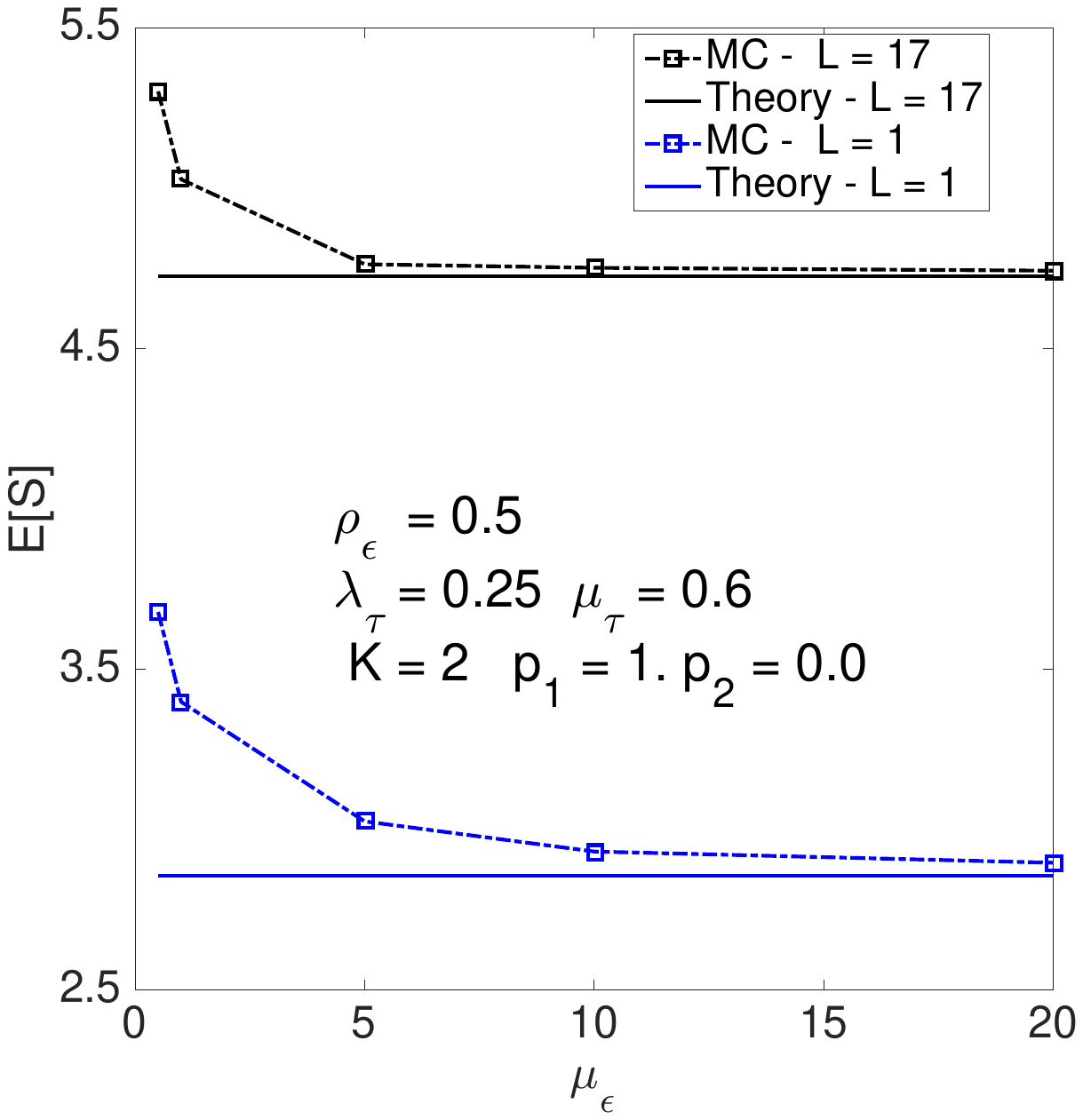}
\vspace{-76mm}
\end{center}
\end{minipage} \hspace{5mm}
\begin{minipage}{7cm}
\begin{center}

\includegraphics[width=5.7cm, height = 6cm]{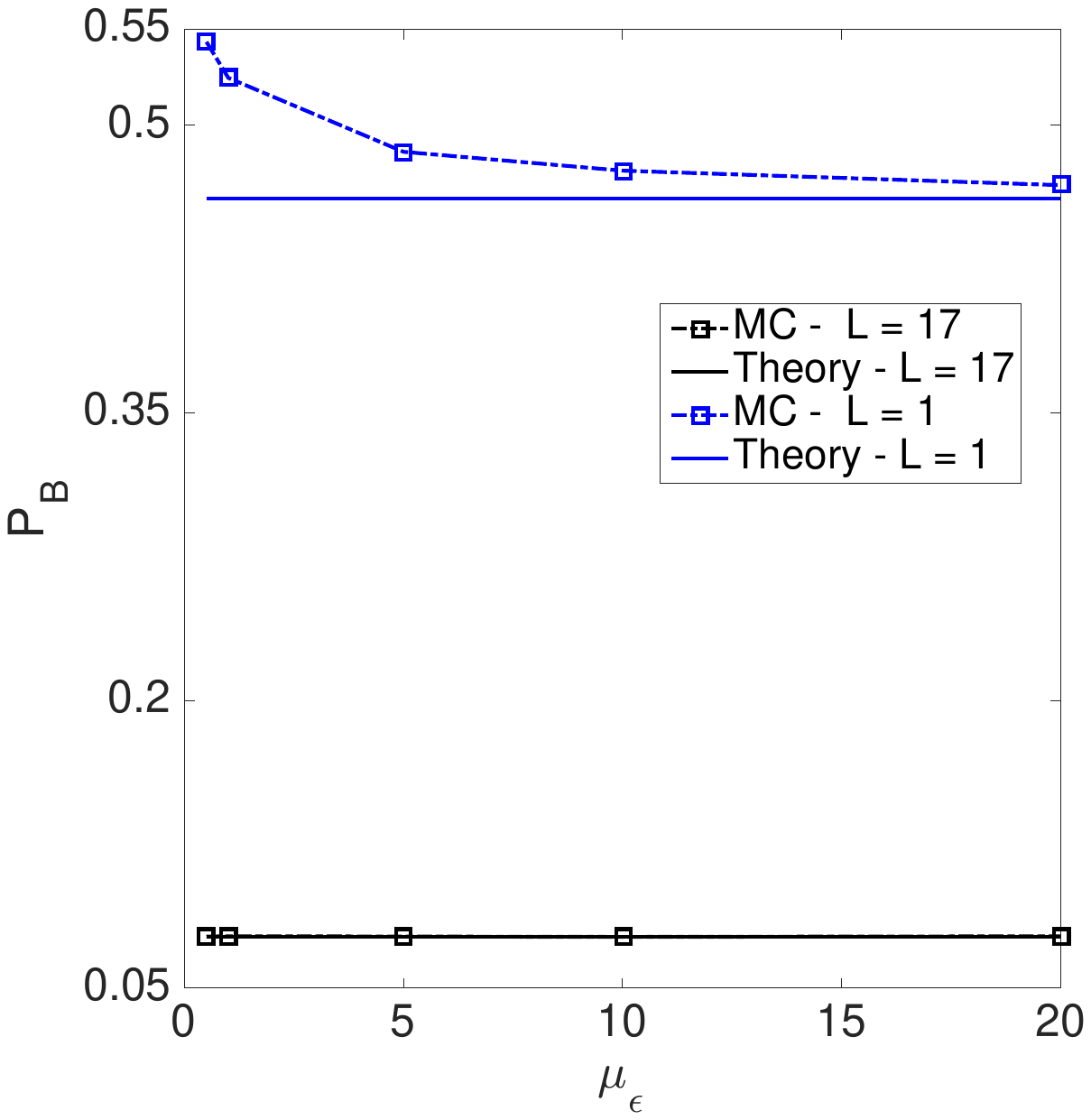}
\vspace{-76mm}
 
%
%
\vspace{-10mm}

\end{center}
\end{minipage}
\vspace{-10mm}
\caption{  $E[S]$, $P_B$ versus   $\mu_\epsilon$ for moderate traffic.  \label{Fig_PB_mue}}
\end{figure}

\revision{
In all the case studies presented in this section, we consider exponentially distributed job sizes  and Poisson arrivals for both the classes. The parameters used for a particular case study are described in the corresponding figure itself. 
}

\revision{
Our first case studies are presented in
Figures~\ref{Fig_PB_mue_light} and~\ref{Fig_PB_mue}.
We plot two performance measures, mean number of tolerant jobs, and
blocking probability of eager jobs, versus $\mu_\epsilon$ for two
different tolerant load factors; Figure~\ref{Fig_PB_mue_light}
corresponds to a light tolerant load, while Figure~\ref{Fig_PB_mue}
corresponds to a moderate tolerant load. (The system parameters used
are mentioned in the figures directly.) Note that as expected, the
simulated system performance metrics approach their SFJ approximations
as $\mue$ increases. Interestingly, our approximations tend to
\emph{under-estimate} both metrics (perhaps because the partial fluid
limit we consider `washes away' the stochasticity of the eager
workload). Moreover, note that in the case of moderate tolerant load,
the approximation is quite accurate for all values of $\mue \geq 1$;
the normalized difference between the theoretical approximation and
the corresponding MC estimate is within $6\%$ for $\mue \ge 1$ with
$L=17$ and for $\mue \ge 5$ with $L=1,$ and the error is negligible
for higher values of $\mue$ (see Figure~\ref{Fig_PB_mue}).  On the
other hand, with light tolerant traffic ($\rho_\tau = 0.0862$) as in
Figure \ref{Fig_PB_mue_light}, the normalized difference is within
$6\%$ only for $\mue \ge 5$ for $L=17$ and only for $\mue \ge 20$ for
$L=1$. In general, we observe that the approximation error is smaller
when the system operates closer to its stability limit. But even for
the case with light traffic, the normalized difference is within
$10\%$ in most cases once $\mue \ge 10.$}

\revision{Next, we consider the Pareto frontier of performance
  vectors, and compare our SFJ approximation with the frontier
  obtained via MC simulations; see Figure~\ref{Fig_MC} .
  We observe that our analytical results corresponding to the SFJ
  limit provide a very accurate approximation of the Pareto frontier,
  even for $\mu_\i$ as small as 1, when $\mu_\t = 0.54$. We reiterate
  that the theory approximates the system performance well even when
  the system does not `flush' $\i$-customers at $\t$-transitions.}


\begin{figure}[h]

\begin{minipage}{7cm}
\vspace{-34mm}
\begin{center}
\includegraphics[width=6.8cm, height =5.7cm]{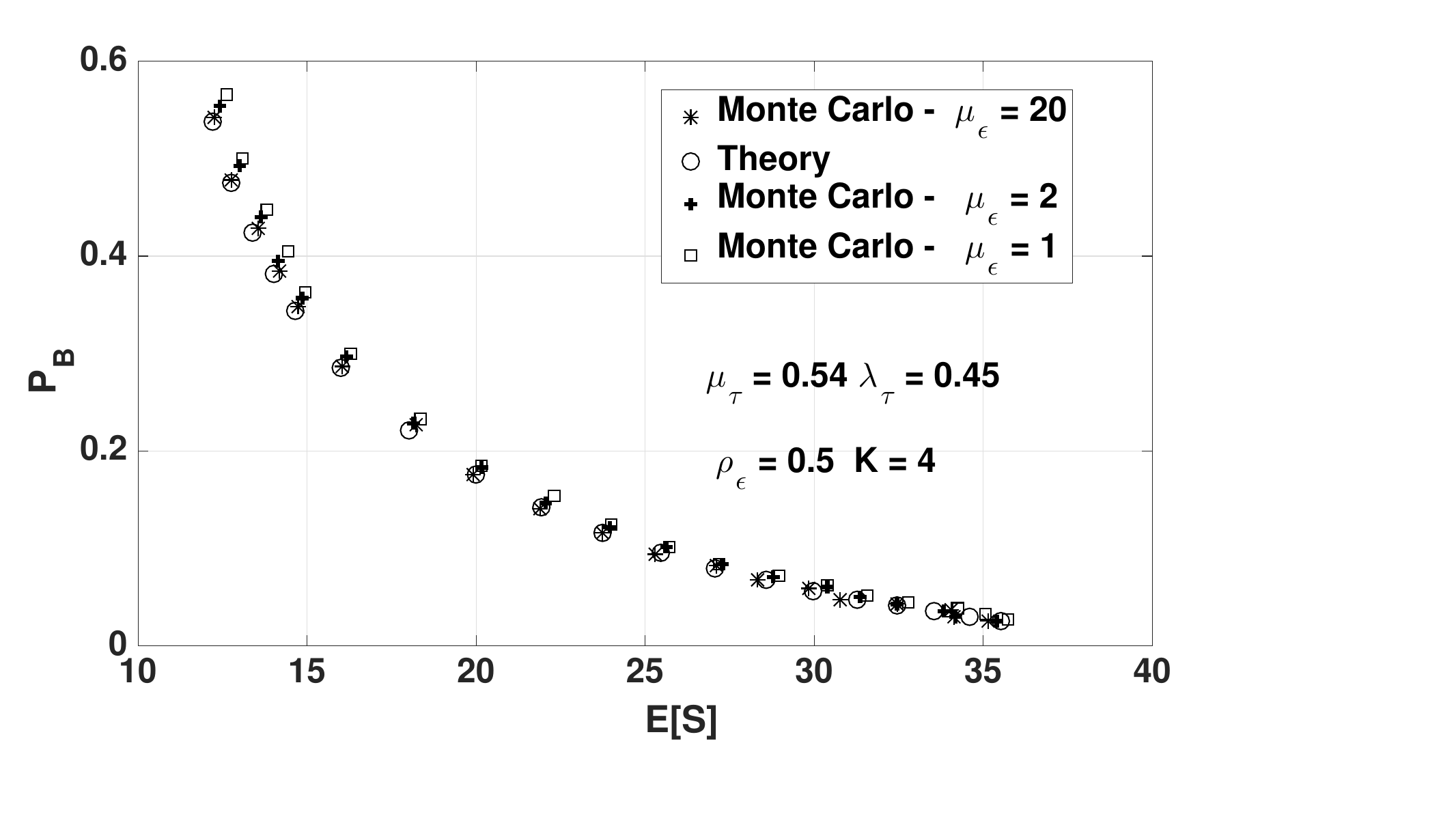}
\vspace{-44mm}
\caption{  Comparison of Theory with Monte Carlo (MC) estimates.  Good  approximation  even for  small   $\mu_\epsilon$.    \label{Fig_MC}}
\end{center}
\end{minipage} \hspace{5mm}
\begin{minipage}{7cm}
\begin{center}

\includegraphics[width=7.5cm, height = 5.6cm]{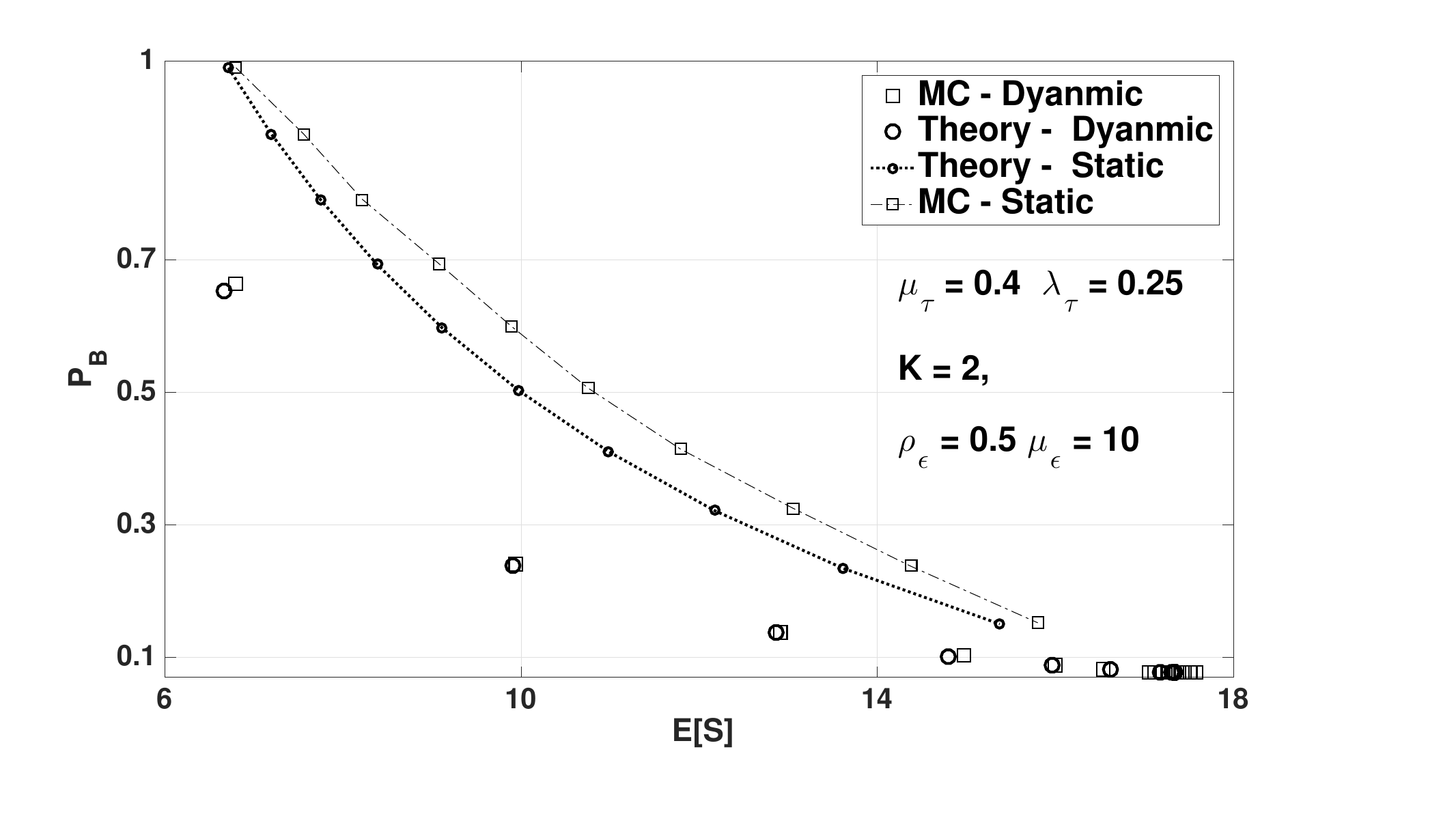}
\vspace{-46mm}
\caption{  Comparison of  static and dynamic policies. Better approximation for dynamic policies.  \label{Fig_static_dynamic}}

\vspace{-10mm}

\end{center}
\end{minipage}
\end{figure} 
 
 Another example is plotted in Figure \ref{Fig_static_dynamic}, where
 $\t$-static policies of \cite{ANOR} are considered along with dynamic
 policies.  We again observe a good approximation between the theory
 and MC estimates for dynamic policies. Interestingly, the
 approximation error is bigger in the static case. One possible
 explanation for this is the following. 
 It is clear that the approximation error gets smaller as the
 $\i$-load factor reduces, under $\t$-static policies.  Under Pareto
 optimal family of schedulers, the $\i$-load equals 0 for all
 $\t$-states greater than $L$.  Thus we see the approximation is
 almost zero towards the right of the two figures (as $P_B$ gets
 smaller, $L$ gets smaller).  We also observe that the dynamic
 policies perform far superior than $\t$-static policies.

\section{Concluding remarks}
\label{sec:conclusion}

In this paper, we analyse a multi-class, single server queueing system
with an eager (lossy) class and a tolerant (queueing) class, under
dynamic scheduling. While the inter-dependence between the service
processes of the two classes makes an exact analysis of this system
difficult, we obtain tractable performance approximations under a
certain (partial) fluid scaling regime. A key feature of our
approximations, proved to be accurate under the  fluid limit,  is a pseudo-conservation law:  the  approximate performance
of both classes is expressed in terms of the standalone blocking
probabilities of the eager schedulers, which are themselves easy to
compute in several cases. 
Further,  the accuracy of our approximations in the pre-limit is validated  via Monte Carlo
simulations.

Finally, we focus on the achievable region
of the limiting performance vectors for our system.
Remarkably, we are able to obtain an explicit family of Pareto-optimal
policies (these resemble  threshold  policies).

There are natural extensions of our results to models where the eager
class exhibits limited patience; for example, models with balking
and/or reneging. For example, the system might include (limited)
waiting room for eager customers.  Alternatively, eager customers
might respond to the resources allocated based on their patience
levels: a) an $\i$-customer may not enter the system depending upon
the $\i$-number already in system according to some probabilistic
rule, as in {\it balking} models; or b) may leave the system after
waiting for a random `patience time', as in {\it reneging} models. At
their core, our proofs require the property that the occupancy process
of the eager class gets time-scaled (fast-forwarded) with increasing
$\mue$ under the SFJ scaling. The analyses apply to the above
mentioned generalizations so long as this scaling property holds. For
example, in the case of reneging, we would require that the patience
time distribution is scaled suitably with $\mue,$ (as is discussed in
the context of static scheduling in \cite{ANOR}).

This work also motivates extensions in other directions. One
interesting extension would be to the multi-server setting, where the
tolerant class is no longer work conserving. Another promising
direction is to consider static/dynamic pricing for such heterogeneous
service systems. Finally, specializing our models to particular
application scenarios, including supermarkets, cognitive radio, and
cloud computing environments, would be of independent interest.

\ignore{
 
  \noindent {\bf Balking and Reneging:}  As a part of a
    particular $\i$-sub-policy, the system may allocate a certain
    number of servers (recall that the system may be viewed as
    multi-server from the standpoint of the eager class) and might
    allocate a certain amount of waiting space for $\i$-class.  The
    $\i$-customers might respond to the resources allocated based on
    their patience levels: a) an $\i$-customer may not enter the
    system depending upon the $\i$-number already in system according
    to some probabilistic rule, as in {\it balking} models; or b) may
    leave the system after waiting for an exponentially distributed
    patience time of rate $\alpha^{\mue},$ as in {\it reneging}
    models.  In the case of reneging, we will require that the
    parameter $\alpha^{\mue}$ scales linearly with $\mue$, i.e.,
    $\alpha^\mue = \alpha \mue$ for some $\alpha \in (0, \infty)$ (as
  in \cite{ANOR}).  We believe that our results can easily be extended to include these cases;  one can model super market based systems more realistically using these extensions  and this would be another interesting future direction.  
 }

\bibliographystyle{IEEEtran}
\bibliography{refs}

\begin{thebibliography}{10}
\providecommand{\url}[1]{#1}
\csname url@samestyle\endcsname
\providecommand{\newblock}{\relax}
\providecommand{\bibinfo}[2]{#2}
\providecommand{\BIBentrySTDinterwordspacing}{\spaceskip=0pt\relax}
\providecommand{\BIBentryALTinterwordstretchfactor}{4}
\providecommand{\BIBentryALTinterwordspacing}{\spaceskip=\fontdimen2\font plus
\BIBentryALTinterwordstretchfactor\fontdimen3\font minus
  \fontdimen4\font\relax}
\providecommand{\BIBforeignlanguage}[2]{{%
\expandafter\ifx\csname l@#1\endcsname\relax
\typeout{** WARNING: IEEEtran.bst: No hyphenation pattern has been}%
\typeout{** loaded for the language `#1'. Using the pattern for}%
\typeout{** the default language instead.}%
\else
\language=\csname l@#1\endcsname
\fi
#2}}
\providecommand{\BIBdecl}{\relax}
\BIBdecl

\bibitem{Chaudhary2019}
K.~Chaudhary, V.~Kavitha, and J.~Nair, ``Dynamic scheduling in a partially
  fluid, partially lossy queueing system,'' in \emph{Proceedings of WiOpt},
  2019.

\bibitem{Mahabhashyam05}
S.~R. Mahabhashyam and N.~Gautam, ``On queues with markov modulated service
  rates,'' \emph{Queueing Systems}, vol.~51, no. 1-2, pp. 89--113, 2005.

\bibitem{Value}
V.~Kavitha and R.~K. Sinha, ``Achievable region with impatient customers,'' in
  \emph{Proceedings of Valuetools}, 2017.

\bibitem{ANOR}
V.~Kavitha, J.~Nair, and R.~K. Sinha, ``Pseudo conservation for partially
  fluid, partially lossy queueing systems,'' \emph{Annals of Operations
  Research}, pp. 1--38, 2018.

\bibitem{CD_close}
A.~Sleptchenko, A.~van Harten, and M.~C. van~der Heijden, ``An exact analysis
  of the multi-class m/m/k priority queue with partial blocking,''
  \emph{Stochastic models}, vol.~19, no.~4, pp. 527--548, 2003.

\bibitem{Li04}
B.~Li, L.~Li, B.~Li, K.~M. Sivalingam, and X.-R. Cao, ``Call admission control
  for voice/data integrated cellular networks: performance analysis and
  comparative study,'' \emph{IEEE Journal on Selected Areas in Communications},
  vol.~22, no.~4, pp. 706--718, 2004.

\bibitem{Tang}
S.~Tang and W.~Li, ``A channel allocation model with preemptive priority for
  integrated voice/data mobile networks,'' in \emph{Proceedings of the First
  International Conference on Quality of Service in Heterogeneous
  Wired/Wireless Networks}, 2004.

\bibitem{Zhang}
Y.~Zhang, B.-H. Soong, and M.~Ma, ``A dynamic channel assignment scheme for
  voice/data integration in gprs networks,'' \emph{Computer Communications},
  vol.~29, no.~8, pp. 1163--1173, 2006.

\bibitem{Kleinrock1965}
L.~Kleinrock, ``A conservation law for a wide class of queueing disciplines,''
  \emph{Naval Research Logistics Quarterly}, vol.~12, no.~2, pp. 181--192,
  1965.

\bibitem{coffman}
E.~G. Coffman~Jr and I.~Mitrani, ``A characterization of waiting time
  performance realizable by single-server queues,'' \emph{Operations Research},
  vol.~28, no. 3-part-ii, pp. 810--821, 1980.

\bibitem{shanthikumar1992multiclass}
J.~G. Shanthikumar and D.~D. Yao, ``Multiclass queueing systems: Polymatroidal
  structure and optimal scheduling control,'' \emph{Operations Research},
  vol.~40, no. 3-supplement-2, pp. S293--S299, 1992.

\bibitem{Harchol-Balter}
M.~Harchol-Balter, \emph{Performance modeling and design of computer systems:
  {Q}ueueing theory in action}.\hskip 1em plus 0.5em minus 0.4em\relax
  Cambridge University Press, 2013.

\bibitem{Gallager2012}
R.~G. Gallager, \emph{Discrete stochastic processes}.\hskip 1em plus 0.5em
  minus 0.4em\relax Springer Science \& Business Media, 2012.

\bibitem{Hoel}
P.~G. Hoel, S.~C. Port, and C.~J. Stone, \emph{Introduction to stochastic
  processes}.\hskip 1em plus 0.5em minus 0.4em\relax Waveland Press, 1986.

\bibitem{Makowski11}
A.~M. Makowski, ``On a random sum formula for the busy period of the m jgj1
  queue with applications,'' 2001.
  
  \bibitem{Puter} Puterman, Martin L. "Markov decision processes." Handbooks in operations research and management science 2 (1990): 331-434.

\end{thebibliography}


\appendix

\begin{table}[]
	\begin{center}
		\begin{tabular}{|l|}
			\hline
			\\
			\centerline{\bf Basic Notations}
            \\	\hline \hline
			$\lame:$   Poisson arrival  rate of eager customer.
			\\ \hline
			$\lambda_\t :$ Poisson arrival rate of tolerant customer. 
			\\ \hline
			$B_\i :$  generic $\i$ job size. 
			\\ \hline 
			$B_\i^\mue \eqdist \frac{B_\i^1}{\mue} :$ $\i$ job-size at scale $\mue$.
			\\ \hline
			$\rho_{\i} :$  the load factor eager class.   
			\\ \hline 
			$A_\t :$ inter-arrival time of tolerant customers.
			\\ \hline
			$B_\t :$   Exponential  $\t$ job-size. 
		    \\ \hline
			$\nu_i :$ tolerant service rate when $i$-sup-policy used for eager. 
			\\ \hline

			\\ 
			\centerline{\bf Embedded chain}  \\
			\hline \hline
			$p_i^\mue$ : probability that tolerant arrival is before departure,  in  $\mue$-system at $\t$ occupancy $i$.
			\\ \hline
			 $q_i^\mue = 1- p_i^\mue$ : probability that $\t$-departure is before arrival.
		    \\ \hline
	     	$p_i^\infty$ : probability that $\t$-arrival is before departure, under limit system at $\t$ occupancy $i$.
	     	\\ \hline
			  $q_i^\infty = 1-  p_i^\infty$ :   probability that $\t$-departure is before arrival. in limit system. 
			 \\ \hline  
			  	$\tilde{\pi}^\mue_i: $ steady state probability  that embedded $\t$-chain is in state $i$, in $\mue$-system. 
			 \\ \hline
			  		$\tilde{\pi}_i : $ steady state probability  that embedded $\t$-chain is in state $i$, in limit system.
			  \\  \hline
			  \hline 
			  \\
			 \centerline{ \bf Continuous time process  }
			  \\
			  \hline \hline 
			$\pi^\mue_i  :$ steady state   probability  that  $\t$-process is in state $i$, in $\mue$ system. \\ \hline
			$\pi_i : $ steady state   probability  that  $\t$-process is in state $i$, in limit system.
		    \\ \hline
			$\Omega_i^\mue (\cdot):$ unused service process by $\i$ class. 
			\\ \hline
			$\rho_{i}  =  \rho^\phi_i   := \frac{\lambda_{\t}}{\mut \nu_i} : $ tolerant load when $i$-sub -policy ($\phi$)  used for eager. 
			\\ \hline 
			$\Upsilon_i^\mue : $ time require to finish $B_\t$ amount of work using $\Omega_i^\mue (\cdot)$. 
			\\ \hline
			$P_{B_i} :$ standalone $\t$ blocking probability when $i$ sub-policy used for eager class. 
			\\ \hline
			$E^\mue[N] :$ stationary expected number of $\t$-customer at $\mue$ system.
			\\ \hline
		    $P_B^\mue :$ overall $\i$ blocking probability at $\mue$ system. 
		    \\ \hline
		    $P_B^\infty :$ overall $\i$ blocking probability in limit system.
			\\ \hline
			$\B^\mue_i :$ $\i$-busy cycle when sub-policy $i$ is used.
			\\ \hline
			$\S^\mue_i :$ total service available for $\t$-class during $\B^\mue_i$.
			\\ \hline
			$\B^\mue_{i,k} :$ $k$-th $\i$-busy cycle when sup-policy $i$ is used. 
			\\ \hline
			$\S^\mue_{i,k} :$ total service available for $\t$-class during $\B^\mue_{i,k}$.
			\\ \hline
		\end{tabular}
		\caption{Summary of Notations   \label{Table_notations}}
	\end{center}
\end{table}

\section{Sample path coupling   and  SFJ limit}
\label{Appendix_sample}
\label{sec:SFJ_scale_details}

In this appendix, we detail the sample path coupling that results from
the SFJ limit, which is central to our arguments.

We consider a family of queuing systems,
parametrized by $\mue \geq 1$ (the service rate of the $\i$-class).
Under SFJ limit, the eager service rate $\mue \to \infty$, and the
eager arrival rate $\lambda_\i \ra \infty,$ while maintaining the
eager load $\rho_\i = \lambda_\i / \mu_\i$ to a constant value.
Specifically, we scale the job size distribution of the eager class
as, $B^{\mu_\i}_\i \eqdist B^1_\i/\mu_{\i}$, where $B^{\mu_\i}_\i$
denotes a generic eager job size at scale $\mu_\i$ and $\eqdist$ is
equality in distribution.

\revision{Consider now this family of queueing systems, immediately
  following an arrival/departure from the tolerant
  queue. Specifically, suppose that the arrival/departure resulted in
  $j$ jobs remaining in the tolerant queue. We couple the sample paths
  across different values of the scale parameter $\mue$ as
  follows. Let $A_{\i, k}^{\mu_\i}$ and $B_{\i, k}^{\mu_\i}$ are
  respectively the $k$th inter-arrival time and $k$th job size of the
  eager class under sub-policy~$j$ (for notational convenience, the
  dependence on $j$ is suppressed here). We relate these quantities
  sample path wise to the $\mu_\i = 1$ system as below:
\begin{equation}
  \label{eq:sample_path_job_ia_coupling}
  A_{\i, k}^{\mu_\i} = \frac { A_{\i, k}^{1} }{\mu_\i} \mbox{ and }
  B_{\i, k}^{\mu_\i} = \frac{ B_{\i, k}^{1} }{\mu_\i}.  
\end{equation}
Note that this coupling is consistent with the SFJ scaling:
$A_{\i, k}^{\mu_\i}$ is indeed exponentially distributed with mean
$1/\rho_{\i} \mue$ given that $A_{\i, k}^{1}$ is exponentially
distributed with mean $1/\rho_{\i}.$ Similarly, note that we have
$B_{\i, k}^{\mu_\i} \eqdist B_{\i}^{\mu_\i} \eqdist \frac{ B_{\i}^{1}
}{\mu_\i},$ as required.

An immediate consequence of \eqref{eq:sample_path_job_ia_coupling} is
that the occupancy process of the eager class at scale $\mue$ can be
viewed as a $\mue$-time-scaled (or fast-forwarded) version of
occupancy process at scale~1; see \cite{ANOR} for more details. In
particular, let $\B_{j,k}^\mue$ and $\S_{j,k}^\mue$ denote the length
of the $k$th $\i$-busy cycle (the interval between the start of two
successive $\i$-busy periods), and total service available for the
$\t$-class during the $k$th $\i$-busy cycle, respectively, when
tolerant queue occupancy equals $j.$ It now follows that
\begin{eqnarray}
  \B^\mue_{j,k} = \B_{j,k}^1 /\mue  \label{Eqn_Bmue}, \\
  \S^\mue_{j,k} = \S_{j,k}^1 / \mue  \label{Eqn_Smue}. 
\end{eqnarray}
}

\ignore{Now, holding the occupancy of the tolerant queue frozen
  at~$j,$ via sample path coupling, the occupancy process of the eager
  class at scale $\mue$ can be viewed as a $\mue$-time-scaled (or
  fast-forwarded) version of occupancy process at scale~1; see
  \cite{ANOR} for more details.  This type of sample path-wise
  coupling allows us to derive almost sure results. With this
  construction, it is easy to see for any sub-policy $j$ (when both
  processes are not interrupted by $\t$-transitions due to assumption
  {\bf A}.3) that
  $$
  \B^\mue_j  =  \B_j^1 /\mue  \mbox { and }   \S^\mue_j  =  \S_j^1 / \mue   \mbox{ almost surely. }
  $$ 
  In the following, for simplicity, we also denote $\B_j^1$ and
  $\S_j^1$ as $\B_j$ and $\S_j,$ respectively.}

\revision{
We conclude by describing another consequence of the sample path
coupling \eqref{eq:sample_path_job_ia_coupling} described above. Let $N_A^\mue ([a,b])$ represent the number
of $\i$-arrivals in time interval $[a,b]$ and let
$N_{B_j}^\mue( [a,b])$ represent the number of drops (or the
$\i$-customers that exit the system without service) in the interval
$[a,b],$ when sub-policy $j$ is used in a $\t$-static manner, at
scale~$\mu_\i$. More compactly, let $N_{B_j}^{\mue}(t)$, $N_A^\mue(t)$
represent $N_{B_j}^{\mue}([0,t])$ and $N_A^\mue([0, t]),$
respectively. Then, by the way of construction:
\begin{eqnarray}
\label{Eqn_Compare_mu_epsilon_schemes}
  N_A^\mue ({\mathcal B}_{j,1}^\mue  ) = N_A^1 ( {\mathcal B}_{j,1}^1 )   \   \mbox{and}   \
  N_{B_j}^\mue ({\mathcal B}_{j,1}^\mue  ) = N_{B_j}^1 ( {\mathcal B}_{j,1}^1 ),
\end{eqnarray}
and a similar equality holds for subsequent $\i$-busy cycles as well.
These kind of consequences play a central role in deriving the results
of the next two appendices.}

\section{Tolerant Performance}  
\label{Appendix_tolerant}

In this appendix, we provide proofs of the results stated in
Section~\ref{sec:tolerant}, corresponding to the performance
evaluation of the tolerant class. 

\ignore{
We begin with some remarks on a
sample path coupling that results from the SFJ limit, which is central
to our arguments.

Under the SFJ scaling, we consider a family of queuing systems,
parameterized by $\mue \geq 1$ (the service rate of the $\i$-class).
Under SFJ limit, the eager service rate $\mue \to \infty$, and the
eager arrival rate $\lambda_\i \ra \infty,$ while maintaining the
eager load $\rho_\i = \lambda_\i / \mu_\i$ to a constant value.
Specifically, we scale the job size distribution of the eager class
as, $B^{\mu_\i}_\i \eqdist B^1_\i/\mu_{\i}$, where $B^{\mu_\i}_\i$
denotes a generic eager job size at scale $\mu_\i$ and $\eqdist$ is
equality in distribution. Now, holding the occupancy of the tolerant
queue frozen at~$j,$ via sample path coupling, the occupancy process
of the eager class at scale $\mue$ can be viewed as a
$\mue$-time-scaled (or fast-forwarded) version of occupancy process at
scale~1; see \cite{ANOR} for more details. Let $\B_j^\mue$ and
$\S_j^\mue$ denote a generic $\i$-busy cycle (the interval between the
start of two successive busy periods), and total service available for
the $\t$-class during an $\i$-busy cycle, respectively, when the
tolerant occupancy equals $j.$ It thus follows that
\begin{eqnarray}
  \B^\mue_j \eqdist \B_j^1 /\mue  \label{Eqn_Bmue}, \\
  \S^\mue_j \eqdist \S_j^1 / \mue  \label{Eqn_Smue}. 
\end{eqnarray}
In the following, for simplicity, we also denote $\B_j^1$ and $\S_j^1$
as $\B_j$ and $\S_j,$ respectively. }

This appendix is organized as follows. In
Section~\ref{sec:busycycles}, we prove a technical result that is used
to prove Lemma~\ref{lemma:uniform}. The proofs of
Lemmas~\ref{lemma:uniform} and~\ref{lemma:TauStability} are provided
in Sections~\ref{sec:lemma:uniform} and~\ref{sec:lemma:TauStability},
respectively. The proofs of
Theorems~\ref{Thm_SD_Convergence},~\ref{Thm_Lim_SD_Conv}
and~\ref{Thm_tau} are provided in
Sections~\ref{sec:Thm_SD_Convergence},~\ref{sec:Thm_Lim_SD_Conv},
and~\ref{sec:Thm_tau}, respectively.

\subsection{Bounding the second moment of $\i$-busy cycles}
\label{sec:busycycles}
To prove Lemma~\ref{lemma:uniform}, we need the following technical
lemma, which states that second moment of the the eager busy cycles
$\B_j,$ corresponding to $\t$-occupancy $j$ and scale $\mue = 1,$ are
uniformly bounded from above over sub-policies $j.$

\begin{lemma}
  \label{lemma:bp_upper_bound}
  Under Assumptions {\bf A.}1-3,  Assumption {\bf A}.4 is true, i.e., there exists a constant $\mathcal{M}$
  such that $\Exp{\B_j^2} \leq \mathcal{M}$ for any~$\i$
  sub-policy~$j.$
\end{lemma}
\begin{proof}
  Note that busy cycle $\B_j\eqdist O_j + I_j,$ where $O_j$ denotes the
  \emph{busy period}, and $I_j$ the idle period between two successive
  busy periods. Moreover, $I_j$ is exponentially distributed with mean
  $1/\lambda_\i.$ Thus, to prove the statement of the lemma, it
  suffices to prove that $\Exp{O_j^2}$ is bounded from above uniformly
  over sub-policies~$j.$ This is proved as follows. The busy period
  $O_j$ under any sub-policy~$j$ can be upper bounded by the busy
  period of an $M/G/\infty$ system with arrival rate~$\lambda_\i$ and
  job size $X :\eqdist B_\i/c_{\min}.$ Indeed, note that since any
  admitted eager job of size $B_\i$ must be served at a minimum rate
  of $c_{\min},$ $X = B_\i/c_{\min}$ is an upper bound on the
  residence time of the job in the system. The statement of the lemma
  now follows from Lemma~\ref{lemma:mginftyBPbound} below, which
  provides an upper bound on the second moment of an $M/G/\infty$ busy
  period.
\end{proof}
 
\begin{lemma}
  \label{lemma:mginftyBPbound}
  Consider an $M/G/\infty$ queue with arrival rate $\lambda,$ with a
  generic job size denoted by $X$. Let $O$ denote a generic busy
  period of this system. Then
  $$\Exp{O^2} \leq e^{\theta} \Exp{X^2} \left(e^{\theta} - 1 + \frac{e^2-1}{2}
  \right),$$ where $\theta := \lambda \Exp{X}.$
\end{lemma}

\begin{proof}
  Given a non-negative random variable $Y$ having finite expectation,
  let $Y^*$ denote its \emph{forward excess}, i.e., $Y^*$ is a
  non-negative random variable satisfying
  $$\prob{Y^* > x} = \frac{1}{\Exp{Y}} \int_x^{\infty} \prob{Y > t} dt
  \qquad (x \geq 0).$$ Note that
  $\Exp{Y^*} = \frac{\Exp{Y^2}}{2 \Exp{Y}}.$

  Now, returning to the $M/G/\infty$ queue, let $K:= 1-e^{-\theta}.$
Define a  non-negative random variable $U$, which is distributed as below:
$$\prob{U > x} = \frac{1}{K} e^{-\theta}(e^{\theta \prob{X^* > x}} - 1)
\qquad (x \geq 0).$$ Finally, let $M$ denote a geometric random
variable such that $$\prob{M = n} = (1-K) K^{n-1} \qquad (n \in \N.)$$
The following representation for the forward excess of the
$M/G/\infty$ busy period was established by Makowski \cite{Makowski11}:
\begin{equation}
  \label{eq:Makowski_result}
  O^* \eqdist \sum_{i = 1}^M U_i,
\end{equation}
where $\{U_i\}_{i \geq 1}$ is an i.i.d. sequence of random variables
distributed as $U$ independent of $M.$

From \eqref{eq:Makowski_result}, we have
$$\Exp{O^*} = \frac{\Exp{O^2}}{2 \Exp{O}} = \Exp{M} \Exp{U},$$ which
yields
\begin{align}
  \Exp{O^2} &= 2 \Exp{O} \Exp{M} \Exp{U} \nonumber \\
             &= 2 \frac{K}{\lambda(1-K)} \frac{1}{(1-K)} \Exp{U}.
               \label{eq:temp1}
\end{align}
Here, we have used $\Exp{O} = \frac{K}{\lambda(1-K)},$ which is also
proved in \cite{Makowski11}.

We upper bound $\Exp{U}$ as follows. Using the Markov inequality,
$\prob{X^* > x} \leq \frac{\Exp{X^2}}{x \Exp{X}}.$ It then follows
that $\theta \prob{X^* > x} \leq 2$ for
$x \geq \beta := \frac{\lambda \Exp{X^2}}{2}.$
Now,
\begin{align}
  \Exp{U} &= \frac{e^{-\theta}}{K} \int_{0}^{\infty} (e^{\theta \prob{X^* > x}} - 1) dx \nonumber \\ 
          &\leq \frac{e^{-\theta}}{K} \left[\beta (e^{\theta} - 1) + \int_{\beta}^{\infty} (e^{\theta \prob{X^* > x}} - 1) dx\right] \nonumber \\
          &\stackrel{(a)}\leq \frac{e^{-\theta}}{K} \left[\beta (e^{\theta} - 1) + \int_{\beta}^{\infty} \gamma \theta \prob{X^* > x} dx\right] \nonumber \\
          &\leq \frac{e^{-\theta}}{K} \left[\beta (e^{\theta} - 1) + \gamma \theta \frac{\Exp{X^2}}{2\Exp{X}} \right]\nonumber \\
          &= \frac{e^{-\theta} \lambda \Exp{X^2}}{2K} \left[e^{\theta} - 1 + \gamma \right].
            \label{eq:temp2}
\end{align}
The bounding in $(a)$ uses the inequality $e^{x} - 1 \leq \gamma x$
for $x \in [0,2],$ with $\gamma := \frac{e^2-1}{2}$ (by convexity of $e^x$). Finally,
combining \eqref{eq:temp1} and \eqref{eq:temp2} gives us the statement
of the lemma.
\end{proof}

\subsection{\bf Proof of Lemma~\ref{lemma:uniform}}
\label{sec:lemma:uniform}

Recall that, when $\t$-occupancy equals $j$ (for any given $j$), the
backward transition probability of the tolerant birth-death chain (at
scale $\mue$) is given by $q_j^\mue$.  We aim to show the convergence
of these transition probabilities, uniformly over $j$, under the SFJ
limit.
The idea is to develop (for each $j$) a lower and an upper bound on
$q_j^\mue$ which uniformly converge to the same constant given by
Equation~\eqref{eq:q_conv}.

\revision{Consider the evolution of the tolerant queue, starting from
  an instant when the $\t$ occupancy changed, via an arrival or a
  departure, to~$j.$ Let $A_\t$ denote the time until the next $\t$
  arrival, let $\B_{j,k}^{\mue}$ denote the $k$th $\i$-busy cycle and
  let $\S_{j,k}^{\mue}$ denote the total amount of service available
  for the tolerant class during the $k$th $\i$-busy cycle, at scale
  $\mue$
}

\noindent {\bf Lower bound:} 
Consider any eager sub-policy $j$. Define
$N := \min\{n \geq 1\ |\ \sum_{k = 1}^n \B_{j,k}^{\mue} > A_\t\}.$
\revision{Note that $N$ is the index of the $\i$-busy cycle in which
  the next $\t$ arrival occurs. Since $A_\t$ is exponentially
  distributed, $N$ is a geometric random variable with parameter
  $\prob{\B_j^{\mue} > A_\t}.$} Further, define
$\bar{q}_j := \prob{\sum_{k = 1}^{N-1} \S_{j,k}^{\mue} > B_{\t}}.$
\revision{Note that $\sum_{k = 1}^{N-1} \S_{j,k}^{\mue},$ which is the
  service received by the~$\t$~queue in the first $N-1$ $\i$-busy
  cycles, is a lower bound on the total service received by the~$\t$
  queue until the next~$\t$-arrival. Thus, the event
  $\sum_{k = 1}^{N-1} \S_{j,k}^{\mue} > B_{\t}$ implies that a $\t$-
  departure would precede the next $\t$-arrival; this yields the bound
  $q^{\mue}_j \geq \bar{q}_j.$} Moreover, the lower bound $\bar{q}_j$,
on $q_j^\mue$ can be represented as follows:\footnote{Let
  ${\bar q}_j :=E({\cal E})$, with indicator of the event
  ${\cal E} := 1_{ \sum_{k = 1}^{N-1} \S_{j,k}^{\mue} > B_{\t} }
  $. Condition on the events of first $\i$-busy cycle:
		\begin{itemize}
		\item 	 if $\S_{j,1}^{\mue} > B_\t$ and $\ \B_{j,1}^{\mue} < A_\t$ then ${\cal E} = 1$; otherwise 
		\item   if $\ \B_{j,1}^{\mue} > A_\t$  then ${\cal E} = 0$;  or  
		\item if $\ \B_{j,1}^{\mue} < A_\t$ and
                  $\S_{j,1}^{\mue} < B_\t$, \revision{then the
                    conditional probability of ${\cal E}$ equals its
                    unconditional probability by the memorylessness of
                    $A_\t$ and $B_\t$.}
\end{itemize}			
		And thus, 
		$$
		\bar{q}_j = \prob{\S^{\mue}_{j,1} >
			B_\t,\ \B^{\mue}_{j,1} < A_\t} + \bar{q}_j  P( \B_{j,1}^{\mue}  < A_\t ,  \  \S_{j,1}^{\mue} < B_\t ),
		$$
		which implies \eqref{eq:p_bar} given that
		$$1 - P( \B_{j,1}^{\mue}  < A_\t ,  \  \S_{j,1}^{\mue} < B_\t )  = \prob{\S^{\mue}_{j,1} >
		B_\t,\ \B^{\mue}_{j,1} < A_\t} + \prob{\B^{\mue}_{j,1} > A_\t}.$$}
\begin{equation}
\label{eq:p_bar}
\bar{q}_j =
\frac{\prob{\S^{\mue}_j > B_\t,\ \B^{\mue}_j < A_\t}}{\prob{\S^{\mue}_j >
		B_\t,\ \B^{\mue}_j < A_\t} + \prob{\B^{\mue}_j > A_\t}}.    
\end{equation}
By Lemma \ref{Lemma_Intermediate_result} below, we have uniform
convergence (here, $\varpi(\mue) \to 0$ as $\mue \to \infty$
\textit{u.f.} in $j$):
%
%
\begin{eqnarray*}
 \left|\prob{\B^{\mue}_j > A_\t} - \frac{\lambda_\t
  \Exp{\B_j}}{\mue} \right| \leq \varpi(\mue), \ \ \mbox{and} \ \  \left|\prob{\S^{\mue}_j > B_\t,\ \B^{\mue}_j < A_\t} -
  \frac{\mu_\t \Exp{\S_j}}{\mue} \right| \leq \varpi(\mue),
\end{eqnarray*} 
where $\B_j$ is a typical $\i$-busy cycle and $\S_j$ is the (typical)
unused service left behind by the eager class in one $\i$-busy cycle,
both with $\mue = 1$ and under sub-policy $j$.
Using the above and because\footnote{When renewal reward theorem is
  applied with renewal cycles as $\i$-busy cycles and with rewards
  being the amount of service available to the $\t$-class in each busy
  cycle. }  $\nu_j = \frac{\Exp{\S_j}}{\Exp{\B_j}}$, we get from
equation \eqref{eq:p_bar} that the lower bound on $q_j^\mue$,
\begin{eqnarray*}
\bar{q}_j \stackrel{\mue \to \infty}{\longrightarrow} \frac{\mu_\t \Exp{\S_j}}{\mu_\t \Exp{\S_j} + \lambda_\t
	\Exp{\B_j} } = \frac{\mut \nu_j}{\mut \nu_j + \lambda_\t} \ \ u.f. \mbox{ over} \ \ j.
\end{eqnarray*}

{\bf Upper bound:} Next, we will prove the \textit{u.f.} convergence
of an upper bound on the transition probabilities $q^{\mue}_j$, to the
same limit. Since $\B^\mue_j \ge \S^\mue_j$, one can upper bound
$q_j^\mue$ as follows:
\begin{eqnarray*}
 q^{\mue}_j \leq \bar{q}_j + \prob{\S^{\mue}_j > B_\t\ |\ \B^{\mue}_j
	> A_\t} \leq \bar{q}_j + \prob{\B^{\mue}_j > B_\t\ |\ \B^{\mue}_j
	> A_\t} .
\end{eqnarray*}
In view of Lemma \ref{Lemma_Intermediate_result_1} below, we have  as $\mue \to \infty$,
\begin{eqnarray*}
\prob{\B^{\mue}_j > B_\t\ |\ \B^{\mue}_j> A_\t} \ra 0 \ \ u.f.  \mbox{ over } \  j.  
\end{eqnarray*}
Therefore the upper bound on $q_j^\mue$, 
\begin{eqnarray*}
  \bar{q}_j + \prob{\B^{\mue}_j > B_\t\ |\ \B^{\mue}_j
	> A_\t} \stackrel{\mue \to \infty}{\longrightarrow} \frac{\mut \nu_j}{\mut \nu_j + \lambda_\t} \ \ u.f. \mbox{ over} \ \ j.  
\end{eqnarray*} 
This proves that the lower and upper bound converge to the same limit
(for any given $j$), uniformly across all $j$. This proves the
statement of the lemma. \qed

\begin{lemma}
	\label{Lemma_Intermediate_result}
We have:  \\  
i) $ |\mue \prob{\B^{\mue}_j > A_\t} - \lambda_\t
\Exp{\B_j}| \leq \frac{\lambda_{\t}^2\Exp{\B^2_j}}{2\mue},$ and \\
ii)   
$|\mue \prob{\S^{\mue}_j > B_\t,\ \B^{\mue}_j < A_\t} - \mu_\t
\Exp{\S_j}| \leq \frac{(\lambda_\t+
  \mu_\t)^2\Exp{\B^2_j}}{\mue}.$ \\
That is, there exists a function $\varpi (\mue)$ (which is independent
of $j,$ and  is finite  by  Lemma~\ref{lemma:bp_upper_bound}), such that
$\varpi (\mue) \to 0$ as $\mue \to \infty$ and for each $j$,
\begin{eqnarray*}
  \left | \prob{\B_j^{\mue} > A_\t} - \frac{\lambda_\t \Exp{\B_j }}{\mue}  \right | 
  &\le  &  \varpi (\mue),  
          \mbox{  and  }  \\
  \left | 	  \prob{\S_j^{\mue} > B_\t,\
  \B^{\mue}_j < A_\t}  - \frac{  \mu_\t \Exp{\S_j} } {\mue} \right | 
  & \le&   \varpi (\mue) \mbox{ with }  \\ 
  && \hspace{-40mm}
     \varpi (\mue) \  :=  \  \max \left \{    \frac{\lambda_{\t}^2}{2\mue^2} ,    \  \       \frac{(\lambda_\t+
     \mu_\t)^2}{\mue^2}     \right \}  \sup_i  \Exp{\B_i^2}  =  \frac{(\lambda_\t+
     \mu_\t)^2}{\mue^2}      \sup_i  \Exp{\B_i^2}  .
\end{eqnarray*}
\end{lemma}

\begin{proof} 
  Proof of $(i)$: Note that, for any random variable $Y,$ the moment
  generating function (MGF) is given by: $M_Y(s) := \Exp{e^{sY}}.$ We
  use the following bounds in our proof repeatedly. Suppose that the
  random variable $X$ is exponentially distributed with mean
  $1/\gamma,$ and the random $Y$ is non-negative and independent of
  $X.$ Then
  \begin{equation}
    \label{eq:mgf_bounds}
    \gamma \Exp{Y} - \frac{\gamma^2 \Exp{Y^2}}{2} \leq \prob{Y > X} = 1-M_Y(-\gamma) \leq \gamma \Exp{Y}.
  \end{equation}
  The equality above follows by conditioning on $Y,$ and the
  inequalities are a consequence of the following bounds on the MGF:
  for $s \geq 0,$
  $$1-s \Exp{Y} \leq M_Y(-s) \leq 1-s \Exp{Y} + \frac{s^2
    \Exp{Y^2}}{2}.$$

  Now, since $A_\t$ is exponentially distributed and because  
  $\B_j^{\mu_\i} \eqdist \B_j^1/\mu_{\i} = \B_j /\mu_{\i}$  (see \eqref{Eqn_Bmue}), we
  have, using~\eqref{eq:mgf_bounds},
	\begin{align*}
          \prob{\B^{\mue}_j > A_\t} \leq
          \frac{\lambda_\t \Exp{\B_j}}{\mue}. 
	\end{align*}
	Using this inequality, and invoking \eqref{eq:mgf_bounds} again,   
        {\small\begin{align*} \bigg|\mue
            \prob{\B^{\mue}_j > A_\t} - \lambda_\t \Exp{\B_j}\bigg|
            =
	\lambda_\t \Exp{\B_j} - \mue \prob{\B^{\mue}_j > A_\t} 
	& \leq \lambda_\t \Exp{\B_j} - \mue
	\left[\frac{\lambda_\t \Exp{\B_j}}{\mue} -
	\frac{\lambda_{\t}^2 \Exp{\B^2_j}}{2 (\mue)^2} \right] \\
	&=\frac{\lambda_{\t}^2\Exp{\B^2_j}}{2\mue}.
               \end{align*}}
             This completes the proof of Part~$(i)$. \\
	
	\noindent
	{\it Proof of (ii):}
        Again, since $A_\t$ is exponentially distributed, using
        \eqref{eq:mgf_bounds}, we have
	\begin{align}
	\label{eq:cl2_1}
	\prob{\S^{\mue}_j > B_\t,\ \B^{\mue}_j < A_\t} & \leq \prob{\S^{\mue}_j
		> B_\t} \le  \frac{\mu_\t \Exp{\S_j}}{\mue}. 
	\end{align}
	 Since, $\B^\mue_j \ge  \S^\mue_j$, we will have
	\begin{align}
	\prob{\S^{\mue}_j > B_\t,\ \B^{\mue}_j < A_\t} & = \prob{\S^{\mue}_j >
		B_\t} - \prob{ \S^{\mue}_j > B_\t,\ \B^{\mue}_j > A_\t} \nonumber \\
	& \geq \prob{\S^{\mue}_j > B_\t} - \prob{\B^{\mue}_j > B_\t,\ \B^{\mue}_j > A_\t
		}. \label{eq:cl2_2}
	\end{align} Because  $A_\t$ and $B_\t$ are exponentially distributed, the second term of the above inequality becomes (using the lower and upper bounds on $M_{\B_j}$ as before):
	\begin{align}
	\prob{\B^{\mue}_j > A_\t,\ \B^{\mue}_j > B_\t} &=\int
	(1-e^{-\lambda_\t x})(1-e^{-\mu_\t x}) dF_{\B^{\mue}_j}(x) \nonumber \\
	&= 1 - M_{\B_j}\left(\frac{-\lambda_\t}{\mue}\right) -
	M_{\B_j}\left(\frac{-\mu_\t}{\mue}\right) +
	M_{\B_j}\left(\frac{-\lambda_\t - \mu_\t}{\mue}\right) \nonumber\\
&\le   1-  \left(1-\frac{\lambda_\t\Exp{\B_j}}{\mue} \right) 
	-  \left(1-\frac{\mu_\t\Exp{\B_j}}{\mue} \right )  \nonumber  \\
	&\quad + \left(1-\frac{(\lambda_\t+ \mu_\t)\Exp{\B_j}}{\mue}  + \frac{(\lambda_\t+ \mu_\t)^2\Exp{\B^2_j}}{2 (\mue)^2}\right) \nonumber \\	
	&=  \frac{(\lambda_\t+ \mu_\t)^2\Exp{\B^2_j}}{2 (\mue)^2} 
	 \label{eq:cl2_3}.
	\end{align}
	The result follows from \eqref{eq:cl2_1}, \eqref{eq:cl2_2}, and
	\eqref{eq:cl2_3}:
	\begin{align*}
	\bigg|\mue \prob{\S^{\mue}_j > B_\t,\ \B^{\mue}_j < A_\t} & - \mu_\t
	\Exp{\S_j}\bigg|  
	= \mu_\t \Exp{\S_j} - \mue \prob{\S^{\mue}_j > B_\t,\
		\B^{\mue}_j < A_\t} \\
	&\leq \mu_\t \Exp{\S_j} - \mue \prob{\S^{\mue}_j > B_\t} + \mue
	\prob{ \B^{\mue}_j > B_\t,\ \B^{\mue}_j > A_\t} \\
	&= \mu_\t \Exp{\S_j} - \mue   E\left [1-  e^{-\mu_\t \S_j  / \mue }  \right ]   + \mue
	\prob{\B^{\mue}_j > A_\t,\ \B^{\mue}_j > B_\t} \\
	&\leq \frac{\mu_\t^2 \Exp{\S^2_j}}{2 \mue} + \mue \prob{\B^{\mue}_j >
		A_\t,\ \B^{\mue}_j > B_\t} \\
	&\le \frac{\mu_\t^2 \Exp{\S^2_j}}{2 \mue} + \frac{(\lambda_\t+ \mu_\t)^2\Exp{\B^2_j}}{2 \mue} \\
	&\leq \frac{(\lambda_\t+ \mu_\t)^2\Exp{\B^2_j}}{\mue}.
	\end{align*}
	This completes the proof of lemma.
\end{proof}
\begin{lemma}
	\label{Lemma_Intermediate_result_1}
The probability $\prob{\B^{\mue}_j > B_\t\ |\ \B^{\mue}_j > A_\t}$
converges to 0 u.f. over $j$, under the SFJ limit.	
\end{lemma}

\begin{proof}  
Using arguments similar to those in the proof of Lemma~\ref{Lemma_Intermediate_result} (see  equation (\ref{eq:mgf_bounds})), we can prove that $$\prob{\B^{\mue}_j > A_\t} \geq \frac{\lambda_\t \Exp{\B_j}}{\mue} -
\frac{\lambda_\t^2 \Exp{\B^2_j}}{2 (\mue)^2}.$$ From this and \eqref{eq:cl2_3} we will have,
	\begin{align*}
	\prob{\B^{\mue}_j > B_\t\ |\ \B^{\mue}_j > A_\t} =
	\frac{\prob{\B^{\mue}_j > A_\t,B_\t}}{\prob{\B^{\mue}_j > A_\t}} 
	\displaystyle \leq \frac{\frac{(\lambda_\t + \mu_\t)^2
			\Exp{\B^2_j}}{ 2 (\mue)^2}}{\frac{\lambda_\t \Exp{\B_j}}{\mue} -
		\frac{\lambda_\t^2 \Exp{\B^2_j}}{2 (\mue)^2}} 
	\displaystyle = \frac{(\lambda_\t + \mu_\t)^2
		\Exp{\B^2_j}}{2 \mue \lambda_\t \Exp{\B_j} -
		\lambda_\t^2 \Exp{\B^2_j}}.
	\end{align*} The statement of the lemma now follows by the uniform bound on $\Exp{\B^2_j}$ (Lemma~\ref{lemma:bp_upper_bound}).
\end{proof}


\subsection{Proof of Lemma~\ref{lemma:TauStability}}
\label{sec:lemma:TauStability}

By Lemma \ref{Lemma_Intermediate_result_2} $(iii)$ below, for any
$\varepsilon > 0$, there exists a $\bar{\mu} > 0$ such that for
$\mue > \bar{\mu},$
$$
\sum_{k = 1}^{\infty} \frac{p_0^{\mu_\epsilon} p_1^{\mu_\epsilon} \cdots p_{k-1}^{\mu_\epsilon}}{q_1^{\mu_\epsilon} q_2^{\mu_\epsilon} \cdots q_k^{\mu_\epsilon}}
	\ \le 
	\sum_{k = 1}^{\infty} \frac{p_0^{\infty} p_1^{\infty } \cdots p_{k-1}^{\infty }}{q_1^{\infty } q_2^{ \infty  } \cdots q_k^{\infty  }}  + \varepsilon  < \infty,
	$$where the last inequality is shown in the proof of  Lemma \ref{Lemma_Intermediate_result_2} $(iv)$.  Thus by standard text book arguments (e.g., \cite{Hoel}) we have the result. \qed

        \ignore{

 it is shown that 
By    uniform convergence of these transition probabilities  as  given by Lemma  \ref{Lem_TP_Convergence},   for any  $\varepsilon > 0$, there exists a  $\bar{\mu} > 0$ such that for $\mue >
		\bar{\mu},$ 
		$$ 
		\sup_j \frac{ p_j^\mue}{q_j^\mue} \   \le \  \sup_j   \frac{\lambda_\t}{  \nu_j \mu_\tau}   + \varepsilon  \  \le   \   \frac{1}{2- \delta_\t}  + \varepsilon,
		$$further using Assumption~{\bf
			B}.4. 

		Pick any  $j \geq 1$. 

		 Recall  $p^{\mue}_j :=
		\prob{X_{n+1} = j +1 \ | \ X_n = j}$, the forward probability.  
	By    uniform convergence of these transition probabilities  as  given by Lemma  \ref{Lem_TP_Convergence},   for any  $\varepsilon > 0$, there exists a  $\bar{\mu} > 0$ such that for $\mue >
		\bar{\mu},$ 
		$$ 
		\sup_j p_j^\mue \   \le \  \sup_j   \frac{\lambda_\t}{\lambda_\tau + \nu_j \mu_\tau}   + \varepsilon  \  \le   \   \frac{1}{2- \delta_\t}  + \varepsilon,
		$$further using Assumption~{\bf
			B}.4. Thus  the birth-death chain $\{X_n\}$ is
		positive recurrent for large enough $\mue$  using standard Lyaponov arguments, by choosing appropriate $\varepsilon$.  		\eop 
		
		{\small \begin{equation}
			\label{eq:tau_stab1e}
			\lim_{\mue \ra \infty} p^{\mue}_j = \frac{\lambda_\t}{\lambda_\t +
				\mu_\t \nu_j} < 1 - \delta,
			\end{equation}}
		where the final inequality is a consequence of Assumption~{\bf
			B}.4. As from previous lemma, we have u.f. convergence over all the sub-policies, inequality \eqref{eq:tau_stab1e} implies that for a suitably small
		$\delta > 0,$ there exists $\bar{\mu} > 0$ such that for $\mue >
		\bar{\mu},$ $p^{\mue}_j < 1- \delta$
		 for all $j.$
	 \eop
	 \\
}

\begin{lemma}
	\label{Lemma_Intermediate_result_2}
	Assume the hypothesis of Lemma \ref{Lem_TP_Convergence}. Then we have following under SFJ limit:\\
	(i) $\frac{1}{p_i^\mue} \stackrel{\mue \to \infty}{\longrightarrow} \frac{1}{p_i^\infty}$,   u.f.  over $i$,\\
	(ii)  $\frac{p_i^\mue}{q_i^\mue} \stackrel{\mue \to \infty}{\longrightarrow} \frac{p_i^\infty}{q_i^\infty}$,   u.f.  over $i$, \\
	(iii) $\sum_{k = 1}^{\infty} \frac{p_0^{\mu_\epsilon} p_1^{\mu_\epsilon} \cdots p_{k-1}^{\mu_\epsilon}}{q_1^{\mu_\epsilon} q_2^{\mu_\epsilon} \cdots q_k^{\mu_\epsilon}}
	\stackrel{\mue \to \infty}{\longrightarrow}
	\sum_{k = 1}^{\infty} \frac{p_0^{\infty} p_1^{\infty } \cdots p_{k-1}^{\infty }}{q_1^{\infty } q_2^{ \infty  } \cdots q_k^{\infty  }}$.
	\\
	(iv) $\left ( \sum_{k = 1}^{\infty} \frac{p_0^{\mu_\epsilon} p_1^{\mu_\epsilon} \cdots p_{k-1}^{\mu_\epsilon}}{q_1^{\mu_\epsilon} q_2^{\mu_\epsilon} \cdots q_k^{\mu_\epsilon}} \right )^{-1}
	\stackrel{\mue \to \infty}{\longrightarrow}
	\left (	\sum_{k = 1}^{\infty} \frac{p_0^{\infty} p_1^{\infty } \cdots p_{k-1}^{\infty }}{q_1^{\infty } q_2^{ \infty  } \cdots q_k^{\infty  }} \right )^{-1}$.
\end{lemma}
\begin{proof}
	By Lemma  \ref{Lem_TP_Convergence}, 
	$$p_i^{\mu_\epsilon} \stackrel{\mu_\epsilon \to \infty}{\longrightarrow}  \frac{\lambda_\tau}{\lambda_\tau + \nu_i \mu_\tau} = p_i^{\infty} \mbox{ and  } q_i^{\mu_\epsilon} \stackrel{\mu_\epsilon \to \infty}{\longrightarrow} \frac{\nu_i \mu_\tau}{\lambda_\tau + \nu_i \mu_\tau} = q_i^\infty,$$ uniformly $(u.f.)$ over all the sub-policies $i$. Also, note that $0 < p_i^{\mu_\epsilon}, q_i^{\mu_\epsilon} <1$.
	
	Therefore, for every $\varepsilon >0, \ \exists \ \bar{\mu}$ such that, for all sub-policies $i$, whenever $\mu_\epsilon > \bar{\mu}$, we have:
	\begin{eqnarray}
	\label{Eqn_q_bounds}
	q_i^\infty - \varepsilon < q_i^\mue < q_i^\infty + \varepsilon \ \ \mbox{and} \ \ p_i^\infty - \varepsilon < p_i^\mue < p_i^\infty + \varepsilon.
	\end{eqnarray}
	
	%
	
	{\it  Proof of (i):} 
	Using equation (\ref{Eqn_q_bounds}), for every $\varepsilon > 0$, there exists a $\bar{\mu}$ such that,
	\begin{eqnarray}
	\label{Eqn_fraction}
	\bigg| \frac{1}{q_i^{\mu_\epsilon}} - \frac{1}{q_i^\infty} \bigg| &= &\bigg| \frac{q_i^{\mu_\epsilon} - q_i^{\infty}}{q_i^{\mu_\epsilon} q_i^{\infty}}    \bigg|  \ \le \
	\frac{\varepsilon}{ q_i^{\mu_\epsilon} q_i^{\infty} }, \hspace{0.3cm} \mbox{for all} \  \mu_\epsilon \geq \bar{\mu}. 
	\end{eqnarray}

	Since, $\t$-system is stable for all $i$, i.e., $\frac{\lambda_\t}{\nu_i\mu_\t} < 1 -\delta_\t$   for all $i$ (Assumption {\bf B}.4), the following is true
	\begin{eqnarray}
	\label{Eqn_q_inf_lower}
	\frac{1}{q_i^\infty} = \frac{\lambda_\tau + \nu_i \mu_\tau}{\nu_i \mu_\tau}
	< 2, \ \ \forall \ \ i \  \mbox{ and } \    \frac{1}{q_i^\mue} < \frac{1}{q_i^\infty - \varepsilon} =   \frac{  \frac{1}{q_i^\infty} }{1- \varepsilon   \frac{1}{q_i^\infty}}  < \frac{2}{1 - 2 \varepsilon}.
	\end{eqnarray}
	Using equation (\ref{Eqn_q_inf_lower}) in   (\ref{Eqn_fraction}) we get, for every $\varepsilon > 0$, there exists a $\bar{\mu}$ such that:
	\begin{eqnarray}
	\sup_i	\bigg| \frac{1}{q_i^\mue} - \frac{1}{q_i^\infty} \bigg| < \frac{4\varepsilon}{1 - 2 \varepsilon},  \ \   \mbox{for all} \  \mu_\epsilon \geq \bar{\mu}. 
	\end{eqnarray}
	This implies, $\frac{1}{q_i^{\mu_\epsilon}} \stackrel{\mue \to \infty }{ \longrightarrow} \frac{1}{q_i^{\infty}}$ uniformly over sub-policies $i$.\\
	
	{\it  Proof of (ii):}
	%
	Using equation (\ref{Eqn_q_bounds}), for every $\varepsilon > 0$, there exists a $\bar{\mu}$ such that, for all $\mue \geq \bar{\mu}$,
	\beq
	\bigg| \frac{p_i^\mue}{q_i^\mue} -   \frac{p_i^\infty}{q_i^\infty} \bigg| 
	= \bigg| \frac{p_i^\mue q_i^\infty - p_i^\infty q_i^\mue}{q_i^\mue q_i^\infty}   \bigg| 
	\leq \frac{ q_i^\infty  \big|p_i^\mue  -  p_i^\infty \big| +  p_i^\infty   \big| q_i^\mue  - q_i^\infty \big|}{q_i^\mue q_i^\infty}
	<  \frac{\varepsilon ( q_i^\infty +  p_i^\infty)}{q_i^\mue q_i^\infty}
	= \frac{\varepsilon}{q_i^\mue q_i^\infty} 
	<
	\frac{4 \varepsilon}{1 - 2 \varepsilon}.
	\eeq
	In above, the last inequality followed from equation (\ref{Eqn_q_inf_lower}).
	This implies, $\frac{p_i^\mue}{q_i^\mue} \stackrel{\mue \to \infty}{\longrightarrow}  \frac{p_i^\infty}{q_i^\infty}$ $u.f.$ over $i$.

        \ignore{
	One can easily show that, for finite $k \geq 1$,	
	\begin{eqnarray}
	\label{Eqn_Finite_Prod_conv}
	\frac{p_0^{\mu_\epsilon} p_1^{\mu_\epsilon} \cdots p_{k-1}^{\mu_\epsilon}}{q_1^{\mu_\epsilon} q_2^{\mu_\epsilon} \cdots q_k^{\mu_\epsilon}} \stackrel{\mue \to \infty}{\longrightarrow} \frac{p_0^{\infty} p_1^{\infty} \cdots p_{k-1}^{\infty}}{q_1^{\infty} q_2^{\infty} \cdots q_k^{\infty}}. \hspace{0.3cm}  
	\end{eqnarray}
		} 
	
	%
	


	%
	{\it  Proof of (iii):} From part $(ii)$ and assumption {\bf B}.4, 
	for every $\varepsilon > 0$, there exists a $\bar{\mu}$ such that, for each $\mue > \bar{\mu}$
	\beq
	\frac{p_i^\mue}{q_i^\mue} < \varepsilon + \frac{p_i^\infty}{q_i^\infty} 
	= \varepsilon + \frac{\lambda_\t}{ \mut \nu_i}
	= \varepsilon + \rho_i
	< \varepsilon +  1-\delta_\t.
	\eeq 
	By appropriate choice of  $\varepsilon >0$, such that $\varepsilon + 1-\delta_\t < 1$, we have
	(note that $p_0^\mue = 1$),
	\beq
	\frac{p_0^{\mu_\epsilon} p_1^{\mu_\epsilon} \cdots p_{k-1}^{\mu_\epsilon}}{q_1^{\mu_\epsilon} q_2^{\mu_\epsilon} \cdots q_k^{\mu_\epsilon}}
	< (\varepsilon +  1-\delta_\t )^{k-1} \left( \frac{1}{q_k^\infty} + \varepsilon \right)
	%
	= (\varepsilon +  1-\delta_\t )^{k-1} \left(1 + \rho_k + \varepsilon \right)
	%
	< \left( 2 + \varepsilon \right) (\varepsilon + 1-\delta_\t )^{k-1}. 
	\eeq
	Consider the following summation,
	\beq
	\sum_{k = 1}^\infty \left( 2 + \varepsilon \right) (\varepsilon + 1-\delta_\t )^{k-1} = \frac{2 + \varepsilon}{1 - (\varepsilon +  1-\delta_\t) } < \infty.
	\eeq
	Therefore, by dominated convergence theorem (DCT) for series (note each of the terms inside the series converge by part $(ii)$), we have,
	\beq
	\sum_{k = 1}^\infty \frac{p_0^{\mu_\epsilon} p_1^{\mu_\epsilon} \cdots p_{k-1}^{\mu_\epsilon}}{q_1^{\mu_\epsilon} q_2^{\mu_\epsilon} \cdots q_k^{\mu_\epsilon}} &\stackrel{\mue \to \infty}{\longrightarrow}& \sum_{k = 1}^\infty \frac{p_0^{\infty} p_1^{\infty} \cdots p_{k-1}^{\infty}}{q_1^{\infty} q_2^{\infty} \cdots q_k^{\infty}}. 
	%
	\eeq

	{\it  Proof of (iv):} It suffices to show that the limit in part $(iii)$ lies in $\left(0,\infty\right).$ The limit is clearly positive as it is the summation of strictly positive terms. The finiteness of the limit follows from the DCT as in the proof of part $(iii).$	
%
\end{proof}

%
%
%
\subsection{Proof of Theorem~\ref{Thm_SD_Convergence}}
\label{sec:Thm_SD_Convergence}

For any  $i \geq 0 $,  $p_i^{\mu_\epsilon} = P (A_\t \le \Upsilon_i^\mue (B_\t) )$  denotes the probability that  $\t$-arrival is before $\t$-departure in the corresponding $\mue$-system, when
 $\t$ occupancy equals $i$.

	From equation (\ref{Eqn_SD_pre-limit}), the stationary distribution of the tolerant embedded BD chain is given by,
	\begin{eqnarray}
	\label{Eqn_SD}
	\tilde{\pi}_i^{\mu_\epsilon} = \frac{p_0^{\mu_\epsilon} p_1^{\mu_\epsilon} p_2^{\mu_\epsilon} \cdots p_{i-1}^{\mu_\epsilon}}{q_1^{\mu_\epsilon} q_2^{\mu_\epsilon} q_3^{\mu_\epsilon} \cdots q_i^{\mu_\epsilon}} \tilde{\pi}_0^{\mu_\epsilon} ; \hspace{0.3cm} i \geq 1 \hspace{0.5cm}\mbox{with} \hspace{0.3cm} \tilde{\pi}_0^\mue = \frac{1}{1 + \sum_{k = 1}^{\infty} \frac{p_0^{\mu_\epsilon} p_1^{\mu_\epsilon} \cdots p_{k-1}^{\mu_\epsilon}}{q_1^{\mu_\epsilon} q_2^{\mu_\epsilon} \cdots q_k^{\mu_\epsilon}}}. 
	\end{eqnarray} 
	By Lemma \ref{Lemma_Intermediate_result_2},  
	\begin{eqnarray}
	\label{Eqn_pi_0_conv}
	\lim \limits_{\mue \to \infty} \tilde{\pi}_0^\mue \to \tilde{\pi}_0^\infty \ \ \mbox{with} \ \ \tilde{\pi}_0^\infty  = 1 \big/ \bigg( 1 + \sum_{k = 1}^{\infty} \frac{p_0^{\infty} p_1^{\infty} \cdots p_{k-1}^{\infty}}{q_1^{\infty} q_2^{\infty} \cdots q_k^{\infty}} \bigg).  
	\end{eqnarray}
	and
	\begin{eqnarray}
	\label{Eqn_tilde_pi_conv}
	\lim \limits_{\mue \to \infty} \tilde{\pi}_i^\mue = \tilde{\pi}_i^\infty,  \ \   \mbox{ with } \ 
	\tilde{\pi}_i^\infty  = \frac{p_0^{\infty} p_1^{\infty} \cdots p_{k-1}^{\infty}}{q_1^{\infty} q_2^{\infty} \cdots q_k^{\infty}} \tilde{\pi}_0^\infty.
	\end{eqnarray}
	Thus as $\mue \to \infty$, the stationary distribution of the
        embedded tolerant BD chain converges to the stationary
        distribution of the embedded chain corresponding to the limit
        SDSR-M/M/1 queue. Thus we have weak convergence of the
        stationary distributions. \qed

        \subsection{Proof of Theorem~\ref{Thm_Lim_SD_Conv}}
        \label{sec:Thm_Lim_SD_Conv}
        
        The statement of Theorem~\ref{Thm_Lim_SD_Conv} follows
        directly from Theorem~\ref{Thm_SD_Convergence} invoking the
        following lemma, which relates the steady state distribution
        of a continuous time queue with the stationary distribution of
        the discrete-time embedded Markov chain obtained by sampling
        the queue-occupancy at arrival and departure epochs.

        \begin{lemma}
          \label{lemma:EMC_dist}
          Consider a stable queueing system with Poisson job arrivals
          and no simultaneous departures (i.e., jobs depart one at a
          time with probability 1), such that the time-average
          distribution of queue occupancy $\pi = \{\pi_i\}_{i \geq 0}$
          is well defined, i.e.,
	$$ 
	\pi_i  = \lim_{t \ra \infty} \frac{1}{t} \int_{0}^t
	\indicator{X(s)=i} ds \quad \forall \ i \quad (a.s.),
	$$ where $X(t)$ denote the queue occupancy at time $t.$ Let $\tilde{\pi} = \{\tilde{\pi}_i\}_{i \geq 1}$ denote the (discrete-time)
	time-average distribution of the queue occupancy sampled just
	following arrival/departure epochs. Then 
	\begin{eqnarray*} 
	\tilde{\pi}_0 = \frac{1}{2} \pi_0,  \mbox{ and }
	\tilde{\pi}_i = \frac{1}{2} (\pi_{i-1} + \pi_i)  \quad (i \geq 1).  \end{eqnarray*}
\end{lemma}
	{\bf Proof:} 
	Let $X_n$ denote the queue occupancy just following the $n$th
	arrival/departure epoch. Let $X^A_n$ (respectively, $X^D_n$) denote
	the queue occupancy just before the $n$th arrival (respectively, just following the $n$th departure).
		By PASTA,
		 $$
		 \pi_i = \frac{1}{n} \sum_{j = 1}^n \indicator{X^A_j = i}.
		 $$ Moreover, since the time-average seen by arrivals matches the
	time average seen by departures (see, for
	example, \cite[Chapter~26]{Harchol-Balter}),
	 $$
	 \pi_i = \frac{1}{n} \sum_{j = 1}^n \indicator{X^D_j = i}.
	 $$
	
	Define, for $n \geq 1,$
	\begin{align*}
	a_n &= \indicator{X_{n} = X_{n-1} + 1} \\ 
	A_n &= \sum_{i = 1}^n a_i \\
	d_n &= \indicator{X_{n} = X_{n-1} - 1} \\ 
	D_n &= \sum_{i = 1}^n d_i 
	\end{align*}
	
	For $i \geq 1,$
	\begin{eqnarray*}
	\tilde{\pi}_i &=& \lim_{n \ra \infty} \frac{1}{n} \sum_{k = 1}^n
	\indicator{X_k = i}\\
	&=& \lim_{n \ra \infty} \left[ \frac{1}{n} \sum_{k = 1}^n \indicator{X_k = i} a_k
	+ \frac{1}{n} \sum_{k = 1}^n \indicator{X_k = i} d_k \right]\\
	&=& \lim_{n \ra \infty} \left[ \frac{A_n}{n} \frac{1}{A_n}\sum_{j = 1}^{A_n} \indicator{X^A_k = i-1} 
	+ \frac{D_n}{n} \frac{1}{D_n} \sum_{j = 1}^{D_n} \indicator{X^D_j = i} \right]\\
	&=& \frac{1}{2} (\pi_{i-1} + \pi_i). 
	\end{eqnarray*}
Also,	
	\begin{align*}
	\tilde{\pi}_0 &= \lim_{n \ra \infty} \frac{1}{n} \sum_{k = 1}^n
	\indicator{X_k = 0}\\
	&= \frac{1}{n} \sum_{k = 1}^n \indicator{X_k = 0} d_k \\
	&= \frac{D_n}{n} \frac{1}{D_n} \sum_{j = 1}^{D_n} \indicator{X^D_j = 0} \\
	&= \frac{1}{2} \pi_0. 
	\end{align*}
	\qed

\subsection{Proof of Theorem~{\ref{Thm_tau}}}
\label{sec:Thm_tau}

Note that, by Lemma \ref{Lemma_Intermediate_result_2}, for every $\varepsilon >0$, there exists a $\bar{\mu}$ such that, for all sub-policies $i$ and $\mue > \bar{\mu}$,  
		$$
	\frac{p_i^\mue}{q_i^\mue} <  \varepsilon +  1-\delta_\t.
	$$
As, for all $i \geq 0$, $\pi_i^\mue \leq 2 \tilde{\pi}_i^\mue$, we will have, for all $\mue > \bar{\mu}$ and $i$,
$$
\big|i \pi_i^\mue \big|
\leq 	\big|2 i \tilde{\pi}_i^\mue \big|
= \bigg| 2  i \frac{p_0^{\mu_\epsilon} p_1^{\mu_\epsilon} \cdots p_{i-1}^{\mu_\epsilon}}{q_1^{\mu_\epsilon} q_2^{\mu_\epsilon} \cdots q_i^{\mu_\epsilon}} \tilde{\pi}_0^\mue \bigg|
< 2 i \left(2 + \varepsilon \right) \left(  \varepsilon +  1-\delta_\t \right) ^{i-1}.
$$
	Note that the following summation is finite, as:
	\begin{equation}
	\sum_{i = 0}^{\infty} 2 i \left(2 + \varepsilon \right) \left( \varepsilon + 1 + \delta_\t \right) ^{i-1} 
	= 2 \left(2 + \varepsilon \right) \sum_{i = 0}^{\infty} i \left( \varepsilon + 1 - \delta_\t \right) ^{i-1}
	%
	=  \frac{2 \left(2 + \epsilon \right) }{\left(  1 - \left(\epsilon + 1 - \delta_\t \right)   \right)^2} < \infty \label{Eqn_up}.
	\end{equation}
	By the  above upper bound, DCT   can be applied   (note by Theorem \ref{Thm_Lim_SD_Conv},  $\pi_i^{\mu_\i} \to \pi_i$ for all $i$),  and   we have:
	\begin{eqnarray}
	\lim\limits_{\mue \to \infty}  \sum_{i = 0}^{\infty} i {\pi}_i^\mue 
%
=  \sum_{i = 0}^{\infty}  \lim\limits_{\mue \to \infty} i {\pi}_i^\mue
= \sum_{i=0}^{\infty}  i {\pi}_i = \Exp{N} \nonumber.
	\end{eqnarray}	
This completes the proof.\qed
\section{Eager Performance}  
\label{Appendix_eager}

In this appendix we provide the proofs related to eager performance. These proofs depend extensively upon the coupling  details provided in Appendix \ref{Appendix_sample}.

\noindent {\bf Proof of Theorem~ \ref{Thm_PB}}: Let $N_B^\mue (t)$ denote the total number of eager customers that are blocked, in the pre-limit, during the time interval $\left[ 0, t\right]$ and  $N_A^\mue (t)$ be  the total number of eager customers arrived in the same time interval. Moreover, let $N_{B_j}^\mue (t)$ and $N_{A_j}^\mue (t)$ respectively denote the total number of eager customers blocked and arrived respectively before time $t$, in the pre-limit, in the time period  during  which the sub-policy $j$ has been used. 
Then the overall pre-limit blocking probability of the eager class is given by the following  (the time limits exist as shown by Theorem \ref{Thm_PB_for_one_subpolicy} and the arguments given below), which can be split as below:
	$$
P_B^\mue =	\lim\limits_{t \to \infty } \frac{N_B^\mue (t)}{N_A^\mue (t)}
	= \lim\limits_{t \to \infty }  \sum_{j=0}^\infty \frac{N_{B_j}^\mue (t)}{N_A^\mue (t)}
	= \lim\limits_{t \to \infty }  \sum_{j=0}^{\bar{g}} \frac{N_{B_j}^\mue (t)}{N_A^\mue (t)}  + \lim\limits_{t \to \infty }  \sum_{j> \bar{g}}^\infty \frac{N_{B_j}^\mue (t)}{N_A^\mue (t)} , \mbox{ for any } {\bar g}. 
	$$
	Consider the following difference,
	\begin{eqnarray}
	\left| \lim\limits_{t \to \infty } \frac{N_B^\mue (t)}{N_A^\mue (t)} - \sum_{j = 0}^{\infty} P_{B_j} {\pi}_j \right| 
	%
	%
	&=& \left| \lim\limits_{t \to \infty }  \sum_{j=0}^{\bar{g}} \frac{N_{B_j}^\mue (t)}{N_A^\mue (t)} - \sum_{j = 0}^{{\bar{g}}} P_{B_j} {\pi}_j  + \lim\limits_{t \to \infty }  \sum_{j> \bar{g}}^\infty \frac{N_{B_j}^\mue (t)}{N_A^\mue (t)} - \sum_{j > \bar{g}} P_{B_j} {\pi}_j  \right|  \nonumber  \\
 && 	\hspace{-15mm} \leq
		  \left| \lim\limits_{t \to \infty }  \sum_{j=0}^{\bar{g}} \frac{N_{B_j}^\mue (t)}{N_A^\mue (t)} - \sum_{j = 0}^{{\bar{g}}} P_{B_j} {\pi}_j \right| + \left| \lim\limits_{t \to \infty }  \sum_{j> \bar{g}}^\infty \frac{N_{B_j}^\mue (t)}{N_A^\mue (t)} - \sum_{j > \bar{g}}^{\infty} P_{B_j} {\pi}_j  \right| , \label{Eqn_Final_Diff} \hspace{6mm}
	\end{eqnarray}
	and we are done if we can show that the above difference tends to zero as $\mu_\i \to \infty$. We first consider the second term (number of losses are upper bounded by number of arrivals, 
	$N_{B_j}^\mue (t) \le N^\mue_{A_j}(t)$):
	\begin{eqnarray}
	 \left| \lim\limits_{t \to \infty }  \sum_{j> \bar{g}}^\infty \frac{N_{B_j}^\mue (t)}{N_A^\mue (t)} - \sum_{j > \bar{g}}^{\infty} P_{B_j} {\pi}_j  \right|
	%
	&\leq&   \left| \lim\limits_{t \to \infty }  \sum_{j> \bar{g}}^\infty  \frac{N_{B_j}^\mue (t)}{N_A^\mue (t)}  \right| + \left| \sum_{j > \bar{g}}^{\infty} P_{B_j} {\pi}_j  \right| \nonumber \\
	&\leq&   \lim\limits_{t \to \infty }    \frac{\sum_{j> \bar{g}}^\infty N_{A_j}^\mue (t)}{N_A^\mue (t)}  + \sum_{j > \bar{g}}^{\infty} P_{B_j} {\pi}_j  \nonumber \\
	&\stackrel{a}{\leq}&   P \left( Q_\t^\mue > \bar{g}\right) +  \sum_{j > \bar{g}}^{\infty} P_{B_j} {\pi}_j   \nonumber \nonumber\\
	&\leq&    \sum_{j  > \bar{g} } \left( {\pi}^\mue_j -  {\pi}_j  \right) + 2 \sum_{j > \bar{g}}^{\infty}  {\pi}_j  \nonumber \\
	&\leq&    \sum_{j  > \bar{g} } j \left( {\pi}^\mue_j -  {\pi}_j  \right) + 2 \sum_{j > \bar{g}}^{\infty}  {\pi}_j.  \label{Eqn_Final_difference_eager}
	\end{eqnarray} In the above, the inequality `a'  follows from PASTA (Poisson Arrivals See Time Averages). Further, for the limit system
	$$
	\sum_{j > g}  {\pi}_j   \to 0 \mbox{ as } g \to \infty .
	$$
	Therefore, for every $\delta >0$, there exists $\bar{g}$, such that for all $j \geq \bar{g}$,
	\begin{eqnarray}
	\label{Eqn_Pi_Sum_Bound}
	\sum_{j > g}  {\pi}_j   < \delta.
	\end{eqnarray}
	Fix this ${\bar g}$.  From Theorem \ref{Thm_PB_for_one_subpolicy}  we have, 
	$$
	\lim\limits_{t \to \infty }  \sum_{j=0}^{\bar{g}} \frac{N_{B_j}^\mue (t)}{N_A^\mue (t)} \stackrel{\mue \to \infty}{\longrightarrow} \sum_{j = 0}^{{\bar{g}}} P_{B_j} {\pi}_j,
	$$ 
	and hence   there exists a $\bar{\mu_1}$, such that 
	\begin{eqnarray}
	\label{Eqn_Finite_Sum_Bound}
	\left| \lim\limits_{t \to \infty }  \sum_{j=0}^{\bar{g}} \frac{N_{B_j}^\mue (t)}{N_A^\mue (t)} - \sum_{j = 0}^{{\bar{g}}} P_{B_j} {\pi}_j \right| < \delta, \ \ \mbox{for all $\mue > \bar{\mu_1}$}.
	\end{eqnarray}
	%
	%
	%
 Further from Theorem~\ref{Thm_tau},  
	$$
	\lim_{\mue \to \infty } \sum_{i = 0}^\infty i {\pi}_i^\mue 
	 = \sum_{i = 0}^\infty i {\pi}_i,
	$$ and thus   there exists a $\bar{\mu_2}$, such that for all $\mue \geq \bar{\mu_2}$,
	\begin{eqnarray}
	\label{Eqn_Sum_ipi}
	\sum_{j = 0}^\infty i \left( {\pi}_j^\mue - {\pi}_j \right) < \delta.  
	\end{eqnarray}
	Choose, $\bar{\mu} = \max \left(\mu_1, \mu_2\right)$. First using, (\ref{Eqn_Pi_Sum_Bound}) and (\ref{Eqn_Sum_ipi}) in (\ref{Eqn_Final_difference_eager}) and then  (\ref{Eqn_Final_difference_eager}) and (\ref{Eqn_Finite_Sum_Bound}) in (\ref{Eqn_Final_Diff}), we get, for every $\delta > 0 $, there exists $\bar{\mu}$ such that,
	\begin{eqnarray*}
		\hspace{20mm}
		\left| \lim\limits_{t \to \infty } \frac{N_B^\mue (t)}{N_A^\mue (t)} - \sum_{j = 0}^{\infty} P_{B_j} {\pi}_j \right| < \delta + \delta + 2 \delta = 4 \delta \ \ \mbox{for all} \ \ \mue \geq \bar{\mu}.
		\hspace{20mm}  \mbox{ \eop}
	\end{eqnarray*}
%
%


\begin{thm}
	\label{Thm_PB_for_one_subpolicy}
	Fix a sub-policy $j$, then the time  limit  $ \lim_{t \to \infty} \frac{N_{B_j}^\mue (t)}{N_A^\mue (t)} $  exists almost surely for all $\mu_\i \ge {\bar \mu}$, where    $ {\bar \mu}$ is given in Lemma \ref{lemma:TauStability}. Further we have
	\begin{eqnarray*}
	\lim\limits_{\mue \to \infty } \lim\limits_{t \to \infty }   \frac{N_{B_j}^\mue (t)}{N_A^\mue (t)} = P_{B_j} {\pi}_j \mbox{ almost surely. }
	 \end{eqnarray*}
\end{thm}
\begin{proof}
For the purpose of almost sure comparison we construct the
$\i$-process influenced by dynamic decisions of the scheduling policy
depending upon the $\t$-state in the following manner. This procedure
is the natural extension of the procedure used in \cite{ANOR}, which is also summarized in Appendix \ref{sec:SFJ_scale_details}.  

\begin{itemize}
\item For every  sub-policy $j$ construct one sample path of $\epsilon$-evolution (basically sequence of $\i$-inter-arrival times, and $\i$-service times)  that lasts  for complete time axis   for $\mu_\i  = 1$ and $\lambda_\i = \rho_\i$ and 
define the process for any  general $\mu_\i < \infty$   and for the same sub-policy,  using the constructed sample path as explained  in  Appendix~\ref{sec:SFJ_scale_details}. 
	\item Say we begin with the system in $\t$-state equal  to $j$  and $\i$-state equal to 0 (by assumption {\bf A}.1, $\i$-state is reset to 0   at any  $\t$-transition  epoch).  Then start  $\i$-process of the system  with the initial part of $\i$-process corresponding to  sub-policy $j$. 
	
	\item We refer the time duration from the start of an $\i$-idle period
	and the end of the consecutive $\epsilon$-busy period as one
	$\epsilon$-cycle (note that this $\i$-cycle is stochastically same as the $\i$-busy cycle defined in Appendix~\ref{sec:SFJ_scale_details}).  The $\i$-process of the system starts with first full $\i$-cycle corresponding to sub-policy $j$ and will use some initial number of full cycles (the number can be greater than or equal to zero) and  it may also  use partially  the  cycle  following those full cycles (till the instance of first $\t$-transition).  
	
	
	\item  Once the sub-policy has to switch, due to $\t$-transitions (to state $j-1$ or $j+1$), start using a required part of the $\i$-process corresponding to the new sub-policy. Always start from the first new $\i$-cycle that is previously not used. This process is appropriate in view of assumption {\bf A.}1:  the existing $\i$-customers are dropped at every $\t$-transition and the $\i$-process starts afresh.  
	
	\item Note that   the end part of any interrupted $\i$-cycle  (a partial $\i$-cycle), interrupted by a  $\t$-transition,  is not used for any further construction.   
	Thus this full $\i$-cycle can be used to upper bound the partial   $\i$-cycle just before the $\t$-transition,    when required. Further, this upper bound  (cycle) is independent of the other $\i$-cycles.

\end{itemize}  
Recall 
by the way of construction, when full $\i$-cycle is used (notations as in Appendix \ref{Appendix_sample} and as in Theorem \ref{Thm_PB}):
\begin{eqnarray*} 
N_A^\mue ({\mathcal B}_j^\mue  ) = N_A^1 ( {\mathcal B}_j^1 )   \   \mbox{and}   \    N_{B_j}^\mue ({\mathcal B}_j^\mue  ) = N_{B_j}^1 ( {\mathcal B}_j^1 )  \mbox{ etc., sample path wise. }
\end{eqnarray*}
%

We consider positive recurrent systems, i.e., systems with  $\mu_\i \ge {\bar \mu}$, the threshold beyond which the system is stable as given by Lemma \ref{lemma:TauStability}. 
For such systems,
all the time limits mentioned in this proof exist   and  this existence is proved  along with the rest of the arguments, in the following. 

Refer   
the  time interval from start of state $j$ and  before the forward or backward transition (to state $j-1$ or $j+1$), by $\tau$-state process $X_\tau (\cdot)$,  as   
a  $j$-cycle. 
Let $\Lambda_{j}(t)$ represent the total time consisting of full $\epsilon$-cycles during time  period $[0, t]$, such that the $\tau$-state is $j$. 
This is composed  of portions of all  $j$-cycles that intersect with $[0, t]$, and,  obtained after removing the partial $\epsilon$-cycles at the end of each $\t$-state transition (occurring in those $j$-cycles). Note  by assumption {\bf A}.1 a new $\i$-cycle starts at every $\t$-transition epoch.  
For $ j \geq 1 $,  let 
$$
F_{j } (t): =  \int_0^t \indc_{\{ X_\tau (s) = j\}} ds \ \ \mbox{  }  
$$represent the  portion of time spent in $j$-cycles before $t$ and   define $\Psi_{j} (t) := F_{j} (t) - \Lambda_{j}(t)$. 
The portion $\Psi_{j} (t)$  is  basically made up of  the  residual   $\i$-cycles, that are ongoing  at  the end  of $\t$-state transitions in all $j$-cycles before $t$.
Thus\footnote{All the limits mentioned below exist for all $\mu_\i \ge {\bar \mu}$  (the threshold of Lemma \ref{lemma:TauStability}), because  the process for any such  $\mu_\i \ge {\bar \mu}$  is positive recurrent.
This is because all the time averages  can be written as the time averages in an appropriate renewal process and renewal reward theorem (RRT) can be applied, thanks to assumption {\bf A}.4. 
For example time limit of  $N_{\partial_j} (t) / { t}  $  exists almost surely,   using  RRT with 
 $j$-cycles as renewal periods and with the reward in a cycle as the expected number of $\i$-customers dropped  at $\t$-transition epochs with $\t$-state  equal to $j$. 
},
\begin{eqnarray}
\lim _{t \to \infty } \frac{N_{B_j}^{\mue}(t)}{N_A^{\mue}(t)}
= \lim _{t \to \infty } \frac{  N_{B_j}^{\mue}(\Lambda_{j}(t)) + N_{B_j}^{\mue}( \Psi_{j} (t))  + N_{\partial_j} (t) }{N_A^{\mue}(t)} ,
\label{Eqn_split_NB_By_NA}
\end{eqnarray}where $N_{\partial_j} (t)$ is the number of $\i$-jobs dropped at  $\t$-transition epochs with $\t$-state equal to $j$ (see   assumption {\bf A}.1). 
We evaluate each of these components one after the other.  
{Note that we are considering the case: $j \geq 1$. The proof follows in exactly the similar way even for $j = 0$ state, but would need  obvious changes ($\t$-job sizes should not be considered).}
 
{\bf First fraction   in  RHS  (right hand side) of  (\ref{Eqn_split_NB_By_NA})}. By memoryless property of Poisson process representing the $\epsilon$-arrival process, 
\begin{eqnarray}
\label{Eqn_split_1}
\frac{ N_{B_j}^{\mue}(\Lambda_{j}(t))}{N_A^{\mue}(t)}  = \frac{ N_{B_j}^{\mue}(\Lambda_{j}(t))}{N_{A}^{\mue}(\Lambda_{j}(t))} \
\frac{ N_{A}^{\mue}(\Lambda_{j}(t))}{ \Lambda_{j}(t)}  \   \frac{ \Lambda_{j}(t)}{t}  \frac{t}{N_A^{\mue}(t)}.
\end{eqnarray}
On keen observation, it is easy to realize that $\Lambda_{j}(t) $ is made up of i.i.d. $\epsilon$-cycles, and hence is a renewal process. Thus one can analyze it using Renewal Reward theorem (RRT). These $\i$-cycles are the special cycles,   in that these\footnote{By memoryless property of exponential service times and arrival processes.}  are the   
cycles in which the $\tau$-state has not changed, i.e.,  the $\tau$-service is not completed and during which a $\tau$-arrival did not occur.  Let 
$U$ represent      one such  special $\i$-cycle, i.e.:
\begin{eqnarray}
\label{Eqn_special_cycle}
U :=  {\mathcal  B}_j^{\mu_\epsilon}  \Big |   
{\mathcal  B}_j^{\mu_\epsilon}  <  A_\tau  \ \cap \ \S^{\mu_\epsilon}_j  <   B_\tau,  
\end{eqnarray} 
where 
${\mathcal  B}_j^{\mu_\epsilon} = {\mathcal  B}_j^1 / \mu_\epsilon$ represents  (see (\ref{Eqn_Bmue})) 
an $\i$-cycle,    
$
\S^{\mu_\epsilon}_j  =  \S^1_j / {\mu_\epsilon} 
$ (see (\ref{Eqn_Smue})) represents the server time available to the $\tau$-customers during this period and  $A_\tau$, $ B_\tau $  respectively represent the inter-arrival time before next $\tau$-arrival and the service time of current $\tau$-job (the residual ones equal the actual ones by memoryless property).  Let ${\mathcal I}^{\mu_\epsilon}$ represent the following indicator random variable:
\begin{eqnarray}
\label{Eqn_indicators}
{\mathcal I}^{\mu_\epsilon}   =  \indc_{ \{ {\mathcal  B}_j^{\mu_\epsilon}  <  A_\tau  \ \cap \ \S^{\mu_\epsilon}_j  <   B_\tau \} }  = 
\indc_{ \{  {\mathcal  B}^1_j  <{\mu_\epsilon}    A_\tau  \ \cap \ \S^1_j   <  {\mu_\epsilon}   B_\tau  \} }  .
\end{eqnarray}
It is clear that the expected length of the renewal cycle (by   conditioning  on $({\mathcal B}_j^\mue ,     \S^\mue_j )$ and using (\ref{Eqn_Compare_mu_epsilon_schemes}), (\ref{Eqn_Bmue})), 
\beq
E[U] = \frac { E[ {\mathcal  B}_j^{\mu_\epsilon}    {\mathcal I}^{\mu_\epsilon}   ]}{ E[ {\mathcal I}^{\mu_\epsilon}  ] } >   E[ {\mathcal  B}_j^{\mu_\epsilon}  {\mathcal I}^{\mu_\epsilon}   ] 
&=& E \left [ {\mathcal  B}_j^{\mu_\epsilon}    (1- \exp(-\lambda_\t {\mathcal B}_j^\mue ) (1- \exp (- \mu_\tau  \S_j^\mue ) ) \right ] \\
& = & \frac{1}{\mue}  E \left [ {\mathcal  B}_j^1      (1- \exp(-\lambda_\t {\mathcal B}_j^1 / \mue ) (1- \exp (- \mu_\tau  \S_j^1/ \mue ) ) \right ]  > 0
\eeq  for any given $j$ and for any $\mue < \infty$. {\bf A}.4, we also have  $E[U]  < \infty.$
Thus by resorting to RRT two times we have,  as $t \to \infty$,
$$
\frac{ N_{B_j}^{\mue}(\Lambda_{j}(t))}{N_{A}^{\mue}(\Lambda_{j}(t))}  
= \frac{ N_{B_j}^{\mue}(\Lambda_{j}(t))}{ \Lambda_{j}(t))  }   \frac{  \Lambda_{j}(t))  }  {N_{A}^{\mue}(\Lambda_{j}(t))}  \to  
\frac{ E[N_{B_j}^{\mu_\epsilon} (U) ] }{ E[  U  ] } \frac{ E[    U     ] }{ E[N_A^{\mu_\epsilon} (U) ] } \   a.s.
$$ 
Therefore,
\beq
\lim_{t \to \infty} \frac{ N_{B_j}^{\mue}(\Lambda_{j}(t))}{N_{A}^{\mue}(\Lambda_{j}(t))}  
&=& \frac{ E \big [   N_{B_j}^{\mu_\epsilon} ({\mathcal  B}_j^{\mu_\epsilon}  ) \ \big |   {\mathcal I}^{\mu_\epsilon}  \big ]  }
{  E \big [N_A^{\mu_\epsilon} ({\mathcal  B}_j^{\mu_\epsilon}  )  \  \big |  {\mathcal I}^{\mu_\epsilon}   \big ] }   \\
&=& 
\frac{ E \big [   N_{B_j}^{\mu_\epsilon} ({\mathcal  B}_j^{\mu_\epsilon}  )    \     {\mathcal I}^{\mu_\epsilon}   \big ]  }
{  E \big [N_A^{\mu_\epsilon} ({\mathcal  B}_j^{\mu_\epsilon}  )    {\mathcal I}^{\mu_\epsilon} \big ] }   \\
&=& 
\frac{ E \big [   N_{B_j}^1 ({\mathcal  B}_j^1 )    \     {\mathcal I}^{\mu_\epsilon}   \big ]  }
{  E \big [N_A^1  ({\mathcal  B}_j^1   )    {\mathcal I}^{\mu_\epsilon} \big ] } \mbox { a.s. } 
\eeq
The last equality follows   by way of construction (details in  (\ref{Eqn_Compare_mu_epsilon_schemes}) and more details are in \cite{ANOR}). 
It is clear from (\ref{Eqn_indicators}) that the indicators  ${\mathcal I}^{\mu_\epsilon} \to 1 $ almost surely as $\mu_\epsilon \to \infty$,   
$$ N_{B_j}^1  ({\mathcal  B}_j^{1}  )    \     {\mathcal I}^{\mu_\epsilon}  \le N_{B_j}^1 ({\mathcal  B}_j^{1}  )  \  a.s.  \mbox{ for all }  \mu_\epsilon,
$$   and $ N_{B_j}^1  ({\mathcal  B}_j^{1}  )$ is integrable (by {\bf A}.4).  Similar conclusions follow for number of arrivals $  N_A^1 $.  Thus by Dominated convergence theorem   as $\mu_\epsilon \to \infty$,
\beq
\lim_{t \to \infty} \frac{ N_{B_j}^{\mue}(\Lambda_{j}(t))}{N_{A}^{\mue}(\Lambda_{j}(t))}   \to  
\frac{ E \big [   N_{B_j}^1 ({\mathcal  B}_j^1 )    \      \big ]  }
{  E \big [N_A^1  ({\mathcal  B}_j^1   )     \big ] }   =  {P_B}_j  \mbox { a.s. }, \eeq
the last equality follows by applying RRT to a system that implements $j$-th sub-policy for $\i$-class   in $\t$-static manner.   

Now we consider the second and last/fourth term of the RHS of (\ref{Eqn_split_1}).  Using similar logic  (RRT, DCT) as before
and by the way of construction as in equation (\ref{Eqn_Compare_mu_epsilon_schemes}):
\beq
\lim_{t \to \infty } \frac{ N_{A}^{\mue}(\Lambda_{j}(t))}{ \Lambda_{j}(t)}  \     \frac{t}{N_A^{\epsilon}(t)} 
&=& 
\frac{  \frac{ E\left[  N_{A}^{1}  ({\mathcal  B}_j^{1}  ) {\mathcal I}^{\mu_\epsilon} \right] } {E\left[ {\mathcal I}^{\mu_\epsilon}\right] }  }{ \frac{ E\left[    \frac {  { \mathcal  B}_j^{1}  } { \mu_\epsilon  }  {\mathcal I}^{\mu_\epsilon}     \right]  }  {E \left[ {\mathcal I}^{\mu_\epsilon}\right] }}  \frac{1}{\lambda_\epsilon} \\
&=&
\frac{1}{ \rho_\epsilon } \frac{ E[  N_{A}^{1}  ({\mathcal  B}_j^{1}  ) ] }{ E[     { \mathcal  B}_j^{1}       {\mathcal I}^{\mu_\epsilon}     ]  }  \\
&
\to &  
\frac{1}{ \rho_\epsilon } \frac{ E[  N_{A}^{1}  ({\mathcal  B}_j^{1}  ) ] } { E[       { \mathcal  B}_j^{1}         ]  }  \mbox{, as } \mu_\epsilon \to \infty \\
&& =  \frac{1}{ \rho_\epsilon }  \rho_\epsilon = 1 \mbox { a.s. }
\eeq

The last equality follows because the rate of arrivals in $\mu_\epsilon = 1$ system equals $\rho_\epsilon$. 

Consider the remaining term (third term) of the RHS of (\ref{Eqn_split_1}). 
The fraction of time spent by $\tau$-customer in $j$-cycle is $ F_j (t) = \Lambda_{j} (t) + \Psi
_j(t)$, therefore
\begin{eqnarray}
\label{Eqn_Lamda_t_act}
\lim_{t \to \infty} \frac{ \Lambda_j(t)}{t} = \frac{F_j (t) -  \Psi_j(t) }{t}.
\end{eqnarray}

For all $\mu_\epsilon \ge {\bar \mu}$ we have:  
\begin{eqnarray*}
	\lim_{t \to \infty} \frac{F_j(t)}{t} =    \pi_{j}^{\mu_\epsilon}   \ \ \ \mbox{ a.s.}  ,  
\end{eqnarray*}
where $\pi_{j}^{\mu_\epsilon} $ is the stationary probability that $\t$-queue is in state $j$ with  $\mu_\epsilon$ system. 
By Lemma \ref{Lemma_Psij},  as $\mue \to \infty$ 
$$
\lim_{t \to \infty}   \frac{\Psi_j(t) }{t}  \to 0. 
$$
Therefore, from Equation (\ref{Eqn_Lamda_t_act}),
\begin{eqnarray*}
	\lim_{\mu_\epsilon \to \infty} \lim_{t \to \infty} \frac{ \Lambda_j(t)}{t} = \lim_{\mu_\epsilon \to \infty} \pi_{j}^{\mu_\epsilon} = \pi_{j}.
\end{eqnarray*}

Overall, we get from Equation (\ref{Eqn_split_1}),
\begin{eqnarray*}
	\lim_{\mu_\epsilon \to \infty} \lim_{t \to \infty} \frac{ N_{B_j}^{\mue}(\Lambda_{j}(t))}{N_A^{\mue}(t)} = {P_B}_j \pi_{j}.
\end{eqnarray*}

{\bf Second fraction of Equation (\ref{Eqn_split_NB_By_NA}):}
One can split the second fraction\footnote{The time limit of the LHS exists, because time limit of $\Psi_j(t)/t$ exists as discussed in the proof of Lemma \ref{Lemma_Psij}. } as below:
\begin{eqnarray}
\label{Eqn_split_2nd_fraction}
\lim_{t \to \infty} \frac{N_{B_j}^{\mue}( \Psi_{j} (t))}{N_A^{\mue}(t)} \le 
\lim_{t \to \infty} \frac{N_{B_j}^{\mue}( {\tilde \Psi}_{j} (t))}{N_A^{\mue}(t)} 
= \lim_{t \to \infty} \frac{N_{B_j}^{\mue}( {\tilde \Psi}_{j} (t))}{{\tilde \Psi}_{j} (t)} \frac{{\tilde \Psi}_{j} (t)}{t} \frac{t}{N_A^{\mue}(t)}.
\end{eqnarray} 
%
In the above  the end residual $\i$-cycles (which form $\Psi(t)$), belonging to the same $j$-group  are i.i.d., can by upper bounded by the corresponding full i.i.d. $\i$-cycles (which form  ${\tilde \Psi}_{j} (t)$ as explained  in  the proof of  Lemma \ref{Lemma_Psij}). 
We refer any    typical (upper bounding)
  full $\i$-cycle by ${\tilde \B}_j$.
Then
by  RRT, we  obtain the following for the first and last terms of the RHS of equation (\ref{Eqn_split_2nd_fraction}):  

\vspace{-4mm}
{\small\begin{eqnarray*}
		\lim_{t \to \infty} \frac{N_{B_j}^{\mue}( {\tilde \Psi}_{j} (t))}{{\tilde \Psi}_{j} (t)} \frac{t}{N_A^{\mue}(t)} &\le &
				\lim_{t \to \infty} \frac{N_{A}^{\mue}( {\tilde \Psi}_{j} (t))}{{\tilde \Psi}_{j} (t)} \frac{t}{N_A^{\mue}(t)}
	\ =  \   \frac{E\big(  N^{\mue}_A(\tilde{\B_j})  \big)}{E\big(\tilde{\B_j} \big)} \frac{1}{\lambda_\epsilon} \\  &=&  
		\frac{E \left[ N^{\mue}_A \big( \mathcal{ B}_j^{\mu_\epsilon} \big) \ ; \  \mathcal{ B}_j^{\mu_\epsilon} > A_\tau \cup \  \S_j^{\mu_\epsilon}  >   B_\tau  \right]}{E \left[ \mathcal{ B}_j^{\mu_\epsilon} ;  \mathcal{ B}_j^{\mu_\epsilon} > A_\tau \cup \  \S_j^{\mu_\epsilon}  >   B_\tau  \right] } \frac{1}{\lambda_\epsilon}\\
		&\leq & \frac{E \left[ N^{\mue}_A \big( \mathcal{ B}_j^{\mu_\epsilon} \big) \ ; \  \mathcal{ B}_j^{\mu_\epsilon} > A_\tau \right] + E \left[  N^{\mue}_A \big( \mathcal{ B}_j^{\mu_\epsilon} \big) \ ; \  \S_j^{\mu_\epsilon}  >   B_\tau  \right]}{E \bigg(  \mathcal{ B}_j^{\mu_\epsilon} ;   \mathcal{ B}_j^{\mu_\epsilon} > A_\tau \bigg) \lambda_\epsilon } .
\end{eqnarray*}}
By construction  we have:
$
N^{\mue}_A (   \mathcal{B}_j^{\mu_\epsilon} ) =  N^1_A ( {\mathcal B}_j^1)
$ (see (\ref{Eqn_Compare_mu_epsilon_schemes}))
and hence:
\begin{eqnarray*}
	\lim_{t \to \infty} \frac{N_{B_j}^{\mue}({\tilde  \Psi}_{j} (t))}{{\tilde \Psi}_{j} (t)} \frac{t}{N_A^{\mue}(t)}  \hspace{-2mm} &&
	\\
	&\leq&
	\frac{E \left[ N^1_A \big( \mathcal{B}_j^1\big) \ ; \  \mathcal B^{\mu_\epsilon}_j > A_\tau \right] + E \left[  N^1_A \big( \mathcal{B}_j^1\big) \ ; \  \S_j^{\mu_\epsilon}  >   B_\tau  \right]}{E \bigg(  \mathcal B_j^{\mu_\epsilon}; \mathcal B^{\mu_\epsilon}_j > A_\tau \bigg)}  \frac{1}{\lambda_\epsilon}\\
	&=& \frac{E \left[ N^1_A \big(  \mathcal{B}_j^1\big) \ ; \  \mathcal{B}_j^{\mu_\epsilon} > A_\tau \right] + E \left[  N^1_A \big(  \mathcal{B}_j^1\big) \ ; \  \S_j^{\mu_\epsilon}  >   B_\tau  \right]}{ \rho_\epsilon \ E \bigg(  \mathcal{ B}_j^{1}  ;  \mathcal B_j^{\mu_\epsilon} > A_\tau \bigg)} \\
	&=&  \frac{E \left[ N^1_A \big(  \mathcal{B}_j^1\big) \big( 1 - \exp( - \lambda_\tau  \frac{\mathcal B_j^1}{\mu_\epsilon}) \big) \right] + E \left[  N^1_A \big(  \mathcal{B}_j^1\big) \left ( 1 - \exp \left ( - \mu_\tau  \frac{\S_j^1}{\mu_\epsilon} \right ) \right )  \right]}{ \rho_\epsilon \ E \bigg[   \mathcal{ B}_j^{1}  \bigg( 1 - \exp \left ( - \lambda_\tau  \frac{\mathcal B_j^1}{\mu_\epsilon} \right ) \bigg ) \bigg] }.
\end{eqnarray*}
As $\mu_\epsilon \to \infty$, using the L'Hopitals rule
\begin{eqnarray*}
	\lim_{\mu_\epsilon \to \infty} \lim_{t \to \infty} \frac{N_{B_j}^{\mue}( {\tilde \Psi}_{j} (t))}{{\tilde \Psi}_{j} (t)} \frac{t}{N_A^{\mue}(t)} \hspace{-30mm} \\
	&\leq& \lim_{\mu_\epsilon \to \infty} 
	\frac{E \left[ N^1_A \big(  \mathcal{B}_j^1\big) \bigg(  \frac{\lambda_\tau \mathcal B_j^1}{\mu_\epsilon^2} \exp( - \lambda_\tau  \frac{\mathcal B_j^1}{\mu_\epsilon}) \bigg) \right] + E \left[  N^1_A \big( \mathcal{B}_j^1\big) \bigg( \frac{\mu_\tau \S_j^1}{\mu_\epsilon^2} \exp( - \mu_\tau  \frac{\S_j^1}{\mu_\epsilon}) \bigg)  \right]}{ \rho_\epsilon \ E  \left[   \frac{\lambda_\tau \left ( \mathcal B_j^1 \right )^2 }{\mu_\epsilon^2} \exp( - \lambda_\tau  \frac{\mathcal B_j^1}{\mu_\epsilon})  \right]}\\
	&=&\lim_{\mu_\epsilon \to \infty}
	\frac{E \left[ N^1_A \big(  \mathcal{B}_j^1\big) \bigg(  \lambda_\tau \mathcal B_j^1 \exp( - \lambda_\tau  \frac{\mathcal B_j^1}{\mu_\epsilon}) \bigg) \right] + E \left[  N^1_A \big( \mathcal{B}_j^1\big) \bigg( \mu_\tau \S_j^1 \exp( - \mu_\tau  \frac{\S_j^1}{\mu_\epsilon}) \bigg)  \right]}{ \rho_\epsilon \ E \left[ \lambda_\tau \left (\mathcal B_j^1 \right )^2 \exp( - \lambda_\tau  \frac{\mathcal B_j^1}{\mu_\epsilon}) \right]}\\
	&< & \infty,  
\end{eqnarray*}as the limit is  a finite constant, by {\bf A}.4 (note $\S_j^1 \le \mathcal{B}_j^1$ a.s., $ N^1_{B_j}  \le N^1_A$ etc).
Using above and equation (\ref{Eqn_split_2nd_fraction}) and Lemma \ref{Lemma_Psij}, we get,
\begin{eqnarray}
\lim_{\mu_\epsilon \to \infty} \lim_{t \to \infty} \frac{N_{B_j}^{\mue}( \Psi_{j} (t))}{N_A^{\mue}(t)} =  0 \mbox{ a.s. }
\end{eqnarray}
\kcmnt{Need to mention that all these moments exist for the $(\mue = 1)$ case.}

{\bf Third fraction of RHS of (\ref{Eqn_split_NB_By_NA}):} By assumption {\bf A}.3, one can serve (say) at maximum $\mathcal{ K}$
number of $\i$-customers at any time. Let $M(t)$ represent the total number of $\tau$-transitions during time $t$. Then one can upper bound the number of losses, $N_{\partial_j} (t)$, during $\t$-state transition  as: 
$$
N_{\partial_j} (t) \le  \sum^{M(t)}_{i=1} \mathcal{ K} = M(t) \mathcal{ K},
$$since one serves at maximum $\mathcal{ K} $ $\i$-customers at any time and hence the maximum drops per transition equals $\mathcal{ K}$.  As in the proof of Lemma \ref{Lemma_Psij},  $\lim_{t \to \infty } M(t) / t  \le 2\lambda_\t $. 
Thus 
\beq
\frac{N_{\partial_j} (t)}{N_A^{\mu_\epsilon} (t) }  &\le&   \mathcal{ K} \frac{M(t)}{ t }  \frac{t}{N_A^{\mu_\epsilon} (t)}  \mbox{ and thus } \\
\lim_{t \to \infty} \frac{N_{\partial_j} (t)}{N_A^{\mu_\epsilon} (t) }  & \le &  \frac{ 2\lambda_\t \mathcal{ K}  }{\lambda_\epsilon }  \to 0, \  \mbox{ as } \mu_\epsilon \to \infty, \mbox{ under SFJ limit}. 
\eeq
Combining the three fractions of  RHS of (\ref{Eqn_split_NB_By_NA}) we have the result:  $$
 \lim _{t \to \infty } \frac{N_{B_j}^{\mue}(t)}{N_A^{\mue}(t)} = \lim _{t \to \infty } \frac{ \left( N_{B_j}^{\mue}(\Lambda_{j}(t)) + N_{B_j}^{\mue}( \Psi_{j} (t)) \right) + N_{\partial_j} (t) }{N_A^{\mue}(t)}   \convas P_{B_j}  \pi_{j},  \mbox{ as } \mu_\epsilon \to \infty. 
$$
Hence the proof of the lemma.
	
\end{proof}

\begin{lemma}
	\label{Lemma_Psij}
	Fix any $j$. Then 
	$$
	\lim_{t \to \infty}   \frac{\Psi_j(t) }{t}   \le  \frac{ 2 \lambda_\tau  }{\mue}   
	\frac{E \left[ \mathcal B_j^1  \bigg(1 - \exp \big(-\lambda_\tau \frac{\mathcal B_j^1}{\mu_\epsilon} \big)\bigg) \ \right] + E \left[ \ \mathcal B_j^1  \bigg(1 - \exp \big( -\mu_\tau \frac{\S_j^1}{\mu_\epsilon} \big)\bigg) \ \right]}{E\left[ 1 - \exp \big(-\lambda\tau \frac{\mathcal B_j^1}{\mu_\epsilon} \big)\right]} .
	\mbox {  a.s.  }
	$$

	Further, as  $\mu_\epsilon \to \infty,$
	$$
	\lim_{t \to \infty}   \frac{\Psi_j(t) }{t}  \to 0 \mbox {  a.s.  }
	$$
\end{lemma}
{\bf Proof:} Let $M_j(t)$ represent the total number of transitions to the state $j$ during time $[0, t]$. Then,
$$
\frac{\Psi_j(t)}{t} =  \sum_{i=1}^ {M_j(t)}   \frac{  {\check \B}_{j,i} } {M_j(t)} \frac{M_j(t)}{t} ,
$$
 with ${\check \B}_{j,i}$ being  {\it the   residual $\epsilon$-cycle} before the end of the  $i$-th $\t$-transition,  occurred during $j$-cycles. By Lemma \ref{lemma:TauStability} there exists ${\bar \mu} < \infty$ and the process is stable for all $\mu_\i \ge {\bar \mu} $ and it is easy to see that limit $t \to \infty$ of the above term exists  ($\{{\check \B}_{j,i}\}_i$ are i.i.d. for any given $j$)  almost surely for all such $\mu_\i.$
Further one can upper bound as below:
\begin{eqnarray}
\label{Eqn_Psi_Up}
	0 \le \frac{\Psi_j(t)}{t} \le \sum_{i=1}^ {M_j(t)}   \frac{  \B_{j,i} } {M_j(t)} \frac{M_j(t)}{t} ,
\end{eqnarray} with $\B_{j,i}$ being  {\it the full $\i$-cycle upper-bounding the residual $\epsilon$-cycle} before the end of the  $i$-th $\t$-transition (as discussed in the beginning of the
proof of Theorem \ref{Thm_PB_for_one_subpolicy}),  occurred during $j$-cycles.  {\it Suffices to prove that the upper bound converges to zero, we rename $\Psi_j(t)$ to represent the upper bound itself.}
Note that $\{\B_{j,i} \}_i$ belonging to a particular $j$-cycle are i.i.d..
If $M_j(t)$ converges to  a finite constant as $t \to \infty$,  then clearly the proof of the  Lemma is done  for such sample paths. 

Now consider the sample paths in which, as $t \to \infty$,
$M_j(t) \to \infty$.  Then by 
using LLN, as $t \to \infty$, we get,
\begin{eqnarray}
\label{Eqn_Psi_Split}
\frac{\Psi_j(t)}{t}  \convas E^{\mu_\epsilon}  \left [ U \right]  \lim_{t \to \infty} \frac{M_j(t)}{t} \le E^{\mu_\epsilon}  \left [ \B_j \right]  \lim_{t \to \infty} \frac{M(t)}{t},
\end{eqnarray}
where  $M(t)$ represent the the total number of $\t$-transitions during time $t$, i.e., $M(t) = \sum_j M_j(t)$. 
With $N_{A}^\tau (t)$  representing the number of $\t$-arrivals during interval $[0, t]$ and by noting that the number of $\t$-departures at maximum equal the number of arrivals, we clearly have:
\begin{eqnarray}
\label{Eqn_M(t)}
\frac{M (t)}{t}  \le 2 \frac{N_{A}^\tau (t)}{t}  \stackrel{t \to \infty }{ \longrightarrow} 2 \lambda_\tau \mbox{ a.s. } 
\end{eqnarray}

Any $\B_{j,i}$ in equation~(\ref{Eqn_Psi_Up}) is a special $\i$-cycle,  which is interrupted by a $\t$-transition, either arrival or departure. Thus (with the rest of the notations
as in (\ref{Eqn_special_cycle})): 
\begin{eqnarray*}
	E^{\mu_\epsilon} \left[  \B_j \right] &=& E \left[ \mathcal B_j^{\mu_\epsilon} \ \big | \  \mathcal B_j^{\mu_\epsilon}  > A_\tau \cup \  \S_j^{\mu_\epsilon}  >   B_\tau  \right] 
	= \frac{E \left[\mathcal B_j^{\mu_\epsilon}  \ ; \  \mathcal B_j^{\mu_\epsilon}  > A_\tau \cup \  \S_j^{\mu_\epsilon}  >   B_\tau  \right]}{E \left[\mathcal B_j^{\mu_\epsilon} > A_\tau \cup \  \S_j^{\mu_\epsilon}  >   B_\tau  \right]}
	\\
	&\leq & \frac{E \left[ \mathcal{ B}_j^{\mu_\epsilon}  \  ; \  \mathcal{B}_j^{\mu_\epsilon}  > A_\tau \right] + E \left[   \mathcal{B}_j^{\mu_\epsilon}  \ ; \ \S_j^{\mu_\epsilon}  >   B_\tau  \right]}{E \left[ \mathcal{B}_j^{\mu_\epsilon} > A_\tau \right]}.
\end{eqnarray*}Using the tower property of conditional expectation  (e.g., $E[X] = E[E[X|Y] ]$), by conditioning on $\B_j^{\mu_\epsilon}$ and $ \S_j^{\mu_\epsilon}$,  and because $A_\tau$, $B_\tau$
are exponential random variables with respective parameters $\lambda_\tau$ and $\mu_\tau$:  
\begin{eqnarray*} 
	E^{\mu_\epsilon} \left[ \  \B_j \right] &\le &
	\frac{E \left[ \mathcal B_j^{\mu_\epsilon} \bigg(1 - \exp \big(- \lambda_\tau \mathcal B_j^{\mu_\epsilon} \big) \bigg) \ \right] + E \left[ \ \mathcal B_j^{\mu_\epsilon}  \bigg(1 - \exp \big(- \mu_\tau \S_j^{\mu_\epsilon} \big) \bigg) \ \right]}{E\left[ 1 - \exp \big(- \lambda\tau  \mathcal B_j^{\mu_\epsilon} \big) \right]} \\
	&=& \frac{E \left[ \frac{\mathcal B_j^1}{\mu_\epsilon} \bigg(1 - \exp \big(-\lambda_\tau \frac{\mathcal B_j^1}{\mu_\epsilon} \big)\bigg) \ \right] + E \left[ \ \frac{\mathcal B_j^1}{\mu_\epsilon}  \bigg(1 - \exp \big( -\mu_\tau \frac{\S_j^1}{\mu_\epsilon} \big)\bigg) \ \right]}{E\left[ 1 - \exp \big(-\lambda\tau \frac{\mathcal B_j^1}{\mu_\epsilon} \big)\right]} .
\end{eqnarray*}
By substituting the above inequality in (\ref{Eqn_Psi_Split}) and further referring to  (\ref{Eqn_M(t)}), we have the proof of the first inequality of Lemma.
For the last equality,   recall $\mathcal B_j^{\mu_\epsilon}  = \mathcal B_j^1 / \mu_\epsilon$ and $\S_j^{\mu_\epsilon} = \S_j^1 / \mu_\epsilon$.

By L'Hopitals rule, and using the differentiability properties of the
moment generating function, we get:
\begin{eqnarray*} 
	\lim_{{\mu_\epsilon} \to  \infty} \mu_\i E^{\mu_\epsilon} \left[ \  \B_j  \right] \hspace{-20mm} \\
	& \le  &
	\lim_{{\mu_\epsilon} \to  \infty}  \frac{E \left[  \mathcal B_j^1 \bigg(1 - \exp \big(-\lambda_\tau \frac{\mathcal{B}_j^1}{\mu_\epsilon} \big)\bigg) \ \right] + E \left[ \  {\mathcal B_j^1}   \bigg(1 - \exp \big(-\mu_\tau \frac{\S_j^1}{\mu_\epsilon} \big)\bigg) \ \right]}{E\left[ 1 - \exp \big(-\lambda\tau \frac{\mathcal{B}_j^1}{\mu_\epsilon} \big)\right]} \\
	&=& \lim_{{\mu_\epsilon} \to  \infty}  \frac{E \left[ {  \mathcal  B_j^1}  \bigg( \frac{ \lambda_\tau \mathcal  B_j^1 }{\mu_\epsilon^2} \exp \big(-\lambda_\tau \frac{\mathcal B_j^1}{\mu_\epsilon} \big) \bigg) \ \right] + E \left[  { \mathcal B_j^1}  \bigg( \frac{\mu_\tau \S_j^1 }{\mu_\epsilon^2} \exp \big(-\mu_\tau \frac{\S_j^1}{\mu_\epsilon} \big) \bigg) \ \right]}{E\left[ \frac{ \lambda_\tau  \mathcal  B_j^1}{\mu_\epsilon^2} \exp \big(-\lambda_\tau \frac{ \mathcal  B_j^1}{\mu_\epsilon} \big)\right]} \\
	&=& \lim_{{\mu_\epsilon} \to  \infty}  \frac{E \left[    \bigg(  \lambda_\tau ({\mathcal  B}_j^1)^2  \exp \big(-\lambda_\tau \frac{\mathcal B_j^1}{\mu_\epsilon} \big) \bigg) \ \right] + E \left[    \bigg(  \mu_\tau \S_j^1 { \mathcal B_j^1}  \exp \big(-\mu_\tau \frac{\S_j^1}{\mu_\epsilon} \big) \bigg) \ \right]}
	{E\left[ \  \lambda_\tau  {\mathcal  B}_j^1  \exp \big(-\lambda_\tau \frac{ \mathcal  B_j^1}{\mu_\epsilon} \big)\right]} < \infty .
	\eeq
	Thus 
	\beq
	\lim_{\mue \to \infty}  E^\mue [U] \leq \lim_{\mue \to \infty}   \frac{1}{\mue}  \mu_\i E^{\mu_\epsilon} \left[ \  \B_j  \right] = 0.
	\eeq
	Thus, we get from Equation (\ref{Eqn_Psi_Split}),
	\begin{eqnarray*}
		\lim_{\mu_\epsilon \to \infty} \lim_{t \to \infty} \frac{\Psi_j(t)}{t}   = 0.
	\end{eqnarray*} \eop
	\\
	%
	%
	%

\section{\revision{Proof of Theorem \ref{Thm_withou_A1}}} 
\label{app:withoutA1}

\revision{
In this appendix, we describe the changes required in the proofs of
Theorems~\ref{Thm_Lim_SD_Conv}-\ref{Thm_PB} to establish
Theorem~\ref{Thm_withou_A1}. Note that we no longer make Assumption
{\bf A}.1, but simply assume that eager job sizes are exponentially
distributed.

Without Assumption {\bf A}.1, one can not describe the pre-limit
system using a one-dimensional Markov process. However with
exponential service times for both eager and tolerant customers, one
can describe the system evolution via a two-dimensional Markov chain
$(X_\tau (t), X_\i (t)),$ capturing the number of tolerant and eager
customers in the system. Note that this Markov chain undergoes
transitions due to $\i$ as well as $\t$ customers. As noted before,
when the occupancy of the tolerant changes, a new eager sub-policy
comes into effect immediately, potentially during an ongoing eager
busy period. This new sub-policy might result in scheduling
rearrangements, based on the number of eager customers in the system.

The main difference without Assumption {\bf A}.1 is that once the
$\tau$ occupancy changes, the first $\i$-busy cycle can start with
(random) $c$ number of customers, where
$0 \le c \le {\cal K} := \lfloor 1/c_{min} \rfloor < \infty. $ Let the
duration of this partial first $\i$-busy cycle be represented by
$\B_{j}^{\mue, c} $, when the system is operating under $j$-th
sub-policy; basically, this is the time period that starts immediately
after the $\t$-system changes to state $j$, and lasts till the
$\i$-system becomes empty (or till the $\t$-state changes again,
whichever happens first).\footnote{Recall that \emph{complete}
  $\i$-busy cycles start when the $\i$-system becomes empty and last
  till the end of the next $\i$-busy period.
} Aside from this first partial $\i$-busy-cycle, subsequent $\i$-busy
cycles are referred to using the same notations as before.

While the two-dimensional Markovian description
$(X_\tau (t), X_\i (t)),$ with a random number~$c$ of eager customers
present at times of $\t$-transitions is more challenging to analyse,
we show below that the first dimension of this process can be upper
and lower bounded by one-dimensional Markov processes (as in
\cite{ANOR}). These bounding processes are described next.


{\bf Upper and Lower bounding systems:} Now we construct two systems
for each $\mue$, which sandwich the actual $\t$-system pre-limit. In
both the bounding systems we consider that the eager system,
immediately after each tolerant change, starts with maximum possible
(partial) eager busy cycle: generate one busy cycle $\B_{j}^{\mue, c}$
starting with each $c$ and with coupled eager arrivals, job-sizes
etc., for the given sub-policy (say sub-policy~$j$) and let
$\B_{j}^{\mue, *}:= \max_{0 \le c \le {\cal K}} \B_{j}^{\mue, c}$ be
duration for which the first (hypothetical) eager busy cycle lasts in
the two systems. It is immediate that this max-busy cycle upper bounds
those starting with any $c$-number (as in original system), when they
are served with the same sub-policy.

In the upper bounding system, immediately after the $\t$-occupancy
changes to~$j$, the tolerant queue receives zero service for duration
$\B_{j}^{\mue, *}$. After this period, service of the tolerant queue
in subsequent $\i$ busy cycles continues exactly as in the original
system.
The lower bounding system is almost similar to the upper bounding
system, except that the tolerant queue is served at the maximum
possible rate (1, given the unit server speed assumed) over the first
hypothetical partial busy cycle $\B_{j}^{\mue, *}.$
We then use appropriate coupling rules as in \cite{ANOR}: \\

\noindent $\bullet$ Tolerant inter-arrival times and job-sizes are
always coupled in all three systems. However, the eager quantities are
coupled in the original and any bounding system, only when the two
systems are in the same $\t$-state, as explained below.\\

\noindent $\bullet$ The original and the two bounding systems start
with the same initial system state;\\

\noindent $\bullet$ the number of $\t$-customers in the original
system is less than or equal to that in upper system, and greater than
or equal to that in the lower system, as time progresses (because of
the service differences for the tolerant queue in the first partial
eager busy-cycles across the three systems);\\

\noindent $\bullet$ if at any time point the upper system
meets\footnote{The tolerant number in upper system can become strictly
  bigger than that in original systems, but further increase may be
  slower in upper system due to Markov policies, and then at some
  tolerant departure the $\t$-number in upper system can equal that in
  the original system.}  the original system, we couple the two
systems by using eager random variables (inter-arrival times, service
times, etc.) that define the original system in that $\t$-cycle to
define the upper system $\t$-cycle;\\
  
\noindent $\bullet$ specifically, we start with the maximum partial
busy cycle in upper system corresponding to the sub-policy dictated by
the new tolerant state, after coupling (also) the eager inter-arrival
times and job sizes; observe that under this construction, the first
$\i$-busy cycle in the original system (in the $\t$-cycles in which
upper system meets the original system) is smaller than or equal to
the corresponding maximum partial busy cycle in upper system.\\

These details are similar to those in \cite{ANOR}, where one can also
find more details. Similar coupling ideas can be used to construct the
lower bounding system.

It is easy to observe that the tolerant queue evolution in the two
bounding systems can be modeled as M/G/1 systems with state dependent
birth-death transition probabilities, as the system considered in the
paper with Assumption~{\bf A.}1.  That is, they can be represented by
one-dimensional Markov processes.  All the results can be extended to
the two bounding systems and we would have equivalents of Theorems
\ref{Thm_Lim_SD_Conv}-\ref{Thm_tau} and Lemma~\ref{lemma:TauStability}
for the two systems, if we can show that the impact of the first
maximum partial busy cycle at the start of every tolerant change
vanishes as $\mue \ra \infty$ under the SFJ limit. We show below that
this is the case.

{\bf Convergence of the bounding systems:} This proof goes through
along similar lines as with {\bf A}.1, except for few important
differences. We only mention the differences here.  We first show that
the max busy cycles $\B_{j}^{\mue, *}$ converge to zero, as
$\mu_\i \to \infty$, the moments are upper bounded by a common uniform
bound, etc. For extending Lemmas
\ref{lemma:bp_upper_bound}-\ref{lemma:mginftyBPbound}, we observe that
the max busy cycles can be uniformly upper-bounded (uniformly over
sub-policy $j$) by that in the M/G/$\infty$ queue of
Lemma~\ref{lemma:mginftyBPbound}, which starts with ${\cal K}$
customers.
  
For extending Lemma~\ref{lemma:uniform}, the transition probabilities
also depend upon max busy cycles. Towards this, first define
$q_j^{\mue, *}$ to represent the backward transition probability of
tolerant birth-death process when one starts it with corresponding max
busy cycle.  Then one can lower bound the probability $q_j^{\mue, *}$
using
$$\bar{q}^*_j := \prob{\S_{j,1}^{\mue,*} +\sum_{k = 2}^{N-1} \S_{j,k}^{\mue} > B_{\t}} = (1-r_*) {\bar q}_j + r_*, \  \    r_* := \prob { \S_{j,1}^{\mue,*}  > B_\tau},   $$  
with ${\bar q}_j$ as defined previously and where $\S_{j,1}^{\mue,*} $
equals the service available to the tolerant queue during the first
$\i$ max busy cycle $\B_{j}^{\mue, *}$ (which exactly equals 0 for
upper system and $\B_{j}^{\mue, *}$ in the lower system for any
$\mue$). It is clear that $r_*$ converges to zero (uniformly in $j$)
and the remaining proof goes through, as before.  The proof for the
upper bound of Lemma~\ref{lemma:uniform} can be handled similarly.
Observe here that Lemmas
\ref{Lemma_Intermediate_result}-\ref{Lemma_Intermediate_result_1} are
used by Lemma~\ref{lemma:uniform} only for the subsequent $\i$-busy
cycles, which are same as that with Assumption~{\bf A}.1, and hence
would not require any changes.  The proof of
Lemma~\ref{Lemma_Intermediate_result_2} holds also for the two
bounding systems, because the convergence in Lemma~\ref{lemma:uniform}
is ensured and same is the case with Lemma~\ref{lemma:TauStability}.

Thus, by the extended Lemma~\ref{lemma:TauStability} there exists
${\bar \mu}$ such that the upper and lower bounding systems are stable
for all $\mue \ge {\bar \mu}$.  Further the upper and lower stationary
distributions converge to the same quantity as $\mue \to \infty$ (by
extended versions of Theorems~\ref{Thm_Lim_SD_Conv}-\ref{Thm_tau}).

{\bf Existence of stationary distribution of original system:} We
derive the existence using embedded chains, the discrete time Markov
chains obtained by considering values of the continuous time Markov
process, immediately after the transition epochs.

By extended Lemma~\ref{lemma:TauStability} there exists a ${\bar \mu}$
such that the upper bounding system is stable for all
$\mue \ge {\bar \mu}$. For all such $\mue$, the embedded tolerant
chain in the upper system visits the~$0$ state infinitely often (with
probability one) and within integrable number of epochs (by positive
recurrence and existence of stationary distribution). This implies (by
the bounding of the tolerant occupancy between the original and upper
system) that the embedded tolerant chain in the original system also
visits $0$ state infinitely often (with probability one) and within an
integrable number of epochs. This in turn implies the original
(two-dimensional) embedded Markov chain visits the set
$ {\cal O}:= \{ (0, c), \mbox{ for some } c \le {\cal K} \}$
infinitely often (w.p.1) and within integrable number of epochs.  By
irreducibility, this implies the same for the $(0,0)$ state, as every
time the chain visits any state of the form $(0,c),$ it has a positive
probability of visiting state $(0, 0)$ (uniformly bounded away from 0
across all $c$) and hence visits state $(0,0)$ after at most a
geometric number of visits to the set~${\cal O}.$ This establishes the
existence of a unique stationary distribution
for the original two-dimensional embedded Markov chain, which in turn
implies the existence of stationary distribution for the two
dimensional continuous time Markov process; observe that the
transition rates between all the state pairs can be uniformly upper
bounded by a finite value for any given $\mue.$
 
By regular renewal arguments, this implies almost sure existence of
the time limits:
 \begin{eqnarray}
\label{Eqn_SD_marginal}
\lim_{T \to \infty}  \frac{1}{T} \int_0^T  \indc_{ \{X_\tau(t) = j \}}  =  \pi_j^\mue, \mbox{ for any } j, 
\end{eqnarray}
where $\pi_j^\mue$ is the marginal of the stationary distribution
corresponding to tolerant queue in original (two-dimensional) system.

{\bf Convergence of stationary distribution as $\mue\to \infty$:} The
limit of the $\t$-performance of the two bounding systems, as
$\mue \to \infty$ is exactly the same, because by extended
Lemma~\ref{lemma:uniform} both the sets of transition probabilities
converge to the same limit transition probabilities uniformly.
Furthermore, time averages corresponding to the $\t$-occupancy in the
original system can be bounded by the corresponding time averages in
the upper system (occupancy denoted by $X_\t^U (\cdot)$) and the lower
system (occupancy denoted by $X_\t^L (\cdot)$) as follows:
$$
\lim_{T \to \infty} \frac{1}{T} \int_0^T \indc_{ \{X^{L}_\tau(t) \le j
  \}} \le \lim_{T \to \infty} \frac{1}{T} \int_0^T \indc_{ \{X_\tau(t)
  \le j \}} \le \lim_{T \to \infty} \frac{1}{T} \int_0^T \indc_{
  \{X^U_\tau(t) \le j \}} \mbox{ for any } j,
$$
and hence the marginal stationary distribution (corresponding to
tolerant occupancy) of the original system, converges to the common
limit of the two bounding systems. This implies the statement related
to the tolerant performance in Theorem \ref{Thm_Lim_SD_Conv}, because
(for each $j$):
\begin{eqnarray*}
\sum_{i\le j} \pi_i  = \lim_{\mue \to \infty } \lim_{T \to \infty}  \frac{1}{T} \int_0^T  \indc_{ \{X^{L}_\tau(t) \le  j \}}  \le  \lim_{\mue \to \infty }    \lim_{T \to \infty}  \frac{1}{T} \int_0^T  \indc_{ \{X_\tau(t) \le  j \}}  \hspace{30mm}  \\
 \le
 \lim_{\mue \to \infty }   \lim_{T \to \infty}  \frac{1}{T} \int_0^T  \indc_{ \{X^U_\tau(t) \le  j \}} =  \sum_{i \le j} \pi_i, \  \    \mbox{ w.p.1. }
\end{eqnarray*}
Using similar bounding arguments, and common limit of the two bounding
systems, the statement of Theorem \ref{Thm_tau} also gets extended.

{\bf Changes required for results related to eager performance:} There
are no changes in the main proof of Theorem \ref{Thm_PB}, we only need
to discuss the changes required in the proof of
Theorem~\ref{Thm_PB_for_one_subpolicy}.

For extending the proof of Theorem \ref{Thm_PB_for_one_subpolicy}, one
needs to construct additional sample paths for coupling arguments.
One first needs to construct ${\cal K}+1$ additional sample paths of
$\i$-busy cycles, with $\mue = 1$, one for each sub-policy~$j,$ and
one for each $0 \le c \le {\cal K}$.  Using these ${\cal K}+1$ sample
paths, we need to construct a $({\cal K}+2)$th sample path of max busy
cycles, one for each $j$.  Note again that each such max busy cycle
upper bounds the corresponding busy cycle starting with $c$ number of
$\i$-customers, for any $0 \le c \le {\cal K}$. Whenever the
$\t$-state changes, the first $\i$-busy cycle is coupled using one of
the above ${\cal K}+1$ sample paths, depending upon the random $c$,
that represents the number of eager customers at the time of the
$\t$-transition. Each such first partial busy cycle is is upper
bounded by the corresponding max busy cycle, irrespective of the
random~$c$; this ensures the required coupling (and uniform upper bounding across various $c$) across various $\mue$,
even without Assumption~{\bf A}.1.

The process $\Psi_{j} (t)$ for each $j$-sub-policy now includes the
starting as well as the ending partial busy cycles.  But all the
starting partial $\i$-busy cycles can be upper bounded by the
corresponding max busy cycles and one can show the influence of them
to converge to zero, in the same manner as is done with the ending
busy cycles; this is possible because we have already extended the
proofs of the required uniform bounds (converging to zero or second moments bounded
as required) on all the required moments of the max busy cycles, as
already discussed in the initial parts of this section. All the bounds
that were shown to converge to zero, should now be made up of the
corresponding max busy cycles, to handle the part of $\Psi_{j} (t)$
coming from the starting partial $\i$-busy cycles.  The same is true
for the bounds and arguments considered in the proof of
Lemma~\ref{Lemma_Psij}.

If there any drops due to change in eager sub-policy at tolerant transition epochs, they can be upper bounded by  $N_{\partial_j} (t) $ and can be handled as in original proof. 
\eop

}

%
 
\section{Proofs of Limit Pareto Frontier}  
\label{Appendix_pareto}

\noindent \textbf{Proof of Theorem \ref{Thm_Opt_Policy}:}  We prove this theorem by proving the following steps:
 
(a)  The optimizing problem defined by  (\ref{Eqn_Pareto_optimal}) is equivalent to the following modified optimizing problem in terms of blocking probabilities $\{ d_i \}_{i \geq 0}$:

\begin{eqnarray*}    
\max_{ \phi = \{ d_i \} } f(\phi),  &\mbox{ with }&   f(\phi)  := \sum_i \Pi_{j = 0}^i  \ \rho^\phi_j ,\ \  \mbox{ such that } \\
g(\phi)  \le 0 ,&\mbox{ with }&  g(\phi) :=
\sum_i   (i - C) \ \ \ \Pi_{j = 0}^i \  \rho^\phi_j  \leq 0.  
\end{eqnarray*}
  

(b) Let $\phi^*$ be an optimizing policy. Then, constraint is satisfied with equality at $\phi^*$. \\
(c) Consider any policy $\phi = (d_0, d_1 \cdots)$.  For any $i$,   if $d_i > \dmin$ and $d_{i+1}  < \dmax$, then there exists a policy ${\tilde \phi}$ which strictly improves upon 
$\phi$, i.e., 
$$
f({\tilde \phi} ) >  f(\phi)  \mbox{ and }   g({\tilde \phi} ) \le 0.
$$
d) From part (c), it is clear that the optimal policy $\phi^*$ is  monotone, i.e.,  $d_i^* \le  d_{i+1}^*$ for any $i$. In fact from part (c), optimal policy  $\phi^* = \{ d_1^*, d_2^* \cdots \}$ is of the type as given below: 
\begin{eqnarray}
\label{Eqn_pareto_optimal_policy_appendix}
d_i^* 
= \left \{ \begin{array}{llll}
\dmin   &\mbox{ if }  i < L \\
d          & \mbox{ if }   i = L \\
\dmax  & \mbox{ if }  i > L,
\end{array} \right .
\end{eqnarray}for some   $1 \le L < \infty$ and $\dmin \le  d \le \dmax. $  We then identify $(L, d)$ for the given $C$ in the last step.  
 
 We now provide the proof of all the steps, one after the other.

\textbf{Proof of Part (a):}
Consider the SM policy $\phi = \left( d_0, d_1, \cdots, d_n, \cdots \right)$.  We rewrite the optimization problem  (\ref{Eqn_Pareto_optimal}), here,  for convenience: 
\begin{eqnarray}
\label{Eqn_optimal_h}
\nonumber
\min_{\phi} &   P_B  (\phi) &\mbox{ such that } \ \ \   E[N] (\phi) \le C \mbox{, i.e., } \\
\min_{\phi  }  & \sum_{i=0}^\infty d_i \frac{\h_i^\phi}{ \sum_{l \geq 0} \h_l^\phi}  &  \mbox{ such that }  \\
&  \sum_{i=0}^\infty i \frac{\h_i^\phi}{ \sum_{l \geq 0} \h_l^\phi}   \le C,
\end{eqnarray}
\begin{eqnarray}
 \mbox{ with } \ \  \h_0^{\phi}  = 1, \ \mbox{ and for $i \geq 1$,  }\h_i^{\phi} = \prod_{1 \le j \le i}  \rho_j^{\phi} \ \mbox{,  } \ \ \rho_j^{\phi}  :=  \frac{\lambda_\tau }{ c + \mu_\tau \rho_\epsilon d_j} \ \mbox{and} \ \ c = \mu_\tau \left( 1 - \rho_\epsilon \right).
\nonumber
\end{eqnarray}

In other words, we have to optimize: 
\begin{eqnarray*}
	\min_\phi    \sum_i d_i \frac {\h_i^{\phi} }{ \sum_k \h_k^{\phi} } \mbox{ such that }  
	\sum_i   i \frac {\h_i^{\phi} }{ \sum_k \h_k^{\phi} }  \le C.
\end{eqnarray*} 
As $\rho_\epsilon,$ $ \mut$ and $c$ are constants, optimizing above is equivalent to 
\begin{eqnarray*}
	\min_\phi   \sum_{i=0}^\infty (c + \mut \rho_\epsilon d_i)  \frac {\h_i^{\phi} }{ \sum_{k=0}^\infty \h_k^{\phi} }  \ \ \mbox{ such that } \\
	\sum_{i=0}^\infty   i \frac {\h_i^{\phi} }{ \sum_{k=0}^\infty \h_k^{\phi}  }  \le C.
\end{eqnarray*} 
For all $i \ge 1$, as $c + \mut \rho_\epsilon d_i = \lambda_\tau / \rho_i^\phi =  \lambda_\tau  \h_{i-1}^{\phi} / \h_i^{\phi} $ , minimizing  the above problem is equivalent to minimize the following problem (note that $\lambda_\tau$ is a constant and $\h_0^\phi = 1$), 
\begin{eqnarray*}
	\min_\phi    \left (   \frac{c + \mut \rho_\epsilon d_0}{ \sum_{k=0}^\infty \h_k^{\phi}  } +   \sum_{i=1}^\infty   \frac {\lambda_\tau \h_{i-1}^{\phi} }{ \sum_{k=0}^\infty \h_k^{\phi} } \right ) \mbox{ such that } \\
	\sum_i   i \frac {\h_i^{\phi} }{ \sum_{k=0}^\infty \h_k^{\phi} }  \le C.
\end{eqnarray*} 
It is immediate to show that $d_0^* = \dmin$ (only $P_B$ depends upon $d_0$ and $E[N]$ is independent of it) , thus equivalently we consider the following:
\begin{eqnarray*}
	\min_\phi    \left (   \frac{c + \mut \rho_\epsilon \dmin }{  \sum_{k=0}^\infty \h_k ^{\phi} } + \lambda_\tau \right ) \mbox{ such that } \\
	\sum_i   i  \h_i^{\phi}   \le C  \sum_{k=0}^\infty \h_k^{\phi}.
\end{eqnarray*} 
The above is equivalent to the following simplified version: 
\begin{eqnarray}
\max_\phi    f(\phi)  &\mbox{ with } f(\phi) := \sum_i \h_i^{\phi}  \mbox{ such that } \nonumber \\
g(\phi) \le 0,  &\mbox{ with } g(\phi) :=  \sum_{i=0}^\infty    (i-C)  \h_i^{\phi}  .
\label{Final_optimal}
\end{eqnarray} 
This proves the first part.\\
 \textbf{Proof of Part (b)}
Consider any SM policy $\phi \coloneqq \left(\dmin, d_1, \cdots, d_i, \cdots \right )$ (note $d_0 = \dmin$ is fixed at $\dmin$ in view of Part(a)). Since there is one-one onto relation between $\rho_j^\phi$ and $d_j$ (as $\rho_j^\phi =  \lambda_{\tau} / (c + \mu_\tau \rho_\epsilon d_{j})$) for every $j$, we consider optimization with respect to $\rho_j$'s	 and  redefine the policy as, $\phi \coloneqq \left(\rmax, \rho_1^\phi, \cdots, \rho_i^\phi, \cdots \right )$, where $\rmax$ corresponds minimum blocking $\dmin$ (given by (\ref{Eqn_rhomax_min})). 
This redefinition is only for this proof.   Thus one can re-write the optimization problem as (since $\rho_0 = \rmax$ is fixed):
\begin{eqnarray*}    
	\max_{ \phi = ( \rho_1^\phi, \rho_2^\phi, \cdots, \rho_i^\phi, \cdots  ) } f (\phi) \ \  \mbox{ such that }  
	g(\phi) \leq 0   \mbox{,  where, } \\
	f(\phi) :=    \sum_i  \prod_{j=1}^i \rho_j^\phi \ \ \mbox{ and }  \ \  g(\phi) :=\sum_i   (i - C)   \prod_{j=1}^i \rho_j^\phi    .  
\end{eqnarray*}
Let $\phi^* = \left(\rho_1^*, \rho_2^*, \cdots, \rho_i^*, \cdots \right )$  represent the optimal policy of the optimizing problem defined above. We want to show that constraint of the optimization problem is satisfied at $\phi^*$.
Let us assume on contrary that constraint is not satisfied at $\phi^*$, that means $g(\phi^*) < 0$. Then there exists an $i$ such that $\rho_i^* < \rmax$, as $C < \rmax / (1-\rmax)$.

First, consider the case when there exists one such $i$, and further  $i \geq C$. By Lemma (\ref{Lemma_epsilon}), we can get a policy $\phi^\epsilon$ in the feasible region which $f(\phi) < f(\phi^\epsilon)$. This contradicts the fact that $\phi^*$ is an optimal policy. The contradiction arises because of the wrong assumption that constraint is not satisfied at $\phi^*$. \\
Next, consider the case when $i <C$. Define $\bar{i}$ as
$$
\bar{i} = \min \{ i: \rho_j^* = \rmax \ \ \forall j \geq i \}.
$$
Clearly, $0 < \bar{i} \leq C$ and note $\rho_{\bar{i} - 1} < \rho_{\bar{i}}$. Using Lemma (\ref{Lemma_swap}), we can get an improved policy by swapping  $\rho_{\bar{i} - 1}$ and $\rho_{\bar{i}}$. This again leads to the contradiction that $\phi^*$ is an optimal policy and therefore, constraint is satisfied at the optimal policy.\\

{\bf Proof of Part (c) and (d):} 
Consider the case  where  $\rmin < \rho^\phi_{i+1}$  and $\rho^\phi_i  < \rmax$ for some $i$.   Then we claim that, there exists an alternate policy, say $\tilde{\phi}$, which is better than $\phi$, i.e., such that $f(\phi) <  f(\tilde{\phi})$ and $g(\phi) \ge   g(\tilde{\phi})$.\\
The alternate policy, say  ${\tilde  \phi } := \{ \tilde{\rho_i}\}_{i \geq 0}$, would defer only in $i$ and $i+1$ components, i.e.,  such that $\tilde{\rho_j} = \rho^\phi_j, \ \ \forall j \neq i, i+1$ and $\tilde{\rho_i} = \rho^\phi_i + \epsilon, \tilde{\rho_{i+1}} = \rho^\phi_{i+1} - \tilde{\epsilon}$,  for some positive $\epsilon$ and $\tilde{\epsilon}$.    And we would need the following to make it a valid policy:
\begin{eqnarray}
\label{Eqn_conditions}
0 < \epsilon <  \rmax - \rho^\phi_i   \mbox{ and }  0 < {\tilde \epsilon} <    \rho^\phi_{i+1} - \rmin.
\end{eqnarray}
First, consider the difference in the objective functions under the two policies: 
\begin{eqnarray}
\label{Eqn_star}
f(\tilde{\phi}) - f(\phi) &=&  \sum_j \tilde{\h_j}^\phi - \sum_j \h_j^\phi = \h_{i - 1}^\phi   \left ( ( \tilde{\rho_i} - \rho^\phi_i) + \sum_{j \geq i+1}^{\infty} \left  (\prod_{k = i}^j \tilde{\rho_k} - \prod_{k = i}^j  \rho^\phi_k \right ) \right )  \nonumber\\ \nonumber
&=& \h_{i - 1}^\phi \bigg( (\tilde{\rho_i} - \rho^\phi_i) + (\tilde{\rho_i}\tilde{\rho_{i+1}} - \rho^\phi_i \rho^\phi_{i+1})   +  \sum_{j \geq i+2}^{\infty} (\tilde{\rho_i}\tilde{\rho_{i+1}} - \rho^\phi_i \rho^\phi_{i+1})  \prod_{k = i+2}^j  \rho^\phi_k \bigg) \nonumber \\
&=& \h_{i - 1}^\phi \bigg( (\tilde{\rho_i} - \rho^\phi_i) +  (\tilde{\rho_i}\tilde{\rho_{i+1}} - \rho^\phi_i \rho^\phi_{i+1})\bigg(1 + \sum_{j \geq i+2}^{\infty} \prod_{k = i+2}^j  \rho^\phi_k \bigg) \bigg).
\end{eqnarray}
In a similar way, the difference in the constraints is:
\begin{eqnarray}
\label{Eqn_diff_constraints}
g(\tilde{\phi}) - g(\phi)  &=& \sum_j \tilde{\h_j}^\phi (j - C) \  - \ \sum_j \h_j^\phi (j - C)= \sum_j (\tilde{\h_j}^\phi - \h_j^\phi) (j - C)\nonumber  \\
&& \hspace{-22mm} = \ \h_{i-1}^\phi \bigg( (\tilde{\rho_i} - \rho^\phi_i)(i - C) +    (\tilde{\rho_i}\tilde{\rho_{i+1}} - \rho^\phi_i \rho^\phi_{i+1}) \bigg( (i+1-C) + \sum_{j \geq i+2} (j - C) \prod_{k = i+2}^{j} \rho^\phi_k    \bigg) \bigg). \hspace{10mm}
\end{eqnarray}
The above difference is equal to 0, i.e., the constraint function remains the same under  $\tilde{\phi}$,  if the following holds,
\beq
\tilde{\rho_i}\tilde{\rho_{i+1}} - \rho^\phi_i \rho^\phi_{i+1} = \frac{-(\tilde{\rho_i} - \rho^\phi_i)(i - C)}{(i+1-C) + \sum_{j \geq i+2} (j - C) \prod_{k = i+2}^{j} \rho^\phi_k}  .
\eeq
With such a policy the difference in objective functions (see equation (\ref{Eqn_star})) equal: 
\begin{eqnarray*}
	f(\tilde{\phi}) - f(\phi)   
	&=& \h_{i - 1}^\phi \bigg( (\tilde{\rho_i} - \rho^\phi_i) - \frac{(\tilde{\rho_i} - \rho^\phi_i)(i - C)\big(1 + \sum_{j \geq i+2}^{\infty} \prod_{k = i+2}^j  \rho^\phi_k \big)}{(i+1-C) + \sum_{j \geq i+2} (j - C) \prod_{k = i+2}^{j} \rho^\phi_k}  \bigg) \\ 
	&& \hspace{-10mm} = \
	\frac{\h_{i - 1}^\phi (\tilde{\rho_i} - \rho^\phi_i)}{(i+1-C) + \sum_{j \geq i+2} (j - C) \prod_{k = i+2}^{j} \rho^\phi_k}\bigg( (i+1-C) + \sum_{j \geq i+2} (j - C) \prod_{k = i+2}^{j} \rho^\phi_k \  \\
	&& \hspace{60mm} - \  (i - C) \  - \ \sum_{j \geq i+2} (i - C) \prod_{k = i+2}^{j} \rho^\phi_k   \bigg)\\
	&& \hspace{-10mm} = \ \frac{\h_{i - 1}^\phi (\tilde{\rho_i} - \rho^\phi_i)}{(i+1-C) + \sum_{j \geq i+2} (j - C) \prod_{k = i+2}^{j} \rho^\phi_k} \bigg( 1 + \sum_{j \geq i+2} (j - i) \prod_{k = i+2}^{j} \rho^\phi_k  \bigg ).
\end{eqnarray*}
{\bf Case with $i > C$:}  From the above equation, it is clear that we need $\tilde{\rho_i} - \rho^\phi_i > 0$,   for improvement in the objective function. That is, we want $\epsilon > 0$ (this is obviously true). 
Now, for  all $\epsilon > 0$, we get an $\tilde{\epsilon} $ strictly positive as, 
\begin{eqnarray}
\tilde{\rho_i}\tilde{\rho_{i+1}} - \rho^\phi_i \rho^\phi_{i+1} &=& \epsilon\rho^\phi_{i+1} - \tilde{\epsilon}\rho^\phi_i - \epsilon \tilde{\epsilon} =  \frac{-\epsilon(i - C)}{(i+1-C) + \sum_{j \geq i+2} (j - C) \prod_{k = i+2}^{j} \rho^\phi_k} \mbox{ and thus}  \nonumber  \\
\tilde{\epsilon}(\rho^\phi_i + \epsilon) &=& \epsilon\rho^\phi_{i+1}  + \frac{\epsilon(i - C)}{(i+1-C) + \sum_{j \geq i+2} (j - C) \prod_{k = i+2}^{j} \rho^\phi_k} \nonumber \\
\Rightarrow \tilde{\epsilon}  &=& \frac{\epsilon }{\rho^\phi_i + \epsilon}  \left ( \rho^\phi_{i+1 } + \frac{ (i - C)}{  \left((i+1-C) + \sum_{j \geq i+2} (j - C) \prod_{k = i+2}^{j} \rho^\phi_k \right)} \right ) > 0. 
\label{Eqn_Condition2}
\end{eqnarray}
If $i > C$ (more specifically if the last term in the RHS is positive), the above ${\tilde \phi}$ provides the required improvement policy as one can chose $\epsilon, {\tilde{\epsilon}} > 0$
sufficiently small to satisfy (\ref{Eqn_conditions}). \\
{\bf Case with  $i < C$:}  Further, if $\rho^\phi_i < \rho^\phi_{i+1}$, then by Lemma \ref{Lemma_swap} we can get the required improved policy.
We are left with the case when $i < C$ and    $\rho^\phi_i = \rho^\phi_{i+1}$  with $\rmin < \rho^\phi_i < \rmax$.   
If now   the numerator of (\ref{Eqn_Condition2}) is negative, i.e., if
$$
(i - C) +    \rho^\phi_{i+1}  \bigg( (i+1-C) + \sum_{j \geq i+2} (j - C) \prod_{k = i+2}^{j} \rho^\phi_k    \bigg) \le  0,
$$
one can improve $\phi$ by considering a ${\tilde \phi}$ that differs from $\phi$ only in $i$-th component,  and with  ${\tilde \rho}_i = \rmax$.  With such a ${\tilde \phi}$, \     a) the constraint remains satisfied, as  (by substituting ${\tilde \rho}_{i+1} = \rho^\phi_{i+1}$ in (\ref{Eqn_diff_constraints}))
\begin{eqnarray}
\label{Eqn_diff_constraints2}
g(\tilde{\phi}) - g(\phi)   
= \h_{i-1}^\phi (\tilde{\rho_i} - \rho^\phi_i)  \bigg( (i - C) +     \rho^\phi_{i+1} \bigg( (i+1-C) + \sum_{j \geq i+2} (j - C) \prod_{k = i+2}^{j} \rho^\phi_k    \bigg) \bigg)  \le  0 ;
\end{eqnarray} and 
b) the difference in the objective function becomes positive (by substituting ${\tilde \rho}_{i+1} = \rho^\phi_{i+1}$ in (\ref{Eqn_star})):
\begin{eqnarray*}
	f(\tilde{\phi}) - f(\phi)   
	&=& \h_{i - 1}^\phi \bigg( (\tilde{\rho_i} - \rho^\phi_i) - \frac{(\tilde{\rho_i} - \rho^\phi_i)(i - C)\big(1 + \sum_{j \geq i+2}^{\infty} \prod_{k = i+2}^j  \rho^\phi_k \big)}{(i+1-C) + \sum_{j \geq i+2} (j - C) \prod_{k = i+2}^{j} \rho^\phi_k}  \bigg) \\
	&=& \h_{i - 1}^\phi   (\tilde{\rho_i} - \rho^\phi_i)   \left (1 + \rho^\phi_{i+1} \big(1 + \sum_{j \geq i+2}^{\infty} \prod_{k = i+2}^j  \rho^\phi_k \big) \right ).
\end{eqnarray*}  Thus, we get an improved policy.
From (\ref{Eqn_Condition2}),  it further suffices to consider the case when,
\begin{eqnarray}
\label{Eqn_left}
\left ( \rho^\phi_{i+1 } + \frac{ (i - C)}{  \left((i+1-C) + \sum_{j \geq i+2} (j - C) \prod_{k = i+2}^{j} \rho^\phi_k \right)} \right )  < 0 
\mbox{ and }  \\ 
 \bigg( (i - C) +     \rho^\phi_{i+1} \bigg( (i+1-C) + \sum_{j \geq i+2} (j - C) \prod_{k = i+2}^{j} \rho^\phi_k    \bigg) \bigg) > 0. 
\end{eqnarray}
Note that the second term in the above is the numerator of the first term and hence for the conditions to hold the denominator of the first term should be negative, i.e., 
$$
\bigg( (i+1-C) + \sum_{j \geq i+2} (j - C) \prod_{k = i+2}^{j} \rho^\phi_k    \bigg)   < 0.
$$
But with $i < C$, both the components of the second term of  (\ref{Eqn_left}) is negative implying the second term should be negative. Thus such a case does not exist.  
Thus in all cases, we have   an alternate   policy which strictly improves $\phi$. 
\subsection*{Optimal Policy}
In view of the above result it is clear that the optimal policy is of the type (\ref{Eqn_pareto_optimal_policy}). 
Further it is easy to solve the given optimization problem for any $C \le \rmax/(1-\rmax)$  as given below:
Define $L^* = L^* (C)$ in the following manner
\begin{eqnarray}
L^* &=& \max \left  \{ i :    g( \phi^{sp,i} )  < 0 \right  \},  \mbox{ with } \phi_j^{sp, i}  = \rmax \indc_{j \le i} + \rmin \indc_{j  > i}  \nonumber \\
&=&  \max \left  \{ i :  \frac {  \sum_{j \le i}  \rmax^j j +    \rmax^i    \sum_{j > i} j \rmin^{j-i} }{
	\sum_{j \le i} \rmax^j  +    \rmax^i    \sum_{j > i}  \rmin^{j-i}
} < C \right  \}, \mbox{ and then $\rho^*$  satisfies with }  i = L^* \nonumber  \\
0 &=&  C -   \frac {  \sum_{j \le i}  \rmax^j j +     \rmax^i  \rho^* \left  (i+1 +   \sum_{j > i+1} j \rmin^{j-i-1}  \right ) }{
	\sum_{j \le i} \rmax^j  +  \rho^*\rmax^i   \left  (1+     \sum_{j > i+1}  \rmin^{j-i-1} \right ) } \mbox{ thus }  \nonumber \\
\rho^* &=&   \frac {  \sum_{j \le i}  \rmax^j j - C \sum_{j \le i} \rmax^j   }{
	C \rmax^i   \left  (1+     \sum_{j > i+1}  \rmin^{j-i-1} \right ) -  \rmax^i    \left  (i+1 +   \sum_{j > i+1} j \rmin^{j-i-1}  \right )  } \\
&=& \frac {  \sum_{j \le i}  \rmax^j ( j - C)    }{
	-  \rmax^i    \left  (i+1 - C +   \sum_{j > i+1} (j -C )\rmin^{j-i-1}  \right )  }. \nonumber
\end{eqnarray}
\eop

\begin{lemma} \label{Lemma_swap} \textbf{ Improvement  Through  Swapping }\\
For $ i < C$, if $\rho_i^\phi < \rho_{i+1}^\phi$ in a SM policy $\phi$ for the optimization problem defined in part (a) of Theorem \ref{Thm_Opt_Policy}, then we can get an policy, say $\phi^{swap}$, by swapping the $i$ and $i+1$-th components in $\phi$, such that:
\begin{eqnarray*}
f(\phi^{swap}) > f(\phi) \ \ \mbox{and} \ \ g(\phi^{swap}) < g(\phi)
\end{eqnarray*}

That means, we can improve the solution of the optimization problem by defining~ $\phi^{swap} = \left( \rho^{swap}_1, \rho^{swap}_2, \cdots, \rho^{swap}_i, \cdots \right)$ as follows:
\beq
\rho^{swap}_j = \rho^\phi_j \ \ \forall j \neq i, i+1 \ \ and \ \ \rho^{swap}_i = \rho^\phi_{i+1}, \ \  \rho^{swap}_{i+1} = \rho^\phi_i.  
\eeq
\end{lemma}

\begin{proof}
The  term $\h^{\phi^{swap}}_j$ corresponding to the policy $\phi^{swap}$, will be given as:
\begin{eqnarray*}
\h^{\phi^{swap}}_j = \left \{  \begin{array}{llll}
\Pi_{ 1 \le k \leq j} \ \rho^\phi_k  = \h_j^{\phi} \ \ \ \forall \ j \leq i-1 \\
\h_{i-1}^{\phi} \rho^\phi_{i+1}   \ \ \ \ \ \ \ \ \ \ j = i \\\
\h_j^{\phi}  \ \ \ \ \ \ \ \hspace{0.9cm} \forall \ j \geq i+1.
\end{array} \right . 
\end{eqnarray*}
First, consider the difference between the objective function corresponding to these two policies is given by:
\beq
f(\phi^{swap}) - f(\phi) = \sum_j (\h^{\phi^{swap}}_j - \h^\phi_j) = \h^\phi_{i-1} (\rho^\phi_{i+1} - \rho^\phi_{i}).
\eeq
Since, $\rho^\phi_i < \rho^\phi_{i+1}$ it is clear from above equation that, $f(\phi^{swap}) > f(\phi)$. Therefore, the objective function improves.

Next, consider the difference between the constraints,
\beq
g(\phi^{swap}) - g(\phi) = \sum_j (\h^{\phi^{swap}}_j - \h^\phi_j)(j-c) = \h_{i-1}^{\phi} (\rho^\phi_{i+1} - \rho^\phi_{i}) (i-C).
\eeq
From above equation, using $\rho^\phi_i < \rho^\phi_{i+1}$ and $i <C$, we conclude that $g(\phi^{swap}) - g(\phi) < 0 $. That means $g(\phi^{swap}) < g(\phi)$ and hence we get an improved solution.
\end{proof}

\begin{lemma} \label{Lemma_epsilon} \textbf{Improvement Through Adding $\epsilon$}\\
Consider a SM policy $\phi$ for the optimization problem  defined in part (a) of Theorem \ref{Thm_Opt_Policy}. For $i \geq C$, if there exists an $\rho^\phi_i < \rmax$, then we can get a policy, say $\phi^\epsilon$, by adding $\epsilon > 0$ to the $i$-th component of policy $\phi$, such that:
\beq
f(\phi^\epsilon) > f(\phi) \ \ \mbox{and} \ \ g(\phi^\epsilon) \leq 0.
\eeq
\end{lemma}
\begin{proof}
Since $C < C_{max} = \frac{\rmax}{1 - \rmax}$, there exists at least one $i$ such that $\rho^\phi_i < \rmax$. Let $\rho^\phi_i < \rmax$. Define a policy $\phi^\epsilon = \left( \rho^{\phi^\epsilon}_1, \rho^{\phi^\epsilon}_2, \cdots, \rho^{\phi^\epsilon}_i, \cdots \right) $ as follows:
\beq
\rho_j^{\phi^\epsilon} = \rho^\phi_j \ \  \forall j \neq i \ \ \mbox{and} \ \ \rho_i^{\phi^\epsilon} = \rho^\phi_i + \epsilon.
\eeq
Define, $\h_0^{\phi^\epsilon}  = 1,$ and for $i \geq 1$,  $\h_i^{\phi^\epsilon} = \prod_{1 \le j \le i}  \rho_j^{\phi^\epsilon}
$ .
First consider the difference in the objective function under the two policies: 
\begin{eqnarray}
   f(\phi^\epsilon) - f(\phi) &=&  \sum_j \h_j^{\phi^\epsilon} - \sum_j \h_j^{\phi} = \h_{i - 1}^{\phi}   \left ( ( \rho_i^{\phi^\epsilon} - \rho^\phi_i) + \sum_{j \geq i+1}^{\infty} \left  (\prod_{k = i}^j \rho_k^{\phi^\epsilon} - \prod_{k = i}^j  \rho^\phi_k \right ) \right )  \nonumber\\ \nonumber
&=& \h_{i - 1}^\phi (\rho_i^{\phi^\epsilon} - \rho^\phi_i) \bigg( 1  + \sum_{j \geq i+1}^{\infty} \prod_{k = i+1}^j  \rho^\phi_k  \bigg).
\end{eqnarray}

Objective function will improve, if the above difference between the objective function is positive. Clearly, the first and last term in the above equation is positive. So, to make the difference positive we require the middle term positive, i.e., $(\rho_i^{\phi^\epsilon} - \rho^\phi_i) > 0$. This is obviously true as $(\rho_i^{\phi^\epsilon} - \rho^\phi_i) = \epsilon$, which is positive.

Now we claim that the policy $\phi^\epsilon$ is in feasible region, as
\beq
g(\phi^\epsilon)  &=& \sum_j \h_j^{\phi^\epsilon} (j - C) \\
&=& g(\phi) + \h_{i-1}^{\phi}\epsilon \bigg( i-C + \sum_{j \geq i+1} (j-C)\ \Pi_{k = i+1}^j \rho^\phi_k \bigg).
\eeq
$g(\phi) \leq 0$ and we can choose $\epsilon >0$ small enough such that the overall term in the RHS of above equation will remain negative.

Thus, we get that the policy $\phi^\epsilon$ in the feasible region which also improve the objective function.
\end{proof}

\end{document}